%% file: main.tex
\documentclass[twoside,leqno]{article}
\usepackage{graphicx,fullpage} 
\usepackage{amsfonts,amsmath,amsthm,amssymb,mathtools, comment}\usepackage{xcolor}
\usepackage[T1]{fontenc}
\usepackage[colorlinks,citecolor=blue,linkcolor=blue,urlcolor=red,pagebackref]{hyperref}
\usepackage[capitalise]{cleveref}
\usepackage{url}
\usepackage{thmtools, thm-restate}
\usepackage{algorithm, algorithmicx, algpseudocode}
\usepackage{tikz}
\usepackage{bm}
\usepackage{thm-restate}

\newcommand{\eps}{\varepsilon}

\newcommand{\R}{\mathbb{R}}
\newcommand{\N}{\mathbb{N}}

\newcommand{\Z}{\mathbb{Z}}

\newcommand{\floor}[1]{\left\lfloor #1 \right\rfloor}
\newcommand{\ceil}[1]{\left\lceil #1 \right\rceil}

\newtheorem{theorem}{Theorem}[section]  
\newtheorem{lemma}[theorem]{Lemma}

\newtheorem{claim}[theorem]{Claim}
\newtheorem{cor}[theorem]{Corollary}
\theoremstyle{definition}
\newtheorem{definition}[theorem]{Definition}

\newtheorem{obs}[theorem]{Observation}

\title{Beyond 2-approximation for $k$-Center in Graphs}

\author{Ce Jin\\MIT\and Yael Kirkpatrick\\MIT\and Virginia Vassilevska Williams\\MIT \and Nicole Wein\\University of Michigan}
\date{}

\begin{document}

\setcounter{page}{0} \clearpage
\maketitle

\thispagestyle{empty}


\begin{abstract}
We consider the classical $k$-Center problem in undirected graphs. The problem is known to have a polynomial-time 2-approximation. 
There are even $(2+\varepsilon)$-approximation algorithms for every $\varepsilon>0$ running in near-linear time. The conventional wisdom is that the problem is closed, as $(2-\varepsilon)$-approximation is NP-hard when $k$ is part of the input, and for constant $k\geq 2$ it requires $n^{k-o(1)}$ time under the Strong Exponential Time Hypothesis (SETH).

Our first set of results show that one can beat the multiplicative factor of $2$ in undirected unweighted graphs if one is willing to allow additional small additive error, obtaining $(2-\varepsilon,O(1))$ approximations. We provide several algorithms 
that achieve such approximations for all integers $k$ with running time $O(n^{k-\delta})$ for $\delta>0$. For instance, for every $k\geq 2$, we obtain an $O(mn + n^{k/2+1})$ time  $\left(2 - \frac{1}{2k-1}, 1 - \frac{1}{2k-1}\right)$-approximation to $k$-Center, and
for every $k\geq 10$ we obtain an $(3/2,1/2)$-approximation algorithm running in $n^{k-1+1/(k+1)+o(1)}$ time. For $2$-Center we also obtain an $\tilde{O}(mn^{\omega/3})$ time $(5/3,2/3)$-approximation algorithm, where $\omega<2.372$ is the fast matrix multiplication exponent. Notably, the running time of this $2$-Center algorithm is faster than the time needed to compute APSP.

Our second set of results are strong fine-grained lower bounds for $k$-Center. We show that our $(3/2,O(1))$-approximation algorithm is optimal, under SETH, as any $(3/2-\varepsilon,O(1))$-approximation algorithm requires $n^{k-o(1)}$ time. We also give a time/approximation trade-off: under SETH, for any integer $t\geq 1$, $n^{k/t^2-1-o(1)}$ time is needed for any $(2-1/(2t-1),O(1))$-approximation algorithm for $k$-Center. This explains why our $(2-\varepsilon,O(1))$ approximation algorithms have $k$ appearing in the exponent of the running time.
Our reductions also imply that, assuming ETH, the approximation ratio 2 of the known near-linear time algorithms cannot be improved by any algorithm whose running time is a polynomial independent of $k$, even if one allows additive error.
\end{abstract}

	\newpage
    \maketitle

\newenvironment{proofof}[1]
  {\begin{proof}[#1]}
  {\end{proof}}

    \input{1_introduction}
    \section*{Acknowledgments} We acknowledge the Simons Institute Fall 2023 programs “Logic and Algorithms in Database Theory and AI” and “Data Structures and Optimization for Fast Algorithms” for the initiation of this work.

    \input{2_prelim}
    \input{4_lower_bounds}
    \input{3_algorithms}

    \bibliographystyle{alphaurl} 
    \bibliography{main}
\end{document}

%% file: 1_introduction.tex
\section{Introduction}
The $k$-Center problem is a classical facility location problem in the clustering literature.
Given a distance metric on $n$ points, $k$-Center asks for a set of $k$ of the points (``centers'') that minimize the maximum over all points $p$ of the distance from $p$ to its nearest center; this is called the {\em radius}. An important special case of the problem is when the metric is the shortest paths distance metric of a given $n$-node, $m$-edge graph. This special case has a huge variety of applications as seen for instance in the many examples in the 1983 survey \cite{kcentersurvey}. 
When $k$ is part of the input, the problem is well-known to be NP-hard as $k$-Dominating Set is just a special case: a set of $k$ nodes is a dominating set in an unweighted graph if and only if they are a $k$-center with radius $1$. 

As noted by \cite{HsuN79}, this same reduction shows that even for unweighted undirected graphs it is NP-hard to obtain $k$ centers whose radius is at most a factor of $(2-\eps)$ from the optimum for $\eps>0$, as such an algorithm would be able to distinguish between radius $1$ and greater than $1$, and would thus be able to solve Dominating Set. As $k$-Dominating set is $W[2]$-complete, obtaining such a $(2-\eps)$-approximation for $k$-Center is also $W[2]$-hard with parameter $k$, as noted by \cite{Feldmann19}. Via \cite{PatrascuW10}, the reduction also implies that under the Strong Exponential Time Hypothesis (SETH), a $(2-\eps)$-approximation for $k$-Center requires $n^{k-o(1)}$ time for all integers $k\geq 2$.
Further $(2-\eps)$-approximation hardness results are known for special classes of graphs and restricted metrics (e.g. \cite{Feldmann19,FederG88,KatsikarelisLP19}).

Several $2$ and $(2+\eps)$ (for all $\eps>0$) approximation algorithms for $k$-center have been developed over the years \cite{Gonzalez85,DyerF85,HochbaumS86,Thorupcenter04,AbboudCLM23}, even running in near-linear time in the graph size.

Due to the aforementioned hardness of approximation results, it seems that the $k$-center approximation problem is closed: factor $2$ is possible and factor $(2-\eps)$ for $\eps>0$ is  impossible under widely believed hypotheses.

However, what if we do not restrict ourselves to {\em multiplicative} approximations, and allow small {\em additive error} in addition? That is, we allow for the approximate radius $\tilde{R}$ to be between the true radius $R$ and $\alpha R+\beta$ for $\alpha<2$ and some very small $\beta$. 
Such {\em mixed}, $(\alpha,\beta)$ approximations are often studied in the shortest paths literature.

Hochbaum and Shmoys \cite{HochbaumS86} actually considered mixed approximations and showed that it is easy to modify the reduction from $k$-Dominating set to show that it is still NP-hard to obtain $(\alpha,\beta)$-approximations for any $\alpha<2$ and any $\beta$, for $k$-Center on {\em weighted} graphs: 
From a dominating set instance $G=(V,E)$, create a weighted graph $G'=(V,V\times V,w)$ where if $(u,v)\in E$, $w(u,v)=W$ and if $(u,v)\notin E$, $w(u,v)=2W$,\footnote{It is actually not necessary to add edges $(u,v)$ of weight $2W$ when $(u,v)\notin E$. It suffices to keep the same graph but add weights $W$ to all edges.}
where $W=2\beta/(2-\alpha)$; since $\alpha W+\beta<2W$, any $(\alpha,\beta)$-approximation to $k$-Center solves the $k$-Dominating set problem.

Besides NP-hardness, the modified reduction of \cite{HochbaumS86} also shows that $n^{k-o(1)}$ time is needed for any $k\geq 2$ for $(2-\eps,\beta)$-approximation to $k$-Center under SETH. However, the reduction crucially relies on the graph having weights. For $\eps=2-\alpha$, these weights behave as $\Theta(\beta/\eps)$ and are in fact quite large if $\eps$ is small. We thus ask:

\begin{center}
{\em Question 1: Are there $((2-\eps),\beta)$-approximation algorithms for $\beta=O(1)$ and $\eps>0$ that run in $O(n^{k-\delta})$ time for some $\delta>0$ in unweighted graphs?}
\end{center}

In particular, what is the best mixed approximation that one can get in $O(n^{k-\delta})$ time for $\delta>0$? 

Further, the Hochbaum and Shmoys \cite{HochbaumS86} reduction implies that for any fixed $\eps>0$, under SETH, any $O(n^{k-\delta})$ time $(2-\eps,\beta)$-approximation algorithm for $k$-center must have additive error $\beta$ at least a constant fraction of the maximum edge weight in the graph. Thus, it is also interesting whether one can obtain such $O(n^{k-\delta})$ time $(2-\eps,\beta)$-approximation algorithms for integer-weighted graphs, where $\beta$ is allowed to be proportional to the largest weight.

If the answer to question 1 is affirmative, then we ask: 

\begin{center}
{\em Question 2: Can one get $((2-\eps),\beta)$-approximation in near-linear time for $\eps>0$, $\beta= O(1)$ in unweighted undirected graphs?}
\end{center}
The question is also interesting for weighted graphs when $\beta$ is allowed to be proportional to the largest weight in the graph.

A last question is: 
\begin{center}
{\em Question 3: 
When can one avoid computing all pairwise distances (APSP) in the graph?}\end{center}
The best algorithms for APSP in $m$-edge, $n$-vertex graphs run in $\tilde{\Theta}(mn)$ time even when the graphs are unweighted and undirected, when $m\leq n^{\omega-1}$. This was (conditionally) explained by \cite{LincolnWW18}. 
The near-linear time $(2+\eps)$-approximation $k$-Center algorithms (e.g. \cite{Thorupcenter04,AbboudCLM23}) show that APSP computation is not needed for any $k$ if one is happy with a $(2+\eps)$-approximation.
For the case of $k=1$, $1$-Center is also known as the Radius problem in graphs. The hardness for $(2-\eps)$-approximation does not apply for $k=1$ (as $k$-dominating set hardness under SETH for instance only applies for $k\geq 2$). Consequently, there has been a lot of work on $O(mn^{1-\delta})$ time $(2-\eps)$-approximation algorithms for $1$-Center that avoid the computation of all-pairs shortest paths, and have truly subquadratic running times in sparse graphs \cite{AingworthCIM99,AbboudGW23,CairoGR16,BackursRSWW21,AbboudWW16,ChechikLRSTW14,RodittyW13}. 
For what other values of $k\geq 2$ can one obtain such algorithms that achieve $((2-\eps),\beta)$-approximation?

\subsection{Our results}
We present the first 
fast $(2-\eps,\beta)$-approximation algorithms for $\eps>0$ and $\beta= O(1)$ for $k$-Center in unweighted undirected graphs, answering Question 1 in the affirmative and addressing Question 3.
We complement our algorithms with a variety of fine-grained lower bounds, showing strong hardness results, in particular providing convincing evidence that the answer to question 2 is a resounding NO.

\paragraph{Lower bounds.}
Our conditional hardness results are largely based on the popular Strong Exponential Time Hypothesis (SETH) of \cite{ip2,cip10} that states that no $O((2-\eps)^n)$ time algorithm for constant $\eps>0$ can solve $k$-SAT on $n$ variables for arbitrary $k\geq 3$.
Due to a result by Williams \cite{Williams05}, SETH is known to imply strong hardness for the $k$-Orthogonal Vector ($k$-OV) problem of Fine-Grained Complexity (see also \cite{vsurvey}), and some of our lower bounds are from $k$-OV. Some of our lower bounds are from an approximation version of $k$-OV, Gap Set Cover, whose hardness due to \cite{SLM19,Lin19} is also based on SETH. One of our lower bounds is based on the Exponential Time Hypothesis (ETH) that states that $3$SAT on $n$ variables cannot be solved in $2^{o(n)}$ time. ETH is implied by SETH and is an even more plausible hardness hypothesis. 
All our conditional lower bounds assume the word-RAM model with $O(\log n)$ bit words.

We start with a simple conditional lower bound (see \cref{thm:simplelb}) from $k$-OV (and hence SETH) that implies that for every $k\geq 2$, $n^{k-o(1)}$ time is needed for any
$(3/2-\delta,\beta)$-approximation algorithm for $k$-Center for $\eps,\delta>0$ and any constant $\beta$.

This result implies that to get $O(n^{k-\delta})$ time, one can at best hope for a $(3/2,O(1))$ approximation. In fact, since our reduction constructs a sparse graph, we get that this is also true for sparse graphs, and a 
$(3/2-\delta,\beta)$-approximation in fact requires $m^{k-o(1)}$ time.

Next, we prove a more general result which is our main hardness result:

\begin{restatable}{theorem}{thmfulllb}
\label{thm:fulllb}
There is a function $f\colon \N^+ \to \N^+$ such that the following holds.
For all integers $t,\ell\geq 1$ and all $k\geq 2f(t)$, under SETH $n^{\frac{k}{f(t)} -1-o(1)}$ time is necessary to distinguish between radius $\leq (2t+1)\ell$ and radius $\ge (4t+1)\ell$ for $k$-center, even on graphs with $m=O(n)$ edges.  Here, $f(t)=t$ for $t\in \{1,2,3,4\}$, and $f(t) \le t^2$ for all $t\in \N^+$. 
\end{restatable}

The proof of our theorem, to our knowledge, is the first use of Gap Set Cover for approximation hardness results for graph problems.
Some consequences of our result are as follows:
\begin{enumerate}
\item Under SETH, there is no $(2-\frac{1}{2t+1}-\eps,\beta)$-approximation algorithm for $k$-center for any $\eps>0$, $\beta= O(1)$, running in $O(n^{k/t^2-1-\delta})$ time.
E.g., a $(5/3-\eps,\beta)$ approximation requires $O(n^{k-1-o(1)})$ time.

\item Because of the ETH-based hardness of Gap Set Cover \cite{SLM19,Lin19}, we also get that under ETH there can be no $f(k)n^{o(k)}$ time $(2-\eps,\beta)$-approximation algorithm for $k$-Center for $\eps>0,\beta= O(1)$.  
\end{enumerate}

Recall that there are $\tilde{O}(mk)$ time\footnote{The notation $\tilde{O}$ hides polylogarithmic factors.} $2$-approximation algorithms \cite{Gonzalez85,DyerF85,HochbaumS86}, and $\tilde{O}(m)$ time $(2+\eps)$-approximation algorithms \cite{Thorupcenter04,AbboudCLM23}, for $k$-center. Point 2 above shows that these are {\em optimal} in a strong sense: even if we allow arbitrary polynomial time $O(n^c)$ and allow additional additive error $\beta$, one cannot beat the multiplicative factor of $2$ for all $k$ simultaneously. We thus answer Question 2 strongly in the negative.

\paragraph{Algorithms.}
We present the first ever improvements over the known multiplicative factor of $2$ for $k$-center, with very small additive error, answering Question 1 in the affirmative. Our algorithmic results are summarized in Table \ref{table:results}.

\input{3a_algorithms_table}

Our first algorithmic result is that for every $k\geq 3$, there is an $O(n^{k-\delta})$ time algorithm for $\delta>0$ that achieves a $(3/2,1/2)$-approximation. The algorithm utilizes fast matrix multiplication, and the bound is in terms of $\omega<2.372$, the exponent of square matrix multiplication \cite{VXXZ24}. We complement the algorithm with a conditional lower bound showing that the approximation ratio of the algorithm is tight.



\begin{theorem}\label{thm:32}
    Given an unweighted, undirected graph $G$, there is a randomized algorithm that computes a $(3/2, 1/2)$-approximation to the $k$-center w.h.p. and runs in time \begin{itemize}
        \item $\tilde{O}(n^{\omega+1/(\omega+1)})$ time for $k=3$,
        \item $\tilde{O}(n^{k-(3-\omega) + 1 / (k+1)})$ for $k\geq 4$, and 
        \item $n^{k-1 + 1 / (k+1)+o(1)}$ time for $k\geq 10$.
    \end{itemize}

    
    \end{theorem} 
  This result appears in the paper as \cref{thm:32approxk}.

Due to Theorem \ref{thm:simplelb}, under SETH, the multiplicative part of the approximation guarantee is {\bf tight} for sub-$n^k$ time algorithms, as any $(3/2-\eps,O(1))$-approximation algorithm requires $n^{k-o(1)}$ time under SETH.

    


We then turn to address Question 3, and in particular the follow-up question asking for what values of $k\geq 2$ one can achieve $(2-\eps,O(1))$-approximation algorithms that run faster than computing APSP. In particular, for sparse graphs (when $m\leq O(n)$), when can such algorithms run in $O(n^{2-\delta})$ time for $\delta>0$?

The algorithms of \cref{thm:32} don't address this question. In particular, they explicitly compute APSP in unweighted graphs using Seidel's $O(n^\omega)$ time algorithm. One could improve the running time by instead using an approximate APSP algorithm such as the additive approximations of \cite{DorHZ00,saha} or the mixed approximations of \cite{elkin}, with a slight cost to the approximation. However, even then, since all $n^2$ distances are computed, the algorithm would always run in $\Omega(n^2)$ time.

We show that for $2$-Center, one can in fact obtain an algorithm that is polynomially faster than $mn$ and hence polynomially faster than $n^2$ in sparse graphs. The following result appears in the paper as \cref{thm:53approx2center}.

\begin{theorem}\label{thm:2cen}
    There exists a randomized algorithm running in $\tilde{O}(mn^{\omega/3})$ that computes a $(5/3,2/3)$-approximation to the $2$-center of any undirected, unweighted graph, w.h.p. If the $2$-center radius is divisible by $3$, the algorithm gives a true multiplicative $5/3$-approximation.
\end{theorem}

Thus, just like with $1$-center \cite{AingworthCIM99,AbboudGW23,CairoGR16,BackursRSWW21,AbboudWW16,ChechikLRSTW14,RodittyW13}, one can get a better-than-2 approximation algorithm for $2$-center, faster than $n^2$ time in sparse graphs. We also note that the algorithm can be adapted to give a $(5/3,M)$-approximation for the $2$-center of a graph with positive integer weights bounded by $M$. As noted earlier, due to the reduction of Hochbaum and Shmoys \cite{HochbaumS86}, $\Omega(M)$ additive error is necessary for any mixed approximation algorithm with multiplicative stretch $2-\eps$.  

We then extend the techniques used to construct the algorithm for $2$-center to obtain a general approximation scheme that works for any $k$-center. The following result appears in the paper as \cref{thm:2kapproxkcenter} and \cref{thm:2kapproxmatmult}. The value $\alpha$ below is the largest real number so that an $n\times n^\alpha$ matrix can be multiplied by an $n^\alpha \times n$ matrix in $O(n^{2+\eps})$ time for all $\eps>0$.

\begin{theorem} \label{thm:approxscheme}
For any $k\geq 2$, there is a randomized combinatorial algorithm that in $\tilde{O}(mn + n^{k/2 + 1})$  time, computes w.h.p. a $\left(2 - \frac{1}{2k-1}, 1 - \frac{1}{2k-1}\right)$-approximation to $k$-center for any given unweighted, undirected graph.

With the use of fast matrix multiplication, the algorithm's running time can be sped up. In particular, if $\omega = 2$ the running time becomes $\tilde{O}(n^{k/2+22/9(k+1)})$ for $k\geq 3$.  With the current best bounds on $\omega$ and $\alpha$, the algorithm runs in time
\begin{itemize}
    \item $\tilde{O}(mn^{\omega/3})$ for $k=2$, 
    \item $\tilde{O}(n^{\frac{k}{2} + \frac{\beta_0}{k+1} + (\omega - 2)})$ for $3\leq k \leq 13$, where $\beta_0 \coloneqq \frac{-8\omega^2 + 18\omega + 18}{3\omega + 3}\approx 1.5516$, and
    \item $\tilde{O}(n^{\frac{k}{2} + \frac{\beta_1}{k+1}+o(1)})$ for $k\geq 13$, where $\beta_1 \coloneqq  \frac{34\omega^2 - 24 \omega - 66}{3\omega + 3}\approx 6.7533$.
\end{itemize}
\end{theorem}

We thus get that for every constant $k$, there is a $(2-\eps,O(1))$-approximation algorithm running in almost $n^{k/2}$ time. Compare this with our conditional lower bound \cref{thm:betterlbwarmup2}
that says that  $n^{\floor{k/2}-1-o(1)}$ time is needed for any $(9/5-\eps,O(1))$-approximation.

As the approximation quality depends on $k$, we also present a version of our approximation scheme that trades-off running time for approximation quality.
\begin{theorem}\label{thm:ell}
    For any integer $1\leq \ell \leq k$, there is a randomized combinatorial algorithm running in $\tilde{O}(mn + n^{1 + k - \ell + \frac{\ell(\ell + 1)}{2(k + 1)}})$ time that computes a $\left( 2- \frac{1}{2\ell}, 1 - \frac{1}{2\ell}\right)$-approximation to the $k$-center of $G$ w.h.p.
\end{theorem}

This result appears in the text as \cref{thm:2lapproxkcenter}. The running time of the algorithm can also be improved using fast matrix multiplication. 

We note that our hardness result in \cref{thm:fulllb} explains why $k$ appears in the exponent of our running times: under SETH $n^{\Omega(k)}$ time is needed for any $(2-\eps,O(1))$ approximation for constant $\eps>0$ independent of $k$.




Finally, we focus on the special case $k=3$ for which we have stated two approximation algorithms so far. First, the $(3/2,1/2)$-approximation algorithm of \cref{thm:32} runs in $O(n^{\omega+1/\omega})$ time. 
Second, our approximation scheme from \cref{thm:approxscheme} gives a $(9/5,4/5)$-approximation with running time 
$O(n^{13/6})$ if $\omega=2$. The latter is faster than the former (which would run in $O(n^{2.5})$ time if $\omega=2$) but achieves a worse approximation ratio.
As a minor final result, we show that one can
obtain a mixed approximation with multiplicative factor between $3/2$ and $9/5$ (namely, $7/4$) that runs faster than the $O(n^{\omega+1/\omega})$ time of the $3/2$-approximation, and for a nontrivial set of graph sparsities $m$, runs faster than $mn$ time even in weighted graphs. This result appears in the text as \cref{thm:74approx3center}.










\paragraph{Related work.} 
While we cannot hope to exhaustively list all the extensive work on $k$-center, here is some more.
Approximation algorithms for the $1$-center problem and the related diameter problem in graphs have been extensively studied \cite{AingworthCIM99,AbboudGW23,CairoGR16,BackursRSWW21,AbboudWW16,ChechikLRSTW14,RodittyW13}. 
Dynamic $(2+\eps)$-approximation algorithms for $k$-Center is a recent topic of interest \cite{cruciani2024dynamic}. 
$k$-Center is studied in restricted classes of graphs and metrics both for static and for dynamic algorithms (e.g. \cite{planar1,planar2,feldmann2019fixed,feldmann2020parameterized,goranci2021fully,gan2024fully} and many more).
The asymmetric version of $k$-center (e.g. for directed graphs) is a harder problem, although $2$-approximation algorithms are possible for structured metrics (see  \cite{BalcanHW20} and the references therein).

%% file: 3a_algorithms_table.tex
\begin{table}[]
\centering
\begin{tabular}{|c|l|l|l|l|}
\hline
$\bm{k}$ & \textbf{Approximation}                     & \textbf{Runtime}                        & \textbf{Comments} & \textbf{Reference} \\ \hline
2            & $(5/3, 2/3)$                               & $mn^{\omega/3}$                         &                   & \cref{thm:53approx2center} \\ \hline
3            & $(7/4, M)$                               & $ n^{\frac{(5\omega-\omega^2+1)-\mu\omega(\omega-2)}{(3\omega-\omega^2+1)}}$             & Integer edge weights $\leq M$. & \cref{thm:74approx3center}                    \\ \hline

$3$          & $(3/2, 1/2)$                               & $n^{\omega + \frac{1}{\omega}}$    &         & \cref{thm:32approxk}  \\ \hline
$k\geq 4$          & $(3/2, 1/2)$                               & $n^{k - (3-\omega) + \frac{1}{k+1}}$    &           & \cref{thm:32approxk}  \\ \hline
$k \geq 10$          & $(3/2, 1/2)$                               & $n^{k - 1 + \frac{1}{k+1} + o(1)}$    &           & \cref{thm:32approxk}  \\ \hline
$k$          & $(2 - \frac{1}{2k-1}, 1 - \frac{1}{2k-1})$ & $mn + n^{k/2 + 1}$                      & Combinatorial.     & \cref{thm:2kapproxkcenter}                   \\
\hline
$k\geq 3$          & $(2 - \frac{1}{2k-1}, 1 - \frac{1}{2k-1})$ & $n^{k/2 + 22/9(k+1)}$                     & Assuming $\omega=2$.     & \cref{thm:2kapproxmatmult}                   \\
\hline
$k\geq 3$          & $(2 - \frac{1}{2k-1}, 1 - \frac{1}{2k-1})$ & $n^{k/2 + \beta_0/(k+1) + (\omega - 2)}$                      & $\beta_0 = \frac{-8\omega^2 + 18\omega + 18}{3\omega + 3}\approx 1.5516.$     & \cref{thm:2kapproxmatmultgen}                   \\
\hline
$k\geq 12$          & $(2 - \frac{1}{2k-1}, 1 - \frac{1}{2k-1})$ & $n^{k/2 + \beta_1/(k+1) + o(1)}$                      &  $\beta_1 = \frac{34\omega^2 - 24\omega -66}{3\omega + 3}\approx 6.7533.$    & \cref{thm:2kapproxmatmultgen}                   \\

\hline
$k$          & $(2 - \frac{1}{2\ell}, 1 - \frac{1}{2\ell})$                 & $mn + n^{k-\ell +  \frac{\ell(\ell + 1)}{2(k+1)}+1} $  & Combinatorial, $\ell \leq k$.     &   \cref{thm:2lapproxkcenter}                 \\ \hline
\end{tabular}
\caption{Algorithmic Results.\\ We use ``combinatorial algorithms'' to refer to algorithms that do not use Fast Matrix Multiplication. }
\label{table:results}
\end{table}

%% file: 2_prelim.tex
\section{Preliminaries}




Let $G = (V,E)$ be a weighted or unweighted graph. Throughout this paper all graphs will be undirected. Denote by $n$ the number of vertices in the graph $|V|$ and by $m$ the number of edges $|E|$. The distance $d(u,v)$ between two vertices $u,v\in V$ is the length of the shortest path in $G$ between $u$ and $v$. 

The eccentricity of a vertex $v$ is defined as $\max_{u\in V}d(v,u)$. The vertex with smallest eccentricity is called the \textit{center} of $G$ and its eccentricity is the radius of $G$, $\min_{v\in V}\max_{u\in V} d(v,u)$. The $k$-center problem generalizes this definition to sets of $k$ points. Define the $k$-radius of $G$, $R_k(G)$, to be 
\[
R_k(G) = \min_{\substack{C\subseteq V \\  |C|=k}} \max_{v\in V} d(v,C).
\]

The $k$-center of $G$ is defined as the set of points that achieve the $k$-radius, $$\arg \min_{\substack{C\subseteq V\\  |C|=k}} \max_{v\in V} d(v,C).$$ 

When $k$ is clear from context we refer to the $k$-radius as simply the radius. 

Given a point $v\in V$ and $r > 0$ define the ball of radius $r$ around $v$ as $B(v, r)\coloneqq \{u\in V : d(u,v) \leq r\}$. Similarly, for a set $S\subset V$, define the ball of radius $r$ around the set $S$ as $B(S, r) \coloneqq \{u\in V : d(u,S) \leq r\}$. We often want to consider points that are far away from a point or set. For this we consider the complement of the ball around a point or set, $B(S, r)^c \coloneqq \{u\in V : d(u, S) > r\}$. For ease of notation, we define $B(\emptyset, r) = \emptyset$ for any $r\in \R$.

The following values are defined in the arithmetic circuit model.
The exponent $\omega$ is the smallest real number such that $n\times n$ matrices can be multiplied in $O(n^{\omega+\eps})$ time for all $\eps>0$. The exponent $\omega(p,q,r)$ is the smallest real number such that one can multiply an $n^p\times n^q$ matrix by an $n^q\times n^r$ matrix in $O(n^{\omega(p,q,r)+\eps})$ time for all $\eps>0$. It is known that $\omega(p,q,r)$ is invariant under any permutation of $p,q,r$. We also use the notation $\textrm{MM}(a,b,c)$ to denote the best known running time to multiply an $a\times b$ matrix by a $b\times c$ matrix, in particular $\textrm{MM}(n^p,n^q,n^r) = n^{\omega(p,q,r)}$.
The value $\alpha$ is the largest value in $[0,1]$ such that $\omega(1,\alpha,1)=2$. 

%% file: 4_lower_bounds.tex
 \section{Conditional Lower Bounds}
In this section we will prove a simple lower bound for $(3/2-\eps,\beta)$-approximating $k$-center (\cref{thm:simplelb}), and our main lower bound result for $(2-\eps,\beta)$-approximation (\cref{thm:fulllb}), restated below with parameter $k$ rescaled for convenience.
\begin{theorem}[Equivalent form of \cref{thm:fulllb}]
\label{thm:fulllbrestated}
There is a function $f\colon \N^+ \to \N^+$ such that the following holds.
Let  $k\ge 2$, $t,\ell\geq 1$ be constant integers. Assuming SETH,
   distinguishing between radius $\leq (2t+1)\ell$ and radius $\ge (4t+1)\ell$ for $(f(t)\cdot (k+1))$-center cannot be done in $O(m^{k-\delta})$ time, for any constant $\delta>0$.
  Here, $f(t)=t$ for $t\in \{1,2,3,4\}$, and $f(t) \le t^2$ for all $t\in \N^+$. 
\end{theorem}

Since our proof of the main result is quite involved, we will present the proof in an incremental fashion by first proving the simple lower bound (\cref{thm:simplelb}), and then two special cases (for $t=1,2$) of \cref{thm:fulllbrestated}, before finally showing the full recursive construction for \cref{thm:fulllbrestated}.

All our lower bound results hold for undirected unweighted graphs with $n$ nodes and $m$ edges (assuming $m\ge n-1$).

\subsection{Known hardness results for Set Cover}
Our lower bounds rely on several hardness results for the Set Cover problem in the literature.
A \emph{Set Cover instance} is a bipartite graph $G=(A,B,E)$ on $n=|A|+|B|$ nodes, where we want to find the smallest subset $S\subseteq A$, such that every $b\in B$ is adjacent to some node in $S$. 
Without loss of generality, we assume each $b\in B$ is adjacent to at least one node $A$.

Our simple lower bound (proved in \cref{sec:simplelb}) is based on the following SETH-based hardness result for deciding whether a Set Cover instance has a size-$k$ solution (or equivalently, solving the the $k$-Orthogonal-Vectors problem\footnote{A $k$-Orthogonal-Vectors ($k$-OV) instance is a set $A \subseteq \{0, 1\}^d$ of $n$ binary vectors of dimension $d$, and the $k$-OV problem asks if we can find $k$ vectors $a_1, \dots , a_k \in A$ such that they are orthogonal, i.e., $(a_1)[i]\cdot (a_2)[i]\cdot\cdots\cdot (a_k)[i] = 0$ for all $i\in [d]$.
Solving the $k$-OV instance $A \subseteq \{0, 1\}^d$ is equivalent to deciding whether the 
 Set Cover instance $G=(A,[d],E)$ defined by letting $(a,i)\in E$ iff $a[i]=0$ has a size-$k$ solution.}).
\begin{theorem}[\cite{Williams05,PatrascuW10}]
   Assuming SETH, for every integer $k\ge 2$ and $\delta>0$, there is no $O(n^{k-\delta})$-time algorithm that can decide whether an $n$-node Set Cover instance $(A,B,E)$ has a size-$k$ solution, even when $|B| = O_{\delta}(k\log n)$.
   \label{thm:sethkov}
\end{theorem}

Starting from \cref{sec:betterlbwarmup}, our lower bounds will crucially rely on the inapproximability of Set Cover from the parameterized complexity literature \cite{SLM19,Lin19}.\footnote{These papers sometimes stated their results for the Dominating Set problem, which is essentially the same as Set Cover.} These papers proved an inapproximability factor of $(\log n)^{\Omega(1)}$, while in our applications we only need (arbitrarily large) constant inapproximability factor $C$.

\begin{lemma}[{SETH-hardness of Gap Set Cover, implied by \cite[Theorem 1.5]{SLM19}}]
Assuming SETH, for every integer $k \ge 2$ and for every $\delta > 0, C\ge 1$, no $O(n^{k-\delta})$-time algorithm can distinguish whether an  $n$-node Set Cover instance $G=(A,B,E)$ has a solution of size $k$ or has no solutions of size $\le Ck$.
\label{templem:approxsetcoverlb}
\end{lemma}

For our purpose, we need the hardness to hold even for Set Cover instances $(A,B,E)$ with small $|B|$. This can be obtained from \cref{templem:approxsetcoverlb} by a simple powering argument, as shown in the following corollary.
\begin{cor}[Small-$B$ version of \cref{templem:approxsetcoverlb}]
Assuming SETH, for every integer $k \ge 2$ and for every $\delta,\gamma > 0, C\ge 1$, no $O(n^{k-\delta})$-time algorithm can distinguish whether an  $n$-node Set Cover instance $G'=(A',B,E')$ with $|B|\le n^{\gamma}$ has a solution of size $k$ or has no solutions of size $\le Ck$.
\label{cor:approxsetcoverlb}
\end{cor}
\begin{proof}
   Let $g = \lceil 1/\gamma\rceil$.  Given a Set Cover instance $G=(A,B,E)$, we can create a larger Set Cover instance $G' = (A^g,B,E')$, in which a $g$-tuple $(u_1,\dots,u_g)\in A^g$ is adjacent to $b\in B$ in $E'$ iff there exists $u_i$ such that $(u_i,b)\in E$. Observe that for any nonnegative integer $k'$, $G$ has a size-$k'g$ solution if and only if $G'$ has a size-$k'$ solution.

   By \cref{templem:approxsetcoverlb}, it requires $(|A|+|B|)^{kg-o(1)}$ time to decide whether $G$ has a size-$kg$ solution or has no solutions of size $\le Ckg$ under SETH. Hence, the same time is required for deciding whether $G'$ has a size-$k$ solution or has no solutions of size $\le Ck$ under SETH. 
   Note $G'=(A^g,B,E')$ has $|A|^g + |B|$ nodes, and we can pad dummy nodes into the $A^g$ side so that $G'$ now has $n = (|A|+|B|)^g$ nodes, and the number of nodes in $B$ is only $|B|\le n^{1/g} \le n^\gamma$.  The time lower bound then becomes $(|A|+|B|)^{kg-o(1)} = \Omega(n^{k-o(1)})$.
\end{proof}

We will also use the lower bound for Gap Set Cover under Exponential Time Hypothesis.
\begin{lemma}[{ETH-hardness of Gap Set Cover, implied by \cite[Theorem 1.4]{SLM19}}]
Assuming ETH, for any constant $C\ge 1$,
no $O(f(k) n^{o(k)})$-time algorithm can distinguish whether an  $n$-node Set Cover instance $G=(A,B,E)$ has a solution of size $k$ or has no solutions of size $\le Ck$.
\label{lem:approxsetcoverlbETH}
\end{lemma}

\subsection{A simple lower bound for \texorpdfstring{$(\frac{3}{2}-\eps,\beta)$}{(3/2-eps,beta)}-approximation}
\label{sec:simplelb}

Our simple lower bound is stated as follows.

\begin{restatable}{theorem}{thmsimplelb}
\label{thm:simplelb}
   Let $k\ge 2$, $\ell \ge 1$ be constant integers. Assuming SETH,
   distinguishing between radius $\leq 2\ell$ and radius $\geq 3\ell$ for $k$-center cannot be done in $O(m^{k-\delta})$ time, for any constant $\delta>0$.

   As a corollary, for any constants $\eps\in (0,1), \beta\ge 0$, $(\frac{3}{2}-\eps,\beta)$-approximating $k$-center radius requires $m^{k-o(1)}$ time under SETH.
\end{restatable}

To show the corollary, we set $\ell = \lfloor \frac{\beta}{2\eps}\rfloor + 1 $ which satisfies $(\frac{3}{2}-\eps)\cdot 2\ell + \beta < 3\ell$, so a $(\frac{3}{2}-\eps,\beta)$-approximation algorithm can distinguish between radius $\le 2\ell$ and $\ge 3\ell$ and hence requires $m^{k-o(1)}$ time. 

We prove \cref{thm:simplelb} in the rest of this section.

 Suppose we are given a Set Cover instance $G=(A,B,E)$ where $|A|=\Theta(n)$ and $|B| = n^{o(1)}$ and want to decide whether it has a size-$k$ solution. By \cref{thm:sethkov}, this requires $n^{k-o(1)}$ time under SETH.

\newcommand{\Gad}{\mathrm{Gad}}
Based on the Set Cover instance, we define a base gadget graph (which will be repeatedly used in this section and later sections) as follows.
\begin{definition}[Base gadget graph $\Gad(\hat A,\hat B,\hat c,L)$]
   Let bipartite graph $G=(A,B,E)$ be the given Set Cover instance,  and let $L\ge 1$ be an integer parameter.
  Then,  the base gadget graph is defined by the following procedure:
  \begin{itemize}
    \item Create node sets $\hat A,\hat B$ which are copies of the node sets $A,B$. By convention, let $\hat a\in \hat A$ denote the copy of $a\in A$ (and similarly for $\hat b\in \hat B$ and $b\in B$).  
    \item  For every $(a,b)\in E$, connect $\hat a\in \hat A$ and $\hat b\in \hat B$ by an $L$-edge path $(\hat a, v_{\hat a,\hat b,1},  v_{\hat a,\hat b,2},\dots, v_{\hat a,\hat b,L-1}, \hat b)$.\footnote{When we connect two nodes by a path, we mean we add new internal nodes and edges into the graph to form this path.}
        \item Add a new node $\hat c$.
            \item For every $\hat a\in \hat A$, connect $\hat a$ and $\hat c$ by an $L$-edge path $(\hat a,w_{\hat a,1},w_{\hat a,2},\dots,w_{\hat a,L-1},\hat c)$.
  \end{itemize}
  We denote this base gadget graph by $\Gad(\hat A,\hat B,\hat c,L)$. (If we have a base gadget graph named $\Gad( A', B', c',L)$, then we will analogously use the convention that $a'\in A'$ denotes the copy of $a\in A$, and similarly for $b'\in B'$ and $b\in B$.)
  \label{defn:gadget}
\end{definition}
Note that a Base gadget graph with parameter $L$  has $L|E|+L|A| \le O(L|A||B|)\le Ln^{1+o(1)}$ edges.

Using the base gadget graph in \cref{defn:gadget},   we now create a $k$-center instance $G'=(V',E')$ as follows (see \cref{fig0:simplelb}):
 \begin{itemize}
     \item Add a base gadget graph $\Gad(A',B',c',\ell)$.
     \item For every $b'\in B'$, attach an $\ell$-edge path $(b',u_{b',1},u_{b',2},\dots,u_{b',\ell})$ to $b'$.
 \end{itemize}
 Observe that the new graph $G'$ has $O(\ell \cdot n^{1+o(1)})$  edges.

\begin{figure}[h]
    \centering
 \input{fig0replace}
    \caption{The $k$-center instance $G'=(V',E')$ constructed in the proof of \cref{thm:simplelb} (for distinguishing radius $\le 2\ell$ or $\ge 3\ell$). In this example, $\ell=3$. The dashed green box contains the base gadget graph $\Gad(A',B',c',\ell)$. 
    Given a Set Cover solution of size $k$, picking their copies in $A'$ gives a $k$-center solution of radius $2\ell$.}
    \label{fig0:simplelb}
\end{figure}

First we show a good $k$-center solution exists in the YES case.
\begin{lemma}
    If the original Set Cover instance $G=(A,B,E)$ has a solution $S\subseteq A$ of size $|S|= k$, then the $k$-center instance $G'$ has a solution with radius $2\ell$.
    \label{lemma:simpleyes}
\end{lemma}
\begin{proof}
Given the Set Cover solution $S=\{s_1,\dots,s_k\}\subset A$, in $G'=(V',E')$ we simply pick the copies of the same $k$ nodes, $s_1',\dots,s_k'\in A' \subset V'$ as the centers, and we now verify that they cover every node in $V'$ within distance $2\ell$:
\begin{itemize}
    \item Every node $a'\in A' \subseteq V'$ can be reached from the center $s_1'\in A'$ via a length-$2\ell$ path through node $c'$, namely $(s_1', w_{s_1',1}, \cdots , w_{s_1',\ell-1} , c' , w_{a',\ell-1}, \cdots , w_{a',1}, a')$. 
     This also means all intermediate nodes $w_{\cdot,\cdot}$ and the node $c'$ can be reached from the center $s_1'$ within distance less than $2\ell$.
    \item For every $b\in B$, the Set Cover solution $S$ guarantees that there exists an $s_i$ such that $(s_i,b)\in E$. 
    Then, by definition of the base gadget graph, $b'\in B'\subset V'$ can be reached from the center $s'_i$ via a length-$\ell$ path $(s_i',v_{s_i',b',1},\dots,v_{s_i',b',\ell-1},b')$. 

    Then, since all nodes $u_{b',\cdot}$ on the path attached to $b'$, as well as all nodes $v_{\cdot,b',\cdot}$ on the paths between $b'$ and $A'$ in the base gadget graph, are reachable from $b'$ within distance $\le \ell$, they are hence reachable from $s_i'$ within distance $\le \ell+\ell=2\ell$.
\end{itemize}
We have verified all nodes in $V'$ are covered by some center $s_i'$ within distance $2\ell$, finishing the proof.
\end{proof}

Next, we show no good $k$-center solution can exist in the NO case. To show this, we prove the contrapositive.
\begin{lemma}
If the $k$-center instance $G'$ has a solution with radius $< 3\ell$, then the original Set Cover instance $G=(A,B,E)$ has a solution of size $k$.
\label{lemma:simplelbNO}
\end{lemma}
\begin{proof}
    Let $\{\tilde s_1,\dots,\tilde s_k\}\subset V'$ denote the $k$-center solution of $G'=(V',E')$ of radius $< 3\ell$.  For each $i\in [k]$, define $s_i'$ to be the node in $A'\subset V'$ that is the closest to $\tilde s_i$ on graph $G'$ (breaking ties arbitrarily).  In other words,
    \begin{itemize}
        \item If $\tilde s_i   \in \{u_{b',j} \}_{j\in [\ell]} \cup \{b'\}$ for some $b'\in B'$ (i.e., $\tilde s_i$ is on the $\ell$-edge path attached to $b'$), then $s_i'$ is the copy of some $s_i\in A$ such that $(s_i,b)\in E$ (which exists by our initial assumption that $b$ has at least one neighbor in the Set Cover instance).
            \item If $ \tilde s_i  \in \{v_{a',b',j}\}_{b'\in B', j\in [\ell-1]} \cup \{a'\} \cup \{w_{a',j}\}_{j\in [\ell-1]} $ for some $a'\in A'$ (i.e., $\tilde s_i$ is in the base gadget graph excluding $B'$ and $\{c'\}$), then $s_i' = a'$.
                \item If $\tilde s_i = c'$ then $s_i'\in A'$ is arbitrary (this case will not happen in our proof).
    \end{itemize}
    We now show that $\{s_1,\dots,s_k\}\subset A$, namely the nodes corresponding to $s_1',\dots,s_k'\in A'$, is a Set Cover solution.

Fix any $b\in B$. Consider $u_{b',\ell}$, the last node on the $\ell$-edge path attached to $b'$, and suppose it is reachable from the center $\tilde s_i$ within distance $< 3\ell$. 
By inspecting the construction of $G'$, observe that within $< 3\ell$ distance $u_{b',\ell}$ cannot reach node $c'$ or any other $b''\in B'\setminus \{b'\}$. Hence, such center $\tilde s_i$ within $<3\ell$ distance from $u_{b',\ell}$ can only be one of the following two cases: 
\begin{itemize}
    \item Case 1: $\tilde s_i  \in \{u_{b',j} \}_{j\in [\ell]} \cup \{b'\}$.

        By earlier discussion, we have $(s_i,b)\in E$ in this case. 
    
    \item Case 2: $ \tilde s_i  \in \{v_{a',b'',j}\}_{b''\in B', j\in [\ell-1]} \cup \{a'\} \cup \{w_{a',j}\}_{ j\in [\ell-1]} $ for some $a'\in A'$.
        
        By earlier discussion, we have $s_i = a$ in this case. 
    
   By inspecting the construction of $G'$, we can observe that in this case we must have $(a,b)\in E$ (and hence $(s_i,b)\in E$). Intuitively speaking, the shortest path from $u_{b',\ell}$ to the center $\tilde s_i$ can exit $b'$ only through some $v_{a',b',\ell-1}$ for which $(a,b)\in E$, and from there it cannot reach $c'$ or any other $b''\in B'\setminus \{b'\}$, so it always remains the closet to the same $a'\in A'$.
\end{itemize}
This finishes the proof that every $b\in B$ is covered by some $s_i$, so $\{s_1,\dots,s_k\}$ is a Set Cover solution.
\end{proof}

\begin{proofof}{Proof of \cref{thm:simplelb}}
 By combining \cref{lemma:simpleyes} and \cref{lemma:simplelbNO}, we see that any algorithm on graphs of  $m \le n^{1+o(1)}$ edges that distinguishes between $k$-center radius $\le 2\ell$ and $\ge 3\ell$ can be used to decide whether the original Set Cover instance has a size-$k$ solution, which requires $n^{k-o(1)}\ge m^{k-o(1)}$ time under SETH. 
\end{proofof}

\subsection{Warm-up I: Lower bounds via Gap-Set-Cover}
\label{sec:betterlbwarmup}

In this section, we prove the following special case of our full lower bound result \cref{thm:fulllbrestated}, which achieves better inapproximability ratio $(5/3-\eps)$ than the previous $(3/2-\eps)$, but has a lower time lower bound $m^{k- 1 - o(1)}$.
We prove this special case first in order to clearly illustrate how the hardness of Gap Set Cover is helpful, which is one of  the main ideas behind our full result.

\begin{theorem}[special case of \cref{thm:fulllbrestated}]
\label{thm:betterlbwarmup1}
   Let $k\ge 2$, $\ell \ge 1$ be constant integers. Assuming SETH,
   distinguishing between radius $\leq 3\ell$ and radius $\geq 5\ell$ for $(k+1)$-center cannot be done in $O(m^{k-\delta})$ time, for any constant $\delta>0$.

   As a corollary, for any constants $\eps\in (0,1), \beta\ge 0$, $(\frac{5}{3}-\eps,\beta)$-approximating $k$-center radius requires $m^{k-1-o(1)}$ time under SETH.
\end{theorem}

We will use a similar construction as the simple lower bound (\cref{thm:simplelb}) from the previous section. In order to improve the ratio from $(3/2-\eps)$ to $(5/3-\eps)$, we would like to increase the parameter of the base gadget graph from $\ell$ to $2\ell$, while the paths attached to $B'$ still have length $\ell$. But naively doing this would make some parts of the graph no longer covered by the $k$ chosen centers in $A'$ within the desired radius $3\ell$.
To solve this issue, we pick the node $c'$ as an extra center to cover the remaining parts. In this way, a size-$k$ Set Cover solution of  implies a $(k+1)$-center solution of radius $3\ell$. And, in the converse direction, the same argument as before shows that a $(k+1)$-center solution of radius $<5\ell$ would imply a $(k+1)$-size Set Cover solution. Therefore,  we can use the hardness of Gap Set Cover (distinguishing between solution size $\le k$ or $>k+1$) to conclude the proof.


\begin{proofof}{Proof of \cref{thm:betterlbwarmup1}}
Suppose we are given a Set Cover instance $G=(A,B,E)$ where $|A|=\Theta(n)$ and $|B| \le n^{\gamma}$ (where constant $\gamma>0$ can be chosen arbitrarily small), and want to decide whether it has a size-$k$ solution or has no solutions of size $\le k+1$. By \cref{cor:approxsetcoverlb}, this requires $n^{k-o(1)}$ time under SETH.

  We will create a $(k+1)$-center instance $G'=(V',E')$, in a similar way to the proof of \cref{thm:simplelb}:
 \begin{itemize}
     \item Add a base gadget graph $\Gad(A',B',c',2\ell)$ (\cref{defn:gadget}).
     \item For every $b'\in B'$, attach an $\ell$-edge path $(b',u_{b',1},u_{b',2},\dots,u_{b',\ell})$ to $b'$.
 \end{itemize}

 Observe that the new graph $G'$ has $O(\ell|A||B|) \le n^{1+\gamma}$ edges.

\begin{figure}[h]
    \centering
 \input{fig1}
    \caption{The $(k+1)$-center instance $G'=(V',E')$ constructed in the proof of \cref{thm:simplelb} (for distinguishing radius $\le 3\ell$ or $\ge 5\ell$). Here we use a thick segment to denote a path of prescribed number of edges.  
    Given a Set Cover solution of size $k$ (this example has $k=2$ attained by $\{a_1,a_2\}$), picking their copies in $A'$  together with $c'$  gives a $(k+1)$-center solution of radius $3\ell$. In this example, on the $2\ell$-edge path from $a'_3\in A'$ to $b'_3\in B'$, the first half of nodes are covered by center $c'$, and the second half are covered by center $a_2'$.}
    \label{fig1:53lb}
\end{figure}
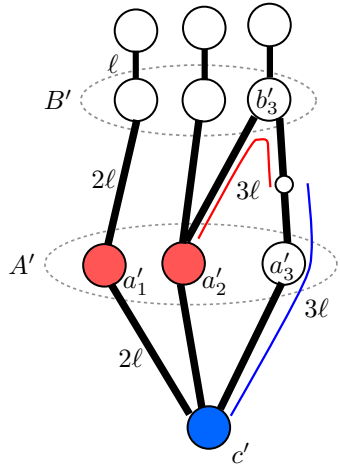

\begin{lemma}
\label{lemma:warmup1claimyes}
    If the original Set Cover instance $G=(A,B,E)$ has a solution $A'\subseteq A$ of size $|A'|= k$, then the $(k+1)$-center instance $G'$ has a solution with radius $3\ell$.
\end{lemma}
\begin{proof}
Given the Set Cover solution $S=\{s_1,\dots,s_k\}\subset A$, in $G'=(V',E')$ we pick the copies of the same $k$ nodes, $s_1',\dots,s_k'\in A' \subset V'$ 
together with $c'\in V'$ as the $(k+1)$ centers, and we now verify that they cover every node in $V'$ within distance $3\ell$ (see \cref{fig1:53lb}):

\begin{itemize}
    \item As before, the Set Cover solution guarantees that every $b'\in B'\subset V'$ can be reached from some center $s'_i$ within distance $2\ell$ via the path $(s_i',v_{s_i',b',1},\dots,v_{s_i',b',2\ell-1},b')$. 

 Then,  the second half of the $2\ell$-edge paths between $A'$ and $b'$, namely the nodes $v_{\cdot,b',j}$ for $j\ge \ell$, are reachable from $b'$ within distance $\le \ell$, and hence reachable from that center within distance $\le 2\ell+\ell=3\ell$.

 Similarly, all nodes $u_{b',\cdot}$ on the $\ell$-edge path attached to $b'$ are also reachable from the center within distance $3\ell$.
    \item Observe that the center $c'$ has distance $\le 3\ell$ to the first half of the $2\ell$-edge paths between $A',B'$, namely $v_{\cdot,\cdot,j}$ for $j\le  \ell$. 
    
    Similarly one can check that $c'$ covers 
    all the remaining nodes in $G'$ (including $c', w_{\cdot,\cdot}$, and nodes in $A'$) within distance $3\ell$.
   \end{itemize}
\end{proof}

\begin{lemma}
If the $(k+1)$-center instance $G'$ has a solution with radius $< 5\ell$, then the original Set Cover instance $G=(A,B,E)$ has a solution of size $(k+1)$.
\end{lemma}
\begin{proofof}{Proof Sketch}
   The proof of this lemma uses exactly the same argument as \cref{lemma:simplelbNO} from the previous section. The current lemma holds for radius $<5\ell$ because the shortest distance from $u_{b',\ell}$ to node $c'$ or any other $b''\in B'\setminus \{b'\}$ is $\ell + 2\ell + 2\ell = 5\ell$.
\end{proofof}
By combining the previous two lemmas, we see that any  algorithm on graphs of $m\le O(n^{1+\gamma})$ edges that distinguishes between $(k+1)$-center radius $\le 3\ell$ and $\ge 5\ell$ can be used to decide whether the original Set Cover instance has a size-$k$ solution or has no solutions of size $\le k+1$, which requires $n^{k-o(1)} \ge m^{k/(1+\gamma) - o(1)}$ time.  Since $\gamma>0$ can be chosen arbitrarily small, we rule out $O(m^{k-\delta})$-time  algorithms for all constant $\delta >0$.
\end{proofof}

\subsection{Warm-up II: Recursively covering the paths}
\label{sec:betterlbwarmup2}
In \cref{sec:betterlbwarmup} we saw how using more centers can lead to higher inapproximability ratio. In this section we develop that idea further and prove the following result.

\begin{theorem}[special case of \cref{thm:fulllbrestated}]
\label{thm:betterlbwarmup2}
   Let $k\ge 2$, $\ell \ge 1$ be constant integers. Assuming SETH,
   distinguishing between radius $\leq 5\ell$ and radius $\geq 9\ell$ for $(2k+2)$-center cannot be done in $O(m^{k-\delta})$ time, for any constant $\delta>0$.

   As a corollary, for any constants $\eps\in (0,1), \beta\ge 0$, $(\frac{9}{5}-\eps,\beta)$-approximating $k$-center radius requires $m^{\lfloor k/2\rfloor -1-o(1)}$ time under SETH.
\end{theorem}
Compared to the previous construction in \cref{sec:betterlbwarmup}, here we will increase the inapproximability ratio by further increasing the parameter of the base gadget graph (namely the distance between $A', B'$ and between $A',c$), which would again cause some of the $A'$-to-$B'$ paths to be uncovered. This time we have to add more centers in order to cover everything: we will create another copy of the base gadget graph $\Gad(\bar A, \bar B, \bar c,\bar L)$ (for some smaller parameter $\bar L$),  and additionally pick $k$ centers from $\bar{A}$ (and $\bar c$) as well. 
We will connect edges appropriately so that the originally uncovered part on the $A'$-to-$B'$ path between node $b'\in B'$ and any $a'\in A'$ will be covered by the new centers by going through $\bar{b}\in \bar{B}$.

We remark that our construction for proving \cref{thm:betterlbwarmup2} will be slightly redundant, and $(2k+2)$ in the theorem statement can in fact be improved to $(2k+1)$ (see \cref{footnote:cbar}). 
We present this slightly weaker version just to keep consistency with the full generalized construction to be described in the next section.

The rest of this section proves \cref{thm:betterlbwarmup2}. 
Suppose we are given a Set Cover instance $G=(A,B,E)$ where $|A|=\Theta(n)$ and $|B| = n^{\gamma}$ (where constant $\gamma>0$ can be chosen arbitrarily small), and want to decide whether it has a size-$k$ solution or has no solutions of size $\le 2k+2$. By \cref{cor:approxsetcoverlb}, this requires $n^{k-o(1)}$ time under SETH.

  We will create a $(2k+2)$-center instance $G'=(V',E')$  of $O(\ell|A||B|) \le n^{1+\gamma}$ edges,
 as follows. See \cref{fig2:95lb}. 

\begin{itemize}
    \item Add two base gadget graphs $\Gad(A',B',c',4\ell)$ and $\Gad(\bar A, \bar B,\bar c,2\ell)$ (\cref{defn:gadget}).
     \item For every $b'\in B'$, attach an $\ell$-edge path $(b',u_{b',1},u_{b',2},\dots,u_{b',\ell})$ to $b'$.
     \item  Then, for every $(a,b)\in E$, let $v_{a',b',2\ell}$ denote the middle node on the $4\ell$-edge path between $a'$ and $b'$ in $\Gad(A',B',c',4\ell)$, and we add a $2\ell$-edge path that connects  $v_{a',b',2\ell}$ and $\bar{b}$ (depicted as dashed blue paths in \cref{fig2:95lb}). 
      Here we stress that both $\bar{b}\in \bar B$ and $b'\in B'$ are copies of the same $b\in B$.
\end{itemize}

\begin{figure}[h]
  \centering
  \begin{minipage}{0.45\textwidth}
    \centering
    \input{fig2}
    \caption{The $(2k+2)$-center instance $G'=(V',E')$ constructed in the proof of \cref{thm:betterlbwarmup2} (distinguishing between radius $\le 5\ell$ or $\ge 9\ell$). For example, the $4\ell$-edge path between $a'$ and $b'$ can be covered within radius $5\ell$ by centers $a'_{\star}, \bar{a}_{\star}, c'$.
    }
    \label{fig2:95lb}
    
  \end{minipage}
  \hfill 
  \begin{minipage}{0.45\textwidth}
    \centering

    \input{fig3}
    \caption{The ``skeleton graph'' of the graph $G'=(V',E')$ in \cref{fig2:95lb}. Here, each super-node represents a subset of nodes in $G'$, and the lengths denote shortest distances between subsets of nodes.}
    \label{fig3:95simple}
    
  \end{minipage}
\end{figure}

\begin{lemma}
\label{lemma:95yes}
    If the original Set Cover instance $G=(A,B,E)$ has a solution $S\subseteq A$ of size $|S|= k$, then the $(2k+2)$-center instance $G'$ has a solution with radius $5\ell$.
\end{lemma}
\begin{proof}
Given the Set Cover solution $S=\{s_1,\dots,s_k\}\subset A$, in $G'=(V',E')$ we pick both copies of the same $k$ nodes, $s_1',\dots,s_k'\in A', \bar{s_1},\dots,\bar{s_k}\in \bar{A}$, and $c',\bar c$ as the $(2k+2)$ centers, and we now verify that they cover every node in $V'$ within distance $5\ell$:

\begin{itemize}
    \item Similarly to \cref{lemma:warmup1claimyes}  from the previous section, here we see that within $5\ell$ radius the centers $s_1',\dots,s_k'\in A'$ together cover all the $\ell$-edge paths attached to $B'$, as well as the last $1/4$ fraction of every $4\ell$-edge path from $A'$ to $B'$.
     Also, within $5\ell$ radius, $c'$ covers all the $4\ell$-edge paths between $c'$ and $A'$, as well as the first $1/4$ fraction of every 
    $4\ell$-edge path from $A'$ to $B'$.
    \item This is the key part in our construction: For every $4\ell$-edge path between $a'\in A$ and $b'\in B'$, it remains to verify that its middle $1/2$ fraction (namely the nodes $v_{a',b',j}$ where $\ell<j<3\ell$) are also covered within $5\ell$ radius, for which it is sufficient to show that the middle node $v_{a',b',2\ell}$ can reach some center within $4\ell$ distance. 
    To show this, suppose $b$ is covered by $s_i\in S$ in the Set Cover solution, and note that the center $\bar{s_i}\in \bar A$ can reach the middle node via a $4\ell$-edge path $\bar{s_i} \overset{2\ell}{\leadsto} \bar{b} \overset{2\ell}{\leadsto} v_{a',b',2\ell}$. (This also shows that the intermediate nodes on the path $\bar{b} \overset{2\ell}{\leadsto} v_{a',b',2\ell}$ are covered.)
    \item Then, observe that the remaining nodes in the graph, namely the nodes on the $2\ell$-edge paths between $\bar A,\bar B$ and between $\bar c, \bar A$, are covered by the center $\bar c$ within $4\ell<5\ell$ distance.\footnote{Alternatively, one can show that they are covered by the centers $\bar s_1,\dots,\bar s_k$ within $4\ell$ distance. Hence it is actually not necessary to include $\bar{c}$ as a center. \label{footnote:cbar}}
\end{itemize}
\end{proof}

Now we proceed to the NO case. For the sake of analysis, we introduce a few terminologies.  We \emph{associate} each node in the $(2k+2)$-center instance $G'$ to at most one node in $A$, and to at most one node in $B$ in the most natural way. Also, some of the nodes in $G'$ are called \emph{type-$A$} (or \emph{type-$B$, type-$C$}). Their precise definitions are given as follows (we state the definitions with a little bit of generality so that it can be reused in the next section where $G'$ may contain even more copies of the base gadget graph):

\begin{definition}
\label{defn:typeassoc}
   In each copy of the base gadget graph $\Gad(\hat A,\hat B,\hat c,L)$ (\cref{defn:gadget}):
   \begin{itemize}
       \item Each $\hat a \in \hat A$ is associated to $a\in A$ (recall that $\hat a$ is a copy of $a$), and we say $\hat a$ is a type-$A$ node.
       \item Similarly, each $\hat b \in \hat B$ is associated to $b\in B$, and we say $\hat b$ is a type-$B$ node.
       \item $\hat c$ is not associated to any node. We say $\hat c$ is a type-$C$ node.
       \item For each $L$-edge path between $\hat a\in \hat A$ and $\hat b\in \hat B$, all internal nodes on this path (namely, $v_{\hat a,\hat b,j}, j\in [L-1]$) are associated to both $a\in A$ and $b\in B$. 
       \item For each $L$-edge path between $\hat a\in \hat A$ and $\hat c$, all internal nodes on this path (namely, $w_{\hat a,j}, j\in [L-1]$) are associated to $a\in A$. 
   \end{itemize}
   Then, in the main construction of $G'=(V',E')$:
   \begin{itemize}
   \item For each $\ell$-edge path $(b',u_{b',1},u_{b',2},\dots,u_{b',\ell})$  attached to $b'\in B'$, all nodes on this path are associated to $b\in B$.
   \item Whenever we add a path from some 
   $v_{\hat a,\hat b,j}$ inside a base gadget graph $\Gad(\hat A, \hat B, \hat c,L)$ to another node outside this base gadget graph, we associate all internal nodes on this added path to both $a\in A$ and $b\in B$. (In the example in \cref{fig2:95lb}, these paths are depicted as dashed blue paths.)
   \end{itemize}
\end{definition}

Now we observe a few simple but useful properties of the constructed instance $G'=(V',E')$ and the way we associated nodes of $V'$ to nodes of $A,B$:
 \begin{obs}
    \label{obs:easy} 
    If a node in $G'$ is associated to both $a\in A$ and to $b\in B$, then $(a,b)\in E$ in the original Set Cover instance.
 \end{obs}
 \begin{obs}
 \label{obs:fourobs}
 For every edge $(x,y)$ in $G'$:
     \begin{enumerate}
         \item (``Remember $a$''): If $x$ is associated to some $a_x\in A$ and $y$ is associated to some $a_y\in A$, then $a_x=a_y$.
         \label{item:remembera}
         \item (``Forget $a$''): If $x$ is associated to some $a_x\in A$, and $y$ is not associated to any $a\in A$, then $y$ is a type-$B$ or type-$C$ node.
         \label{item:forgeta}
         \item (``Remember $b$''): If $x$ is associated to some $b_x\in B$ and $y$ is associated to some $b_y\in B$, then $b_x=b_y$.
         \label{item:rememberb}
         \item (``Forget $b$ but still remember $a$''): If $x$ is associated to some $b_x\in B$, and $y$ is not associated to any $b\in B$, then $y$ is a type-$A$ node, and $y$ is associated to an $a_y\in A$ such that $(a_y,b_x)\in E$ in the Set Cover instance.
         \label{item:forgetb}
     \end{enumerate}
 \end{obs}
 Both observations can be directly verified by carefully examining \cref{defn:gadget}, our definition of $G'=(V',E')$ (see~\cref{fig2:95lb}), and \cref{defn:typeassoc}. 
 
Now we can prove the lemma for the NO case:
\begin{lemma}
\label{lemma:95no}
If the $(2k+2)$-center instance $G'$ has a solution with radius $< 9\ell$, then the original Set Cover instance $G=(A,B,E)$ has a solution of size $2k+2$.
\end{lemma}
\begin{proof}

    Let $\{\tilde s_1,\dots,\tilde s_{2k+2}\}\subset V'$ denote the $(2k+2)$-center solution of $G'=(V',E')$ of radius $< 9\ell$.   We define $\{s_1,\dots,s_{2k+2}\}\subset A$ as follows.
    \begin{definition}
    \label{defn:si}
    Given a center $\tilde s_i \in V'$, we define $s_i\in A$ as follows:
    \begin{itemize}
        \item If $\tilde s_i$ is associated to some $a\in A$ (see \cref{defn:typeassoc}), then let $s_i=a$.
        \item If $\tilde s_i$ is associated to $b\in B$ but not to  any node in $A$, then take an arbitrary $a\in A$ such that $(a,b)\in E$ (which exists by our initial assumption of the Set Cover instance) and let $s_i = a$.
        \item If $\tilde s_i$ is not associated to any node in $A$ or $B$, then just let $s_i$ be an arbitrary node in $A$.
    \end{itemize}
    \end{definition}
 We will show that $\{s_1,\dots,s_{2k+2}\}\subset A$ is a Set Cover solution.
Fix any $b\in B$, and consider $u_{b',\ell}$ (the last node on the $\ell$-edge path attached to $b'\in B'$). There is a center $\tilde s_i$ at distance $< 9\ell$ from $u_{b',\ell}$. It remains to prove that $(s_i,b)\in E$ in the Set Cover instance.

Let $P$ denote the shortest path from $u_{b',\ell}$ to the center $\tilde s_i$ of length $|P|< 9\ell$. 
We make the following definition.
\begin{definition}
\label{defn:badpath}
We say a path $P = (p_0,p_1,\dots,p_{|P|})$ is \emph{bad}, if there exists $0\le x<y\le |P|$ such that $p_x$ is a type-$A$ node, and $p_y$ is a type-$B$ or type-$C$ node.
\end{definition}

Then we will prove the following two lemmas, which together with $|P|<9\ell$ immediately imply $(s_i,b)\in E$ as desired, finishing the proof of \cref{lemma:95no}.
 \begin{lemma}
 \label{lem:badislong9ell}
    In $G'$, any bad path $P$ starting from $u_{b',\ell}$ must have length $\ge 9\ell$. 
 \end{lemma}
\begin{lemma}
\label{lem:nonbadisgood}
    If there is a path $P$ from $u_{b',\ell}$ to $\tilde s_i$ that is not bad, then $s_i$ (as defined in \cref{defn:si}) satisfies $(s_i,b)\in E$ in the Set Cover instance.
\end{lemma}

 \begin{proofof}{Proof of \cref{lem:badislong9ell}}
In our construction of $G'$, type-$A$ nodes are $A'\cup \bar A$, type-$B$ nodes are $B'\cup \bar B$, and type-$C$ nodes are $\{c',\bar c\}$. 
In a bad path $P= (p_0,\dots,p_{|P|})$ where $p_0 = u_{b',\ell}$, suppose $p_x$ is a type-$A$ node and $p_y$ is a type-$B$ or type-$C$ node ($0\le x<y\le |P|$).
 We now use a case distinction. (The distance lower bounds we are using here can be seen more transparently from the ``skeleton graph'' of our construction depicted in \cref{fig3:95simple}.)
\begin{itemize}
    \item Case $p_x\in A'$: Observe $x\ge d_{G'}(u_{b',\ell}, A') = 5\ell$, and $y-x \ge d_{G'}(A', \{c',\bar c\}\cup B'\cup \bar B) = 4\ell$.
    \item Case $p_x \in \bar A$: Observe $x\ge d_{G'}(u_{b',\ell}, \bar A) = 7\ell$, and $y-x \ge d_{G'}(\bar A, \{c',\bar c\}\cup B'\cup \bar B) = 2\ell$.
\end{itemize}
In both cases we have $|P|\ge y= x+(y-x) \ge 9\ell$.
 \end{proofof}

\begin{proofof}{Proof of \cref{lem:nonbadisgood}}
 Let $P=(p_0,p_1,\dots,p_{|P|})$. 
 Pick the smallest $x\in [0, |P|]$  such that $p_x$ is not associated to any node in $B$ (if none exists, let $x=|P|+1$).
 Since $p_0=u_{b',\ell}$ is associated to $b$, by repeatedly applying \cref{item:rememberb} of \cref{obs:fourobs} we know $p_{x-1}$ is also associated to $b$. If $x=|P|+1$, then this means $p_{|P|} = \tilde s_i$ is associated to $b$, and then from \cref{defn:si} and \cref{obs:easy} we have $(s_i,b)\in E$ as claimed. So we  assume $x\le |P|$ from  now on.
 
Since $p_{x-1}$ is associated to $b\in B$ but $p_x$ is not associated to any node in $B$, by \cref{item:forgetb} of \cref{obs:fourobs} we know $p_{x}$ is a type-$A$ node and is associated to some $a\in A$ such that $(a,b)\in E$.
  Pick the smallest $y\in [x,|P|]$ such that $p_y$ is not associated to any node in $A$ (if none exists, let $y=|P|+1$).
 Then, by repeatedly applying \cref{item:remembera} of \cref{obs:fourobs} we know $p_{y-1}$ is also associated to the same $a$.
 If $y=|P|+1$, then this means $p_{|P|}=\tilde s_i$ is associated to the same $a\in A$ satisfying $(a,b)\in E$, and we have $s_i=a$ by \cref{defn:si}, and hence $(s_i,b)\in E$ as claimed.
 So we assume $y\le |P|$ from now on.

Since $p_{y-1}$ is associated to $a\in A$ but $p_y$ is not associated to any node in $A$, by \cref{item:forgeta} of \cref{obs:fourobs} we know $p_{y}$ is a type-$B$ or type-$C$ node. Hence, we have found $0\le x<y\le |P|$ witnessing that  $P$ is a bad path, contradicting the assumption that $P$ is not bad.
\end{proofof}
\end{proof}

\begin{proofof}{Proof of \cref{thm:betterlbwarmup2}}
By combining \cref{lemma:95yes} and \cref{lemma:95no}, we see that any $(2k+2)$-center algorithm on graphs of  $m\le O(n^{1+\gamma})$ edges that distinguishes between  radius $\le 5\ell$ and $\ge 9\ell$ can be used to decide whether the original Set Cover instance has a size-$k$ solution or has no solutions of size $\le 2k+2$, which  requires $n^{k-o(1)} \ge m^{k/(1+\gamma) - o(1)}$ time.  Since $\gamma>0$ can be chosen arbitrarily small, we rule out $O(m^{k-\delta})$-time  algorithms for all constant $\delta >0$.
\end{proofof}

\subsection{The full recursive construction for \texorpdfstring{$(2-\eps,\beta)$}{(2-eps,beta)}-inapproximability}
\label{sec:fulllb}
    \begin{algorithm}[t]
         \caption{Recursive construction for proving \cref{thm:fulllbrestated}}\label{alg:lbrecursion}
         \begin{algorithmic}[1]
            \item \textbf{Global Input:} A Set Cover instance $G=(A,B,E)$ and constant integers $t,\ell\ge 1$.

\State 
            
            \Procedure{Recurse}{$p$}
            \item \textbf{Input:} positive integer $1\le p\le 2t$ 
            \item \textbf{Output:} a tuple $(\bar G, \bar A,\bar B)$, where $\bar G=(\bar V,\bar E)$ is a graph, and  $\bar A, \bar B\subseteq \bar V$ are copies of $A,B$
            \State $\bar G \gets$ a fresh copy of the base gadget graph $\Gad(\bar A, \bar B, \bar c, (2t+1-p)\ell)$\label{line:pathab}
                            \For{$q \in \{1,2,\dots, \lfloor \frac{2t-p}{2p}\rfloor  \}$}\Comment{Non-empty iff $p\le 2t/3$}
                                \State $(G'',A'',B'') \gets \textsc{Recurse}((2q+1)\cdot p)$\label{line:recurse} \Comment{$(2q+1)p\le 2t$ holds}
                                \State Add $G''$ into $\bar G$
                                \For{$(a,b)\in E$}
                               \State Add a $2qp\ell$-edge path in $\bar G$ between $v_{\bar a, \bar b, (2t+1-p-2qp)\ell}\in\Gad(\bar A, \bar B, \bar c, (2t+1-p)\ell)$ and  $b''\in B''$ 
            \label{line:interpath}
            \Comment{By convention, $\bar b\in \bar B, b''\in B''$ are copies of $b\in B$}
                                \EndFor
                            \EndFor
                \State \Return{$(\bar G, \bar A,\bar B)$}
            \EndProcedure

            \State 
            \Procedure{Main}{}
                \State $(G',A',B') \gets \textsc{Recurse}(1)$ 
                 \State   For every $b'\in B'$, attach an $\ell$-edge path $(b',u_{b',1},u_{b',2},\dots,u_{b',\ell})$ to $b'$.
            \label{line:pathb}
                 \State \Return{$G'$}
            \EndProcedure
            \end{algorithmic}
    \end{algorithm}   
\begin{figure}[t]
    \centering
 \input{fig4}
    \caption{The ``skeleton graph'' of graph $G'$ constructed by \cref{alg:lbrecursion} for $t=7$ (for distinguishing between radius $\le 15\ell$ or $\ge 29\ell$ for $f(7)\cdot (k+1)$-center). To avoid clutter, we do not draw the graph $G'$ itself; one can refer to \cref{fig2:95lb} and \cref{fig3:95simple} to understand the relation between $G'$ and its skeleton graph by analogy.}
    \label{fig4:fullt7}
\end{figure}
In this section we prove our full hardness result \cref{thm:fulllbrestated}.
The proof generalizes the construction from the previous section (\cref{thm:betterlbwarmup2}) by adding even more copies of the base gadget graph, and hence achieves better inapproximability ratio. More specifically:
\begin{itemize}
    \item In the proof of \cref{thm:betterlbwarmup2} we only had one junction point on each $4\ell$-edge path between $A'$ and $B'$. Now we will put more junction points (in order to cover a bigger fraction of the path).
    \item The proof of \cref{thm:betterlbwarmup2} only had one level of recursion, but in general we may continue the recursion by, for example, adding junction points on the paths between $\bar A$ and $\bar B$ and connect them to another copy of the base gadget graph.
\end{itemize}

Throughout, let $t,\ell\ge 1$ be fixed integers, and suppose
 we are given a Set Cover instance $G=(A,B,E)$ where $|A|=\Theta(n)$ and $|B| \le n^{\gamma}$ (where constant $\gamma>0$ can be chosen arbitrarily small).
We describe the general construction of graph $G'=(V',E')$ as a recursive procedure in \cref{alg:lbrecursion}.   The recursion is parameterized by a positive integer variable $p$.  We remark that the constructions from \cref{sec:betterlbwarmup} and \cref{sec:betterlbwarmup2} can be obtained from running \cref{alg:lbrecursion} with $t=1$ and $t=2$, respectively. For reference, in \cref{fig4:fullt7} we include the ``skeleton graph'' of the graph $G'$ returned by \cref{alg:lbrecursion} for $t=7$.

Let $h(p)$ denote the total number of base gadget graphs contained in the graph returned by $\textsc{Recurse}(p)$. The final graph $G'$ returned by \cref{alg:lbrecursion} contains $h(1)$ copies of base gadget graphs ($h(1)$ is a function of $t$),  and in total $O(h(1)\cdot |A||B|t\ell) \le O(n^{1+\gamma})$ many edges and nodes. We will analyze the dependence of $h(1)$ on $t$ in the end of this section.

As in the previous section,  we use the same rule as \cref{defn:typeassoc} to associate each node in the constructed instance $G'$ to at most one node in $A$ and at most one node in $B$, and use the same definition of type-$A$ (type-$B$, type-$C$) nodes as  \cref{defn:typeassoc}.
We use the same definition of bad paths as \cref{defn:badpath}.

By inspecting our recursive construction, we can see that both \cref{obs:easy} and \cref{obs:fourobs} still hold for the constructed $G'$.


We inductively prove the following properties of the recursive construction.
\begin{lemma}
\label{lem:recursiveprop}
   Suppose $\textsc{{\rm\textsc{Recurse}}}(p)$ returns $(\bar G,\bar A,\bar B)$.  Then both of the following hold.
   \begin{enumerate}
       \item Any bad path in $\bar G$ starting from any $\bar b\in \bar B$ must have length at least $\ge (4t+2-2p)\ell$.
       \label{item:recurseproperty1}
       \item Suppose the original Set Cover instance has a size-$k$ solution. If we pick the copies of these $k$ nodes in $\hat A$ and node $\hat c$ in all base gadget graphs $\Gad(\hat A,\hat B,\hat c,L)$ in $\bar G$ as centers (picking $h(p)\cdot (k+1)$ centers in total),  then all nodes in $\bar G$ are covered within radius $(2t+1)\ell$. Moreover, all nodes in $\bar B$ are covered within radius $(2t+1-p)\ell$.
       \label{item:recurseproperty2}
   \end{enumerate}
\end{lemma}
\begin{proofof}{Proof of \cref{lem:recursiveprop}, \cref{item:recurseproperty1}}
Take the shortest bad path $P = (p_0,\dots,p_{|P|})$ in $\bar G$ starting from any $p_0 = \bar b\in \bar B$, and let $p_x$ be the first type-$A$ node on $P$. It ends at some type-$B$ or type-$C$ node $p_{|P|}$ (since otherwise $P$ has a proper prefix that is also bad). There are two cases:
   \begin{itemize}
       \item Case 1: $p_x\in \bar A$.

       We have $x\ge d_{\bar G}(\bar B, \bar A) = (2t+1-p)\ell$ by Line~\ref{line:pathab}. The type-$B$ or type-$C$ node $p_{|P|}$ is either from $\bar B \cup \{\bar c\}$, or from some subgraph $G''$ constructed in a recursive call to $\textsc{Recurse}((2q+1)\cdot p)$ at Line~\ref{line:recurse}. In the former case, observe that $|P|-x \ge  d_{\bar G}(\bar A, \bar B\cup \{ \bar c\}) = (2t+1-p)\ell$ (by Lines~\ref{line:pathab}). In the latter case, from Line~\ref{line:interpath} we observe that $|P| - x \ge d_{\bar G}(\bar A, B'') = (2t+1-p-2qp)\ell + 2pq\ell = (2t+1-p)\ell$. In either case we get $|P| = x + (|P|-x) \ge (4t+2-2p)\ell$.

       \item Case 2: $p_x\notin \bar A$.

       Then, $p_x$ can only be some type-$A$ node from some subgraph $G''$ constructed in a recursive call to $\textsc{Recurse}((2q+1)\cdot p)$ at Line~\ref{line:recurse}. 
       Note that the path from $p_0$ to $p_x$ has to visit $B''$ in order to enter $G''$. Let $p_j$ be the maximum $j<x$ such that $p_j\in B''$. Then the path $(p_j,\dots,  p_x)$ must be entirely in $G''$.  Note that the path $(p_x,\dots,p_{|P|})$ is also entirely in $G''$, since otherwise it has to exit $G''$ through some $p_y\in B''$, but in that case the path $(p_0,\dots,p_x,\dots,p_y)$ is a shorter bad path. Hence, we have argued that the bad path $(p_j,\dots,p_x,\dots,p_{|P|})$ starting from $p_j \in B''$ is entirely in $G''$, and hence by the inductive hypothesis it must have at length $ |P| - j\ge (4t+2 - 2(2q+1)\cdot p)\ell$. On the other hand, from Line~\ref{line:interpath} observe that $j = d_{\bar G}(\bar b, p_j)\ge d_{\bar G}(\bar B, B'') = \big ( (2t+1-p)\ell - (2t+1-p-2qp)\ell  \big ) + 2qp\ell = 4qp\ell$. So $|P| = j + (|P|-j) \ge (4t+2-2p)\ell$.
   \end{itemize}
\end{proofof}
\begin{proofof}{Proof of \cref{lem:recursiveprop}, \cref{item:recurseproperty2}}
Inside $\Gad(\bar A, \bar B, \bar c, (2t+1-p)\ell)$ added at Line~\ref{line:pathab}, within radius $(2t+1)\ell$, center $\bar c$ covers all the $(2t+1-p)\ell$-edge paths adjacent to $\bar c$, and also covers all $v_{\bar a,\bar b, j}$ where $1\le j\le (2t+1)\ell - (2t+1-p)\ell  = p\ell$. 
Also, the Set Cover solution guarantees that picking the $k$ centers in $\bar A$ can cover all nodes in $\bar B$ within radius $(2t+1-p)\ell$ (this proves the ``Moreover'' part), and hence cover all $v_{\bar a, \bar b, j}$ where $j\ge (2t+1-2p)\ell$ within radius $(2t+1)\ell$.

    The remaining nodes on the $\bar A$-to-$\bar B$ paths are those $v_{\bar a,\bar b,j}$ with $p\ell < j < (2t+1-2p)\ell$, and we can check that each of them is at distance $\le p\ell$ to a junction node $v_{\bar a, \bar b, (2t+1-p-2qp)\ell}$ for some $q\in \{1,2,\dots, \lfloor \frac{2t-p}{2p}\rfloor \}$, which is connected via a $2qp\ell$-edge path to  $b'' \in B''$ (by Line~\ref{line:interpath}). Since $(G'',A'',B'')$ is returned from $\textsc{Recurse}((2q+1)\cdot p)$, by the inductive hypothesis, $b''$ is covered by some center within radius $(2t+1 - (2q+1)\cdot p)\ell$, so this center covers $v_{\bar a, \bar b, (2t+1-p-2qp)\ell}$ (and all nodes on that $2qp\ell$-edge path) within radius  $(2t+1 - (2q+1)\cdot p)\ell + 2qp\ell = (2t+1-p)\ell$. Hence, all remaining nodes $v_{\bar a,\bar b,j}$ (where $p\ell < j < (2t+1-2p)\ell$) are covered within radius $(2t+1-p)\ell + p\ell = (2t+1)\ell$.

    By the inductive hypothesis, all nodes in the recursively created subgraphs $G''$ are also covered within radius $(2t+1)\ell$. Hence, we have verified that all nodes in $\bar G$ are covered within radius $(2t+1)\ell$.
\end{proofof}
Now we get the corollary for the top level of the recursion.
\begin{cor}
\label{cor:toplevelrecursiveprop}
   Suppose \cref{alg:lbrecursion} returns $(G',A', B')$ (where $G'$ contains $h(1)$ copies of the basic gadget graph).  Then both of the following hold.
   \begin{enumerate}
       \item Any bad path in $G'$ starting from any $u_{b',\ell}$ (defined in Line~\ref{line:pathb}) must have length at least $\ge (4t+1)\ell$.
       \label{item:toplevelrecurseproperty1}
       \item Suppose the original Set Cover instance has a size-$k$ solution. We can cover all nodes in $G'$
       within radius $(2t+1)\ell$ by picking  $h(1)\cdot (k+1)$ centers.
       \label{item:toplevelrecurseproperty2}
   \end{enumerate}
\end{cor}
\begin{proof}
    For \cref{item:toplevelrecurseproperty1}, note that in order for the bad path $P$ from $u_{b',\ell}$ to reach any type-$A$ node, it must first traverse the $\ell$-edge path $(u_{b',\ell},\dots,u_{b',1},b')$. Since none of the traversed nodes so far are type-$A$ nodes, by the definition of bad paths (\cref{defn:badpath}) we know the suffix of $P$ starting from $b'$ is still a bad path. By \cref{lem:recursiveprop} \cref{item:recurseproperty1} (for $p=1$), this suffix path must have length at least $(4t+2-2)\ell = 4t\ell$. Hence $|P| \ge \ell + 4t\ell = (4t+1)\ell$.

    For \cref{item:toplevelrecurseproperty2}, note that every $b'\in B'$ is covered within radius $(2t+1-1)\ell = 2t\ell$ by \cref{lem:recursiveprop} \cref{item:recurseproperty2}, and hence all nodes on the $\ell$-edge path attached to $b$ (at Line~\ref{line:pathb}) are covered within $(2t+1)\ell$ radius. All the remaining nodes in $G'$ are also covered within $(2t+1)\ell$ radius by \cref{lem:recursiveprop} \cref{item:recurseproperty2}.
\end{proof}

\cref{cor:toplevelrecursiveprop} \cref{item:recurseproperty2} already deals with the YES case.
Now we can use \cref{cor:toplevelrecursiveprop} \cref{item:recurseproperty1} with the same arguments as before to analyze the NO case.
\begin{cor}
\label{cor:generalno}
If the  graph $G'$ constructed by \cref{alg:lbrecursion} has a $h(1)\cdot (k+1)$-center solution with radius $< (4t+1)\ell$, then the original Set Cover instance $G=(A,B,E)$ has a solution of size $h(1)\cdot (k+1)$.
\end{cor}
\begin{proof}
   The proof uses the same argument as 
   \cref{lemma:95no} from previous section.
    Let $\{\tilde s_1,\dots,\tilde s_{h(1)\cdot(k+1)}\}\subset V'$ denote the $(h(1)\cdot(k+1))$-center solution of $G'=(V',E')$ of radius
    $< (4t+1)\ell$.  
    We define $\{s_1,\dots,s_{h(1)\cdot(k+1)}\}\subset A$ using \cref{defn:si}. For any $b'\in B'$, consider $u_{b',\ell}$ (the last node on the $\ell$-edge path attached to $b'\in B'$), and consider its shortest path $P$ to a center $\tilde s_i$ at distance $< (4t+1)\ell$.
   By \cref{cor:toplevelrecursiveprop} \cref{item:toplevelrecurseproperty1}, $P$ cannot be a bad path. Then by 
  \cref{lem:nonbadisgood},  we have $(s_i,b)\in E$ in the Set Cover instance, finishing the proof that $\{s_1,\dots,s_{h(1)\cdot(k+1)}\}\subset A$ is a Set Cover solution.
\end{proof}

\begin{proofof}{Proof of \cref{thm:fulllbrestated}}
By combining \cref{cor:toplevelrecursiveprop} \cref{item:recurseproperty2} and \cref{cor:generalno}, we see that any $(h(1)(k+1))$-center algorithm on graphs of  $m\le O(n^{1+\gamma})$ edges that distinguishes between  radius $\le (2t+1)\ell$ and $\ge (4t+1)\ell$ can be used to decide whether the original Set Cover instance has a size-$k$ solution or has no solutions of size $\le (h(1)(k+1))$, which  requires $n^{k-o(1)} \ge m^{k/(1+\gamma) - o(1)}$ time by \cref{cor:approxsetcoverlb}.  Since $\gamma>0$ can be chosen arbitrarily small, we rule out $O(m^{k-\delta})$-time  algorithms for all constant $\delta >0$.  Here, $h(1) = f(t)$ is a function of $t$. The last lemma gives a (loose) upper bound on $f(t)$.

\begin{lemma}
   $f(t)  = h(1)\le \begin{cases} t & 1\le t\le 4,\\ 6 & t=5, \\ t^2 & \text{otherwise.}  \end{cases}$
\end{lemma}
\begin{proof}
By \cref{alg:lbrecursion} we have the following recurrence:
\begin{equation}
    h(p) = \begin{cases} 1 & 2t/3< p\le 2t, \\ 1+\sum_{q=1}^{\lfloor (2t-p)/(2p)\rfloor} h((2q+1)\cdot p) & 1\le p\le 2t/3.\end{cases}
\end{equation}

Equivalently, $h(1)$ counts the number of integer sequences $(p_1,\dots,p_{s-1}, p_s)$ (with flexible length $s \ge 1$) where $p_1 = 1$, $p_{s} \le 2t$ and $p_{i}/p_{i-1} \in \{3,5,7,9,\dots\}$. 
As an alternative way to count these sequences, let $g(T)$ denote the number of sequences $(p_1,\dots,p_{s-1}, p_s)$ (with flexible length $s \ge 1$) where $p_1 = 1$, $p_{i}/p_{i-1} \in \{3,5,7,9,\dots\}$, and $p_s \le T$. Then $h(1)=g(2t)$, and $g$ has the recurrence 
\begin{equation}
g(T)  = \begin{cases}
    1 + g(\lfloor T/3\rfloor ) + g(\lfloor T/5 \rfloor )+g(\lfloor T/7 \rfloor )+g(\lfloor T/9 \rfloor )+\dots & T\ge 1 \\ 0 & T=0.
\end{cases} 
\end{equation}
Here are values of $g(T)$ for some small $T$:  $g(0)=0$, $g(1)=g(2)=1, g(3)=g(4)=2, g(5)=g(6) = 3, g(7)=g(8)=4,g(9)=g(10)=6$.

We will not attempt to study the asymptotics of $g(T)$. Instead, we prove a very loose upper bound $g(T)\le \max\{1, 0.25 T^2\}$ by induction:  For $T\le 10$, this inequality holds. For $T \ge 11$,  we have $g(\lfloor T/(2k+1)\rfloor) = 1$ if $T/(2k+1) \in [1,3)$, and otherwise $g(\lfloor T/(2k+1)\rfloor) \le 0.25T^2/(2k+1)^2$ by the inductive hypothesis. Hence,
\begin{align*}
g(T) &\le 1+|\{k\in \N^+: T/(2k+1)\in [1,3)\}| +  \sum_{k\in \N^+} 0.25T^2/(2k+1)^2 \\
& \le  T/2 + 0.25T^2\cdot 0.234\\
& <  0.25T^2.
\end{align*}
Hence, $f(t) = g(2t) \le t^2$.
\end{proof}
\end{proofof}

Finally, we remark that if we assume ETH instead of SETH, we can use the same proof as above (but using the ETH-hardness of Gap Set Cover \cref{lem:approxsetcoverlbETH} instead) and show that there can be no $f(k)n^{o(k)}$ time $(2-\eps,\beta)$-approximation algorithm for $k$-Center for $\eps>0,\beta = O(1)$.  

%% file: fig0replace.tex
\definecolor{c909090}{RGB}{144,144,144}
\definecolor{cffffab}{RGB}{255,255,171}
\definecolor{c009a00}{RGB}{0,154,0}

\def \globalscale {1.000000}
\begin{tikzpicture}[y=1cm, x=1cm, yscale=\globalscale,xscale=\globalscale, every node/.append style={scale=\globalscale}, inner sep=0pt, outer sep=0pt]
  \path[draw=black,line width=0.0265cm] (4.2556, 20.9805) circle (0.2827cm);

  \path[draw=black,line width=0.0265cm] (4.6752, 23.1808) circle (0.2827cm);

  \path[draw=black,line width=0.0265cm] (5.308, 21.0067) circle (0.2827cm);

  \path[draw=black,line width=0.0265cm] (6.4242, 21.0199) circle (0.2827cm);

  \path[draw=black,line width=0.0265cm] (8.2645, 21.0001) circle (0.2827cm);

  \path[draw=black,line width=0.0265cm] (6.0958, 18.798) circle (0.2827cm);

  \path[draw=black,line width=0.0265cm] (5.6802, 23.1888) circle (0.2827cm);

  \path[draw=black,line width=0.0265cm] (7.2809, 23.1708) circle (0.2827cm);

  \path[draw=c909090,line width=0.0265cm,dash pattern=on 0.053cm off 0.053cm] (5.9742, 23.1752) ellipse (2.0311cm and 0.4687cm);

  \path[draw=c909090,line width=0.0265cm,dash pattern=on 0.053cm off 0.053cm] (6.3239, 20.9953) ellipse (2.8569cm and 0.5357cm);

  \path[draw=black,fill=cffffab,line width=0.0265cm] (4.6871, 23.4654) -- (4.8062, 25.1914);

  \path[draw=black,fill=white,line width=0.0265cm] (4.8136, 25.1542) circle (0.1562cm);

  \path[draw=black,fill=white,line width=0.0265cm] (4.7782, 24.5382) circle (0.1562cm);

  \path[draw=black,fill=white,line width=0.0265cm] (4.7286, 23.9446) circle (0.1562cm);

  \path[draw=black,fill=cffffab,line width=0.0265cm] (5.7156, 23.4659) -- (5.8346, 25.1919);

  \path[draw=black,fill=white,line width=0.0265cm] (5.842, 25.1547) circle (0.1562cm);

  \path[draw=black,fill=white,line width=0.0265cm] (5.8066, 24.5387) circle (0.1562cm);

  \path[draw=black,fill=white,line width=0.0265cm] (5.757, 23.9451) circle (0.1562cm);

  \path[draw=black,fill=cffffab,line width=0.0265cm] (7.2901, 23.4513) -- (7.4092, 25.1773);

  \path[draw=black,fill=white,line width=0.0265cm] (7.4166, 25.1401) circle (0.1562cm);

  \path[draw=black,fill=white,line width=0.0265cm] (7.3812, 24.524) circle (0.1562cm);

  \path[draw=black,fill=white,line width=0.0265cm] (7.3316, 23.9305) circle (0.1562cm);

  \path[draw=black,fill=white,line width=0.0265cm] (4.3821, 20.7424) -- (5.818, 18.9271);

  \path[draw=black,fill=white,line width=0.0265cm] (5.9445, 19.0461) -- (5.4088, 20.7721);

  \path[draw=black,fill=white,line width=0.0265cm] (6.4281, 20.7498) -- (6.1677, 19.0759);

  \path[draw=black,fill=white,line width=0.0265cm] (8.1392, 20.7424) -- (6.3685, 18.9047);

  \path[draw=black,fill=white,line width=0.0265cm] (4.6425, 22.9148) -- (4.2854, 21.2855);

  \path[draw=black,fill=white,line width=0.0265cm] (4.7838, 22.9148) -- (5.2079, 21.2929);

  \path[draw=black,fill=white,line width=0.0265cm] (5.7808, 22.9297) -- (6.2941, 21.2855);

  \path[draw=black,fill=white,line width=0.0265cm] (5.8998, 23.019) -- (8.0276, 21.1813);

  \path[draw=black,fill=white,line width=0.0265cm] (7.4324, 22.9223) -- (8.1764, 21.2855);

  \path[draw=black,fill=white,line width=0.0265cm] (4.526, 22.3922) circle (0.1562cm);

  \path[draw=black,fill=white,line width=0.0265cm] (4.9336, 22.4082) circle (0.1562cm);

  \path[draw=black,fill=white,line width=0.0265cm] (5.0684, 21.7651) circle (0.1562cm);

  \path[draw=black,fill=white,line width=0.0265cm] (5.9273, 22.4164) circle (0.1562cm);

  \path[draw=black,fill=white,line width=0.0265cm] (6.1317, 21.8445) circle (0.1562cm);

  \path[draw=black,fill=white,line width=0.0265cm] (6.5839, 22.4255) circle (0.1562cm);

  \path[draw=black,fill=white,line width=0.0265cm] (7.2618, 21.8462) circle (0.1562cm);

  \path[draw=black,fill=white,line width=0.0265cm] (7.9175, 21.8443) circle (0.1562cm);

  \path[draw=black,fill=white,line width=0.0265cm] (7.6283, 22.4718) circle (0.1562cm);

  \path[draw=black,fill=white,line width=0.0265cm] (4.8333, 20.1687) circle (0.1562cm);

  \path[draw=black,fill=white,line width=0.0265cm] (5.2725, 19.5724) circle (0.1562cm);

  \path[draw=black,fill=white,line width=0.0265cm] (5.7488, 19.6475) circle (0.1562cm);

  \path[draw=black,fill=white,line width=0.0265cm] (6.2556, 19.6508) circle (0.1562cm);

  \path[draw=black,fill=white,line width=0.0265cm] (6.9328, 19.497) circle (0.1562cm);

  \path[draw=black,fill=white,line width=0.0265cm] (7.5719, 20.098) circle (0.1562cm);

  \path[draw=black,fill=white,line width=0.0265cm] (6.3502, 20.2089) circle (0.1562cm);

  \path[draw=black,fill=white,line width=0.0265cm] (5.5713, 20.2097) circle (0.1562cm);

  \path[draw=black,fill=white,line width=0.0265cm] (4.379, 21.7464) circle (0.1562cm);

  \path[draw=c009a00,line width=0.0254cm,dash pattern=on 0.1018cm off 0.0254cm,rounded corners=0.066cm] (2.968, 23.7263) rectangle (9.3525, 18.3537);

  \node[draw=none] at (4.2965, 25.18775) {$u_{b',3}$};

  \node[draw=none] at (4.26305, 24.5442) {$u_{b',2}$};
  
  \node[draw=none] at (4.2147, 23.9564) {$u_{b',1}$};
  
  \node[draw=none] at (4.6871, 23.17895) {$b'$};
  
  \node[draw=none] at (3.6348, 23.16035) {$B'$};
  
  \node[draw=none] at (3.1847, 20.9656) {$A'$};
  
  \node[draw=none] at (3.83525, 22.43125) {$v_{a',b',2}$};
  
  \node[draw=none] at (3.70575, 21.76165) {$v_{a',b',1}$};
  
  \node[draw=none] at (4.2556, 21.0065) {$a'$};
  
  \node[draw=none] at (4.22585, 20.1323) {$w_{a',1}$};
  
  \node[draw=none] at (4.63875, 19.552) {$w_{a',2}$};
  
  \node[draw=none] at (6.0895, 18.79315) {$c'$};
  
  \node[draw=none] at (6.48755, 23.1864) {$\cdots$};
  
  \node[draw=none] at (7.32085, 21.0177) {$\cdots$};
  
  \node[draw=none] at (10.1294, 24.8844) {$(\ell = 3)$};
  
  \node[draw=none] at (10.1294, 18.09885) {{\textcolor{c009a00}{$\Gad(A',B',c',\ell)$}}};

\end{tikzpicture}

%% file: fig1.tex
\definecolor{cff5555}{RGB}{255,85,85}
\definecolor{c0066ff}{RGB}{0,102,255}
\definecolor{c909090}{RGB}{144,144,144}

\def \globalscale {1.000000}
\begin{tikzpicture}[y=1cm, x=1cm, yscale=\globalscale,xscale=\globalscale, every node/.append style={scale=\globalscale}, inner sep=0pt, outer sep=0pt]
  \path[draw=black,fill=cff5555,nonzero rule,line width=0.0265cm] (1.6765, 24.6625) circle (0.2827cm);

  \path[draw=black,fill=c0066ff,line width=0.0265cm] (3.0646, 22.5051) circle (0.2827cm);

  \path[draw=black,line width=0.0265cm] (2.096, 26.8629) circle (0.2827cm);

  \path[draw=black,fill=cff5555,line width=0.0265cm] (2.7289, 24.6888) circle (0.2827cm);

  \path[draw=black,line width=0.0265cm] (4.061, 24.6865) circle (0.2827cm);

  \path[draw=black,line width=0.0265cm] (2.0959, 27.7844) circle (0.2827cm);

  \path[draw=black,line width=0.0265cm] (3.8708, 27.8608) circle (0.2827cm);

  \path[draw=black,line width=0.0265cm] (3.866, 26.8696) circle (0.2827cm);

  \path[draw=black,line width=0.0265cm] (2.9979, 26.8594) circle (0.2827cm);

  \path[draw=c909090,line width=0.0233cm,dash pattern=on 0.0466cm off 0.0466cm] (2.9267, 26.8573) ellipse (1.5637cm and 0.4702cm);

  \path[draw=c909090,line width=0.0219cm,dash pattern=on 0.0438cm off 0.0438cm] (2.8263, 24.6774) ellipse (1.9407cm and 0.538cm);

  \path[draw=black,line width=0.1cm] (2.108, 26.5969) -- (1.7211, 24.9304);

  \path[draw=black,line width=0.0811cm] (2.0931, 27.1475) -- (2.0931, 27.5195);

  \path[draw=black,line width=0.0826cm] (3.8761, 27.1641) -- (3.8761, 27.5808);

  \path[draw=black,line width=0.0265cm] (3.0037, 27.8212) circle (0.2827cm);

  \path[draw=black,line width=0.0826cm] (3.009, 27.1246) -- (3.009, 27.5412);

  \path[draw=black,line width=0.1cm] (4.1108, 24.9638) -- (4.0052, 26.6205);

  \path[draw=black,line width=0.1cm] (1.7658, 24.4096) -- (2.8073, 22.6835);

  \path[draw=black,line width=0.1cm] (4.0357, 24.4295) -- (3.217, 22.7529);

  \path[draw=black,line width=0.1cm] (2.6883, 24.4096) -- (2.971, 22.7728);

  \path[draw=black,line width=0.1cm] (3.6721, 26.6383) -- (2.7671, 24.9862);

  \path[draw=black,line width=0.1cm] (2.7161, 24.9521) -- (2.934, 26.5988);

  \path[draw=blue,line width=0.03cm] (3.3564, 22.6611).. controls (4.4927, 24.6496) and (4.4085, 24.576) .. (4.4296, 24.8496).. controls (4.4506, 25.1231) and (4.377, 25.7544) .. (4.377, 25.7544);

  \path[draw=red,line width=0.03cm] (2.925, 25.0179).. controls (2.925, 25.0179) and (3.6194, 26.3331) .. (3.6931, 26.3331).. controls (3.7667, 26.3331) and (3.8509, 26.3752) .. (3.8509, 26.2594).. controls (3.8509, 26.1437) and (3.8825, 25.7018) .. (3.8825, 25.7018);

;

  \path[draw=black,fill=white,nonzero rule,line width=0.0265cm] (4.0666, 25.7386) circle (0.1156cm);

  \node[draw=none] at (1.0627, 26.8697) {$B'$};

\node[draw=none] at (0.584, 24.6865) {$A'$};

\node[draw=none] at (2.0933, 24.44975) {$a'_1$};

\node[draw=none] at (3.1433, 24.44975) {$a'_2$};

\node[draw=none] at (3.8433, 26.84975) {$b'_3$};

\node[draw=none] at (4.0433, 24.64975) {$a'_3$};

\node[draw=none] at (3.49315, 22.16655) {$c'$};

\node[draw=none] at (3.5879, 25.6702) {$3\ell$};

\node[draw=none] at (4.498, 24.092) {$3\ell$};

\node[draw=none] at (1.78865, 27.3484) {$\ell$};

\node[draw=none] at (1.6677, 25.81755) {$2\ell$};

\node[draw=none] at (2.0254, 23.4291) {$2\ell$};

\end{tikzpicture}

%% file: fig2.tex
\definecolor{c909090}{RGB}{144,144,144}
\definecolor{c0000b9}{RGB}{0,0,185}

\def \globalscale {1.000000}

\vspace{-20pt} 
\begin{tikzpicture}[y=1cm, x=1cm, yscale=\globalscale,xscale=\globalscale, every node/.append style={scale=\globalscale}, inner sep=0pt, outer sep=0pt]
  \path[draw=black,line width=0.0265cm] (0.9821, 24.2838) circle (0.2827cm);

  \path[draw=black,line width=0.0265cm] (2.3702, 22.1263) circle (0.2827cm);

  \path[draw=black,line width=0.0265cm] (6.6297, 23.2782) circle (0.2827cm);

  \path[draw=black,line width=0.0265cm] (1.4016, 26.4841) circle (0.2827cm);

  \path[draw=black,line width=0.0265cm] (5.7199, 26.0515) circle (0.2827cm);

  \path[draw=black,line width=0.0265cm] (7.0169, 26.0564) circle (0.2827cm);

  \path[draw=black,line width=0.0265cm] (2.0345, 24.31) circle (0.2827cm);

  \path[draw=black,line width=0.0265cm] (3.7559, 24.3182) circle (0.2827cm);

  \path[draw=black,line width=0.0265cm] (5.3166, 24.583) circle (0.2827cm);

  \path[draw=black,line width=0.0265cm] (6.369, 24.6093) circle (0.2827cm);

  \path[draw=black,line width=0.0265cm] (7.7333, 24.6324) circle (0.2827cm);

  \path[draw=black,line width=0.0265cm] (1.4015, 27.4056) circle (0.2827cm);

  \path[draw=black,line width=0.0265cm] (3.1343, 27.4504) circle (0.2827cm);

  \path[draw=black,line width=0.0265cm] (3.1294, 26.4593) circle (0.2827cm);

  \path[draw=c909090,line width=0.0233cm,dash pattern=on 0.0466cm off 0.0466cm] (2.2323, 26.4785) ellipse (1.5637cm and 0.4702cm);

  \path[draw=c909090,line width=0.0233cm,dash pattern=on 0.0466cm off 0.0466cm] (6.3784, 26.0376) ellipse (1.5637cm and 0.4702cm);

  \path[draw=c909090,line width=0.0229cm,dash pattern=on 0.0458cm off 0.0458cm] (2.3212, 24.2986) ellipse (2.1296cm and 0.5375cm);

  \path[draw=c909090,line width=0.0214cm,dash pattern=on 0.0428cm off 0.0428cm] (6.5949, 24.5946) ellipse (1.8551cm and 0.5382cm);

  \path[draw=black,line width=0.1cm] (1.4136, 26.2181) -- (1.0267, 24.5516);

  \path[draw=black,line width=0.0811cm] (1.3987, 26.7687) -- (1.3987, 27.1407);

  \path[draw=black,line width=0.0826cm] (3.1396, 26.7538) -- (3.1396, 27.1704);

  \path[draw=black,line width=0.1cm] (3.6604, 24.5665) -- (3.214, 26.2181);

  \path[draw=black,line width=0.1cm] (1.0713, 24.0308) -- (2.1129, 22.3048);

  \path[draw=black,line width=0.1cm] (3.6604, 24.0457) -- (2.5296, 22.3792);

  \path[draw=black,line width=0.1cm] (1.9939, 24.0308) -- (2.2766, 22.394);

  \path[draw=black,line width=0.1cm] (5.6394, 25.8015) -- (5.3865, 24.8492);

  \path[draw=black,line width=0.1cm] (7.1572, 25.8164) -- (7.5887, 24.879);

  \path[draw=black,line width=0.1cm] (5.4609, 24.3433) -- (6.3685, 23.4059);

  \path[draw=black,line width=0.1cm] (6.4727, 24.3433) -- (6.6215, 23.5398);

  \path[draw=black,line width=0.1cm] (7.6184, 24.4177) -- (6.8744, 23.4951);

  \path[draw=c0000b9,line width=0.1cm,dash pattern=on 0.4cm off 0.1cm] (6.7405, 25.9949).. controls (5.5948, 25.2063) and (3.467, 25.37) .. (3.467, 25.37);

  \path[draw=c0000b9,line width=0.1cm,dash pattern=on 0.4cm off 0.1cm] (5.5799, 26.2925).. controls (4.7169, 29.3131) and (1.2053, 25.3254) .. (1.2053, 25.3254);

  \path[draw=black,line width=0.1cm] (1.5475, 26.2628) -- (3.4372, 24.4177);

  \path[draw=black,line width=0.1cm] (5.9668, 25.8908) -- (7.4548, 24.745);

  \path[draw=c0000b9,line width=0.1cm,dash pattern=on 0.4cm off 0.1cm] (5.4311, 26.114).. controls (2.7379, 25.9652) and (2.5147, 25.3551) .. (2.5147, 25.3551);

  \path[draw=black,fill=white,line width=0.0265cm] (1.2201, 25.3402) circle (0.1637cm);

  \path[draw=black,fill=white,line width=0.0265cm] (2.4891, 25.3354) circle (0.1637cm);

  \path[draw=black,fill=white,line width=0.0265cm] (3.4401, 25.375) circle (0.1637cm);

\node[draw=none] at (1.3997, 26.4604) {$b'$};

\node[draw=none] at (2.2897, 26.4604) {$\dots$};

\node[draw=none] at (2.86435, 24.28375) {$\cdots$};

\node[draw=none] at (7.00835, 24.6037) {$\cdots$};

\node[draw=none] at (6.33875, 26.12145) {$\cdots$};

\node[draw=none] at (1.13485, 26.9491) {$\ell$};

\node[draw=none] at (1.00435, 25.8833) {$2\ell$};

\node[draw=none] at (0.77375, 24.9236) {$2\ell$};

\node[draw=none] at (0.97375, 24.2436) {$a'$};

\node[draw=none] at (3.73375, 24.2936) {$a'_{\star}$};

\node[draw=none] at (1.2885, 23.019) {$4\ell$};

\node[draw=none] at (2.3538, 22.18055) {$c'$};

\node[draw=none] at (6.6066, 23.3247) {$\bar{c}$};

\node[draw=none] at (7.6419, 25.4146) {$2\ell$};

\node[draw=none] at (7.44495, 23.73325) {$2\ell$};

\node[draw=none] at (4.2705, 27.62425) {$2\ell$};

\node[draw=none] at (4.26305, 26.22555) {$2\ell$};

\node[draw=none] at (4.4267, 25.1394) {$2\ell$};

\node[draw=none] at (8.77905, 24.5516) {$\bar{A}$};

\node[draw=none] at (8.36905, 26.0516) {$\bar{B}$};

\node[draw=none] at (0.35715, 26.5083) {$B'$};

\node[draw=none] at (-0.0372, 24.29865) {$A'$};

\node[draw=none] at (1.42845, 27.95315) {$u_{b',\ell}$};

\node[draw=none] at (5.7138, 26.0619) {$\bar{b}$};

\node[draw=none] at (5.3344, 24.58135) {$\bar{a}$};

\node[draw=none] at (7.7744, 24.58135) {$\bar{a}_{\star}$};








\node[draw=none] at (0.3847, 25.37745) {$v_{a',b',2\ell}$};

\end{tikzpicture}

\vspace{-5pt} 

%% file: fig3.tex
\definecolor{c0000b9}{RGB}{0,0,185}
\definecolor{cffffab}{RGB}{255,255,171}

\def \globalscale {1.000000}

\vspace{+15pt} 
\begin{tikzpicture}[y=1cm, x=1cm, yscale=\globalscale,xscale=\globalscale, every node/.append style={scale=\globalscale}, inner sep=0pt, outer sep=0pt]

  \path[draw=black,line width=0.0265cm] (1.0342, 29.1508) circle (0.2827cm);

  \path[draw=black,line width=0.0811cm] (1.0467, 28.5043) -- (1.0467, 28.8763);

  \path[draw=black,line width=0.0811cm] (1.0488, 28.5043) -- (1.0488, 28.8763);

  \path[draw=black,line width=0.0265cm] (1.0396, 28.2364) circle (0.2827cm);

  \path[draw=black,line width=0.0265cm] (2.6549, 27.2311) circle (0.2827cm);

  \path[draw=black,line width=0.0265cm] (2.647, 26.1116) circle (0.2827cm);

  \path[draw=black,line width=0.0265cm] (2.6518, 25.0186) circle (0.2827cm);

  \path[draw=black,line width=0.0265cm] (1.0396, 26.0805) circle (0.2827cm);

  \path[draw=black,line width=0.0265cm] (1.0578, 23.9275) circle (0.2827cm);

  \path[draw=black,fill=white,line width=0.1cm] (1.0416, 27.9591) -- (1.0416, 26.3521);

  \path[draw=black,fill=white,line width=0.1cm] (1.0499, 25.8291) -- (1.0499, 24.222);

  \path[draw=black,fill=white,line width=0.0265cm] (1.0359, 27.1184) circle (0.1637cm);

  \path[draw=black,fill=white,line width=0.1cm] (2.6337, 26.9472) -- (2.6337, 26.3967);

  \path[draw=black,fill=white,line width=0.1cm] (2.6486, 25.8313) -- (2.6337, 25.2956);

    \path[draw=c0000b9,fill=white,line width=0.1cm,dash pattern=on 0.2cm off 0.1cm] (1.1904, 27.1407) -- (2.3659, 27.1407);

\node[draw=none] at (1.04155, 29.14945) {$u_{\cdot,\ell}$};

\node[draw=none] at (0.5952, 28.6659) {$\ell$};

\node[draw=none] at (1.0565, 28.2269) {$B'$};

\node[draw=none] at (0.5952, 27.54245) {$2\ell$};

\node[draw=none] at (0.5952, 26.61245) {$2\ell$};

\node[draw=none] at (0.5952, 24.9831) {$4\ell$};

\node[draw=none] at (1.07135, 26.08425) {$A'$};

\node[draw=none] at (1.0565, 23.9192) {$c'$};

\node[draw=none] at (1.7409, 27.3862) {$2\ell$};

\node[draw=none] at (2.6486, 27.2374) {$\bar B$};

\node[draw=none] at (3.17685, 26.672) {$2\ell$};

\node[draw=none] at (2.64115, 26.10655) {$\bar A$};

\node[draw=none] at (3.17685, 25.51875) {$2\ell$};

\node[draw=none] at (2.6635, 24.998) {$\bar c$};
\end{tikzpicture}

%% file: fig4.tex
\definecolor{c0066ff}{RGB}{0,102,255}
\definecolor{c0000b9}{RGB}{0,0,185}

\def \globalscale {1.000000}
\begin{tikzpicture}[y=1.5cm, x=1.5cm, yscale=\globalscale, xscale=\globalscale, every node/.append style={scale=\globalscale}, inner sep=0pt, outer sep=0pt]
  \path[draw=black,line width=0.1cm] (1.9302, 42.25475) -- (1.9302, 42.7783);
  \path[draw=black,fill=white,line width=0.1067cm] (1.93065, 41.9605) -- (1.93065, 37.79835);
  \path[draw=black,fill=white,line width=0.0217cm] (1.92165, 42.62655) circle (0.25995cm);
  \path[draw=black,fill=white,line width=0.0217cm] (1.9266, 42.10745) circle (0.25995cm);
  \path[draw=black,fill=white,line width=0.1cm] (1.93065, 37.4859) -- (1.93065, 36.0462);
  \path[draw=black,fill=c0066ff,line width=0.1cm] (6.73305, 41.4459) -- (2.63745, 41.4459);
  \path[draw=black,fill=c0066ff,line width=0.1cm] (6.7131, 40.20855) -- (4.00125, 40.20855);
  \path[draw=black,fill=c0066ff,line width=0.1cm] (6.71205, 39.5871) -- (4.65855, 39.5871);
  \path[draw=black,fill=c0066ff,line width=0.1cm] (6.702, 38.95215) -- (5.28465, 38.95215);
  \path[draw=black,fill=c0066ff,line width=0.1cm] (6.69945, 38.35065) -- (5.9295, 38.35065);
  \path[draw=black,fill=c0066ff,line width=0.1cm] (6.71565, 40.85715) -- (3.41235, 40.85715);
  \path[draw=black,fill=white,line width=0.0217cm] (2.61855, 41.47935) circle (0.25995cm);
  \path[draw=black,fill=white,line width=0.0217cm] (3.25725, 40.8804) circle (0.25995cm);
  \path[draw=black,fill=white,line width=0.0217cm] (3.9108, 40.2453) circle (0.25995cm);
  \path[draw=black,fill=white,line width=0.0236cm] (1.935, 41.451) circle (0.1392cm);
  \path[draw=black,fill=white,line width=0.0236cm] (1.93365, 40.8495) circle (0.1392cm);
  \path[draw=black,fill=white,line width=0.0236cm] (1.93245, 38.9844) circle (0.1392cm);
  \path[draw=black,fill=white,line width=0.0236cm] (1.92375, 38.37405) circle (0.1392cm);
  \path[draw=black,fill=white,line width=0.0236cm] (1.9287, 40.22955) circle (0.1392cm);
  \path[draw=black,fill=white,line width=0.0236cm] (1.94115, 39.60075) circle (0.1392cm);
  \path[draw=black,fill=white,line width=0.0217cm] (4.60155, 39.62145) circle (0.25995cm);
  \path[draw=black,fill=white,line width=0.0217cm] (5.2044, 38.96535) circle (0.25995cm);
  \path[draw=black,fill=white,line width=0.0217cm] (5.82825, 38.36385) circle (0.25995cm);
  \path[draw=black,fill=white,line width=0.0217cm] (6.7029, 41.48535) circle (0.25995cm);
  \path[draw=black,fill=white,line width=0.0217cm] (5.004, 43.1619) circle (0.25995cm);
  \path[draw=black,fill=white,line width=0.0217cm] (6.96765, 43.16355) circle (0.25995cm);
  \path[draw=black,fill=white,line width=0.0217cm] (1.9239, 37.6713) circle (0.25995cm);
  \path[draw=black,fill=white,line width=0.0236cm] (1.92105, 36.00585) circle (0.1917cm);
  \path[draw=black,fill=white,line width=0.0217cm] (6.7164, 40.8369) circle (0.25995cm);
  \path[draw=black,fill=white,line width=0.0217cm] (6.7101, 40.21065) circle (0.25995cm);
  \path[draw=black,fill=white,line width=0.0217cm] (6.70515, 39.5979) circle (0.25995cm);
  \path[draw=black,fill=white,line width=0.0217cm] (6.6891, 38.9505) circle (0.25995cm);
  \path[draw=black,fill=white,line width=0.0217cm] (6.70395, 38.35515) circle (0.25995cm);
  \path[draw=black,fill=white,line width=0.0236cm] (4.6482, 41.46705) circle (0.1392cm);
  \path[draw=black,fill=c0066ff,line width=0.1cm] (6.86115, 41.43645) -- (9.75375, 41.43645);
  \path[draw=black,fill=c0066ff,line width=0.1cm] (6.86115, 40.8447) -- (9.19155, 40.8447);
  \path[draw=black,fill=c0066ff,line width=0.1cm] (6.86115, 40.20855) -- (8.6463, 40.20855);
  \path[draw=black,fill=c0066ff,line width=0.1cm] (6.86115, 39.5871) -- (8.23455, 39.5871);
  \path[draw=black,fill=c0066ff,line width=0.1cm] (6.86115, 38.92485) -- (7.79295, 38.92485);
  \path[draw=black,fill=c0066ff,line width=0.1cm] (6.86115, 38.36665) -- (7.4892, 38.36665);
  \path[draw=c0000b9,fill=c0066ff,line width=0.1cm,dash pattern=on 0.3cm off 0.15cm] (4.95585, 42.94515) -- (4.6815, 41.58915);
  \path[draw=black,fill=c0066ff,line width=0.1cm] (5.18975, 43.155) -- (6.7735, 43.1439);
  \path[draw=black,fill=c0066ff,line width=0.1cm] (7.14035, 43.1601) -- (8.7241, 43.149);
  \path[draw=black,fill=white,line width=0.0236cm] (7.5354, 38.3619) circle (0.1917cm);
  \path[draw=black,fill=white,line width=0.0236cm] (7.8537, 38.9403) circle (0.1917cm);
  \path[draw=black,fill=white,line width=0.0236cm] (8.20965, 39.54675) circle (0.1917cm);
  \path[draw=black,fill=white,line width=0.0236cm] (8.61645, 40.20285) circle (0.1917cm);
  \path[draw=black,fill=white,line width=0.0236cm] (9.05805, 40.839) circle (0.1917cm);
  \path[draw=black,fill=white,line width=0.0236cm] (9.6855, 41.45175) circle (0.1917cm);
  \path[draw=black,fill=white,line width=0.0236cm] (8.77845, 43.1535) circle (0.1917cm);
  \path[draw=c0000b9,fill=c0066ff,line width=0.1cm,dash pattern=on 0.3cm off 0.15cm] (5.56875, 38.35635) -- (2.0646, 38.35635);
  \path[draw=c0000b9,fill=c0066ff,line width=0.1cm,dash pattern=on 0.3cm off 0.15cm] (4.9326, 38.97015) -- (2.0646, 38.97015);
  \path[draw=c0000b9,fill=c0066ff,line width=0.1cm,dash pattern=on 0.3cm off 0.15cm] (4.3077, 39.60615) -- (2.08695, 39.60615);
  \path[draw=c0000b9,fill=c0066ff,line width=0.1cm,dash pattern=on 0.3cm off 0.15cm] (3.66045, 40.21995) -- (2.10915, 40.21995);
  \path[draw=c0000b9,fill=c0066ff,line width=0.1cm,dash pattern=on 0.3cm off 0.15cm] (2.9796, 40.83375) -- (2.0757, 40.83375);
  \path[draw=c0000b9,fill=c0066ff,line width=0.1cm,dash pattern=on 0.3cm off 0.15cm] (2.36595, 41.4588) -- (2.0646, 41.4588);

  \node[draw=none] at (1.5866, 42.292375) {$\ell$};
  \node[draw=none] at (1.57735, 42.89315) {$u_{\cdot,\ell}$};
  \node[draw=none] at (1.9251, 42.105975) {$B'$};
  \node[draw=none] at (1.607025, 41.737725) {$2\ell$};
  \node[draw=none] at (1.584675, 41.168625) {$2\ell$};
  \node[draw=none] at (1.573575, 40.543575) {$2\ell$};
  \node[draw=none] at (1.59585, 39.90195) {$2\ell$};
  \node[draw=none] at (1.579125, 39.31605) {$2\ell$};
  \node[draw=none] at (1.56795, 38.668725) {$2\ell$};
  \node[draw=none] at (1.545675, 38.071725) {$2\ell$};
  \node[draw=none] at (1.91385, 37.68675) {$A'$};
  \node[draw=none] at (1.6014, 36.83295) {$14\ell$};
  \node[draw=none] at (1.913925, 35.996025) {$c'$};
  \node[draw=none] at (2.220825, 41.6931) {$2\ell$};
  \node[draw=none] at (3.57675, 41.69865) {$6\ell$};
  \node[draw=none] at (5.747325, 41.709825) {$6\ell$};
  \node[draw=none] at (8.2359, 41.704275) {$12\ell$};
  \node[draw=none] at (7.850925, 41.090475) {$10\ell$};
  \node[draw=none] at (5.03865, 41.05695) {$10\ell$};
  \node[draw=none] at (2.511, 41.034675) {$4\ell$};
  \node[draw=none] at (2.700675, 40.426425) {$6\ell$};
  \node[draw=none] at (5.30085, 40.415225) {$8\ell$};
  \node[draw=none] at (7.689075, 40.398525) {$8\ell$};
  \node[draw=none] at (3.11355, 39.80145) {$8\ell$};
  \node[draw=none] at (5.62455, 39.779175) {$6\ell$};
  \node[draw=none] at (7.549575, 39.795975) {$6\ell$};
  \node[draw=none] at (3.572225, 39.187425) {$10\ell$};
  \node[draw=none] at (3.822225, 38.580375) {$12\ell$};
  \node[draw=none] at (5.83835, 39.180375) {$4\ell$};
  \node[draw=none] at (7.23425, 39.196275) {$4\ell$};
  \node[draw=none] at (6.23835, 38.580375) {$2\ell$};
  \node[draw=none] at (7.13425, 38.596275) {$2\ell$};
  \node[draw=none] at (5.166975, 42.318075) {$6\ell$};
  \node[draw=none] at (6.05985, 43.422825) {$6\ell$};
  \node[draw=none] at (7.962525, 43.44475) {$6\ell$};
  \node[draw=none] at (2.6226, 41.4978) {$B''$};
  \node[draw=none] at (6.701475, 41.5257) {$A''$};
  \node[draw=none] at (10.02235, 41.5146) {$c''$};
  \node[draw=none] at (4.994025, 43.194075) {$B'''$};
  \node[draw=none] at (6.9693, 43.194075) {$A'''$};
  \node[draw=none] at (9.118025, 43.143825) {$c'''$};
\end{tikzpicture}

%% file: 3_algorithms.tex
\newcommand{\step}{\textsc{Step}}

\section{Improved Approximation Algorithms for \texorpdfstring{$k$}{k}-Center}
We now turn to our upper bounds and show that if we allow a small additive error, we are able to break the 2-approximation barrier for $k$-center. In fact, we achieve a $<2$ multiplicative approximation for any unweighted graph with $k$-radius $>1$. We begin by showing an approximation algorithm for $2$-center that obtains a multiplicative error of $5/3$ and additive error of $\leq 2/3$. We then construct a general algorithm for any $k$ that achieves a $(3/2, 1/2)$-approximation. Next we construct a faster algorithm that achieves an approximation of $\left(2 - \frac{1}{2k-1}, 1 - \frac{1}{2k-1}\right)$. We then show how to interpolate between these two algorithms and obtain a $\left(2 - \frac{1}{2\ell}, 1 - \frac{1}{2\ell}\right)$ approximation for any $\ell < k$. We conclude with an improved algorithm for $3$-center, which allows for bounded integer edge weights.

\subsection{Warm Up - \texorpdfstring{$5/3$}{5/3} Approximation for \texorpdfstring{$2$}{2}-Center}
\label{sec:algo53warmup}
We begin by presenting the first $<2$ approximation algorithm for $k$-center running in $o(n^k)$ time (in sparse graphs), a $5/3$ approximation algorithm for $2$-center running in $O(mn^{\omega/3})$ time, which incurs an additive error when the radius is not divisible by $3$. Our algorithm runs in two steps, both of which we generalize later to use in $k$-center approximation for any $k$. We prove the following statement.

\begin{theorem}\label{thm:53approx2center}
  There exists an algorithm running in $\tilde{O}(mn^{\omega/3})$ that computes a $(5/3,2/3)$-approximation to the $2$-center of any undirected, unweighted graph, w.h.p. If the $2$-center radius is divisible by $3$, the algorithm gives a true multiplicative $5/3$-approximation.
\end{theorem}

\begin{proof}
    We construct an algorithm that is given an integer $R$ that is divisible by $3$ and determines whether $R_2(G) > R$ or $R_2(G) \leq \frac{5R}{3}$. 
    By performing a binary search for the correct value of $R$, we obtain a multiplicative $5/3$ approximation if the $2$-radius of $G$ is divisible by $3$. If $R$ is not divisible by $3$, the following algorithm can instead determine if $R_2(G) > R$ or $R_2(G) \leq 2R - \floor{\frac{R}{3}}$. This gives an additional additive error of at most $\frac{2}{3}$. For the sake of readability, we will only prove the case when $R$ is divisible by $3$, and note that by replacing $\frac{2R}{3}$ with either $R - \floor{\frac{R}{3}}$ or $2\floor{\frac{R}{3}}$ we obtain the general case.

    Assume $R_2(G) \leq R$, in this case we show that we are able to find a pair of points that cover all vertices with radius $\frac{5R}{3}$. If we are unable to find such a pair, we output $R_2(G) > R$.

    Let $c_1, c_2$ be the optimal $2$-centers, so that for every $v\in V$, either  $d(c_1, v) \leq R$ or $d(c_2, v) \leq R$. Set $\delta = (2\omega - 3) / (3\omega - 3)$  and sample a random subset $S$ of $V$ of size $O(n^{1-\delta} \log n)$. In $\tilde{O}(mn^{1-\delta})$ time compute all distances between every $s\in S$ and $v\in V$. 

    Let $w$ be the furthest node from the set $S$. Run BFS from $w$ and let $W$ be the closest $n^\delta$ nodes to $w$. In $\tilde{O}(mn^\delta)$ time compute the distances between all $s\in W$ and $v\in V$. Note that $\delta < 1-\delta < \omega / 3$, and thus so far we have spent $\tilde{O}(mn^\delta + mn^{1-\delta}) \leq \tilde{O}(mn^{\omega/ 3})$ time.

    \begin{lemma} \label{lm:step1}
        Either there are two nodes $s_1, s_2\in S$ that cover the graph with radius $\frac{5R}{3}$, or w.h.p for some $s_3\in W$ we have $d(s_3, c_i) \leq \frac{R}{3}$ for $i=1$ or $i=2$.
    \end{lemma}
    
    \begin{proof}
        If $\exists s_1, s_2 \in S$ such that $d(s_1, c_1) \leq \frac{2R}{3}$ and $d(s_2, c_2)\leq \frac{2R}{3}$, then any vertex $v\in V$ will have either $d(s_1, v) \leq \frac{5R}{3}$ or $d(s_2, v) \leq \frac{5R}{3}$.
        
        Otherwise, for some $c_j$, $j=1,2$, $d(s,c_j)> \frac{2R}{3}$ for all $s\in S$. As $w$ is picked to be the furthest node from $S$ we have that $d(w,S)> \frac{2R}{3}$.
        Since $S$ is large enough, w.h.p it hits the closest $n^\delta$ nodes to every vertex in the graph. Thus in particular $S\cap W\neq \emptyset$. Let $s\in S\cap W$. Since $d(w,s)>\frac{2R}{3}$ and every node closer to $w$ than $s$ is in $W$, $W$ must contain all nodes at distance $\le \frac{2R}{3}$ from $w$. Next note that for one of $i=1$ or $i=2$, $d(w,c_i)\leq R$.
        Let $s_3$ be the node on the shortest path between $w$ and $c_i$ at distance exactly $2R/3$ to $w$, $d(w,s_3) = \frac{2R}{3}, d(s_3, c_i) \leq \frac{R}{3}$. By the above observation, $s_3$ must be in $W$ and is at distance $\leq R/3$ from $c_i$, as desired.
    \end{proof}
    
    We begin by handling the first case, when there exist $s_1, s_2\in S$ that cover all vertices with radius $\frac{5R}{3}$.
    Build an $|S|$ by $|V|$ matrix $X$
    such that for any $s\in S$ and $v\in V$, $X[s,v]=0$ if $d(s,v)\leq 5R/3$, and $X[s,v]=1$ otherwise. Multiply $X$ by $X^t$. The running time exponent for this product is 
    $$\omega(1-\delta,1,1-\delta)\leq \delta +\omega\cdot(1-\delta)=\omega-\delta(\omega-1)=\omega-(2\omega-3)/3=1+\omega/3.$$
    
    By construction, $XX^t[s,s']=0$ if and only if $s$ and $s'$ together cover all vertices of $V$ at distance $5R/3$. Therefore $XX^t[s_1, s_2] = 0$, and so we are able to find this pair.

    We are left to handle the second case. We can assume that case 1 did not work, and so there is a node $s_1\in W$, s.t. $d(s_1,c_1)\leq R/3$ (up to relabeling $c_1, c_2$). Set $\gamma=(\omega^2-3\omega+3)/(3\omega-3)$ and let $T$ be a random subset of $V$ of size $O(n^{1-\gamma}\log n)$. Note that this choice of $\gamma, \delta$ gives:
     
    $$\gamma+\delta=(\omega^2-3\omega+3+2\omega-3)/(3\omega-3)=\omega/3,$$
    and
    $$1-\gamma\leq (1-\gamma)+\delta(\omega-2)=\frac{3\omega-3-\omega^2+3\omega-3+2\omega^2-4\omega-3\omega+6}{3\omega-3}=\frac{\omega^2-\omega}{3\omega-3}=\omega/3.$$
    
    Run BFS from $T$ to compute all distances between $T$ and $V$ in time $O(mn^{1-\gamma}\log n)\leq \tilde{O}(mn^{\omega/3})$. For every $s\in W$, define $U(s)$ to be the nodes $v\in V$ with $d(s, v) > \frac{4R}{3}$. Note that if $d(s_1, c_1) \leq \frac{R}{3}$, then all the nodes in $U(s_1)$ are at distance $>R$ from $c_1$, and hence must be covered by $c_2$ within distance $R$. Further define $w(s)$ to be the furthest node from $T$ in $U(s)$, and let $W(s)$ be the closest $n^\gamma$ nodes to $w(s)$. We again show that one of two cases holds.

    \begin{lemma}\label{lm:step2}
        Given $s_1$ such that $d(s_1, c_1) \leq \frac{R}{3}$, w.h.p either there exists a node $s_2\in W(s_1)$ such that $s_1, s_2$ cover all vertices with radius $\frac{5R}{3}$, or $Q\coloneqq \bigcap_{t\in T \cap U(s_1)} B(t, R) \neq \emptyset$ and for any $s'_2\in Q$, the pair $s_1, s'_2$ cover all vertices with radius $\frac{5R}{3}$.
    \end{lemma}

    \begin{proof}
        If $d(w(s_1), T) > \frac{R}{3}$, then by the same argument as we made in the proof of \cref{lm:step1}, since $T$ hits $W(s_1)$ w.h.p, all nodes of distance $\leq \frac{R}{3}$ from $w(s_1)$ will be contained in $W(s_1)$. As previously noted, $d(w(s_1), c_2) \leq R$ and so there exists a node $s_2$ on the shortest path from $w(s_1)$ to $c_2$ of distance $d(s_2, w(s_1)) = \frac{R}{3}$ and $d(s_2, c_2) \leq \frac{2R}{3}$. Therefore, $s_2\in W(s_1)$ and the nodes $s_1, s_2$ cover all vertices with radius $\frac{5R}{3}$. 

        Otherwise, any node $u\in U(s_1)$ has $d(u, T) \le \frac{R}{3}$. Since all nodes in $U(s_1)$ are covered by $c_2$, $c_2\in B(t, R)$ for all $t\in T \cap U(s_1)$. Therefore, $c_2\in Q$ so $Q\neq \emptyset$. 
        
        Let $s'_2\in Q$ and let $u\in V$. If $d(s_1, u) > \frac{4R}{3}$ then $u\in U(s_1)$ and so there exists a node $t\in T$ such that $d(t, u) \leq \frac{R}{3}$. If $d(s_1, t) \leq \frac{4R}{3}$, then $d(s_1, u) \leq \frac{5R}{3}$. Otherwise, $t\in T\cap U(s_1)$ and so $d(s'_2, t) \leq R$, in which case $d(s'_2, u) \leq \frac{4R}{3}$. We conclude that $s_1, s'_2$ cover all vertices with radius $\frac{5R}{3}$.
    \end{proof}
    
    To handle the first case of \cref{lm:step2}, for every $s_1\in W$ run BFS from all nodes in $W(s_1)$ in time $\tilde{O}(mn^{\gamma + \delta})\leq \tilde{O}(mn^{\omega / 3})$. For every $s_2\in W(s_1)$, check if the pair $s_1, s_2$ cover all vertices with radius $\frac{5R}{3}$, in total time $O(n^{1+\gamma + \delta})\leq O(n^{1+\omega / 3})$. If we are in the first case, for the $s_1\in W$ such that $d(s_1, c_1)\leq \frac{R}{3}$ we will find the appropriate $s_2$ and complete the algorithm.

    To handle the second case of \cref{lm:step2}, create an $|W|=\tilde{O}(n^{\delta})$ by $|T|=\tilde{O}(n^{1-\gamma})$ Boolean matrix $A$ with rows indexed by $W$ and columns indexed by $T$.
    Define
    \[A[s,t]=1 \textrm{ iff } d(s,t)>4R/3.\]
    
    Next, create an  $\tilde{O}(n^{1-\gamma})$ by $n$ Boolean matrix $B$ with rows indexed by $T$ and columns indexed by $V$.
    Define
    \[B[t,v]=1 \textrm{ iff } d(t,v)>R.\]
    
    Multiply $A$ and $B$ in time $\tilde{O}(n^{\omega(\delta,1-\gamma,1)})$. Notice that by our choice of parameters we get
    $$\omega(\delta,1-\gamma,1)\leq (1-\gamma-\delta)+(1-\delta)+\omega\delta=2-\omega/3+\delta(\omega-1)=2-\omega/3+(2\omega-3)/3=1+\omega/3.$$

    Consider the $s_1$-row of $AB$, where $d(s_1, c_1) \leq \frac{R}{3}$. $AB[s_1, v] = 0$ if and only if $v\in Q$. Therefore, after computing $AB$, for every $s_1\in W$ we can look for  $AB[s_1, s'_2] = 0$ and then check if $s_1, s'_2$ cover all vertices with radius $\frac{5R}{3}$. As all the distances have already been computed, and we only need to check one pair for every $s_1\in W$, this step takes $O(|W|\cdot n)\leq \tilde{O}(n^{1+\omega / 3})$. 
    
    The total running time (as $m\geq n-1$) is within polylog factors 
    $mn^{\omega/3}$.
    If $\omega=2$, the running time is $\tilde{O}(mn^{2/3})$. If $\omega>2$, one can obtain some improvements using rectangular matrix multiplication. 
\end{proof}

We note that similar to Subsection \ref{sub:3center}, one can obtain an $(5/3,O(M))-$approximation algorithm running in time $\tilde{O}(mn^{\omega/3})$ for graphs with integer weights bounded by $M\leq \textrm{poly}(n)$.

\subsection{\texorpdfstring{$3/2$}{3/2} Approximation for \texorpdfstring{$k$}{k}-Center}
Next, for $k$-center for any $k$, 
we present a simple $3/2$ approximation algorithm
that incurs an additive error of $+\frac{1}{2}$ when the $k$-radius is odd. 

First we see a lemma we will use in many of our algorithms, which generalized \cref{lm:step1} to any value of $k$. Let $G$ be a graph with $k$-radius $\leq R$ achieved by centers $c_1, c_2, \ldots, c_k$. Let $S\subseteq V$ be a random sample of size $O(n^{1-\delta}\log n)$ and let $w$ be the furthest node in $V$ from $S$. Take $W$ to be the closest $n^\delta$ vertices to the vertex $w$.

\begin{lemma}\label{lm:step1gen}
    For any $r\in \N$, either there exist $k$ nodes $t_1, t_2, \ldots, t_k\in S$ that cover the graph with radius $R + r$, or w.h.p there exists some $s_1\in W$ such that $d(s_1, c_i) \leq R - r$ for some $i\in [k]$, w.l.o.g $i=1$.
\end{lemma}
\begin{proof}
    If for every center $d(c_i, S) \leq r$, denote the closest node to $c_i$ in $S$ by $t_i$. Now these $k$ nodes $t_1, t_2, \ldots, t_k$ cover the entire graph with radius $R + r$, as any node $v\in V$ has $d(v, c_i) \leq R$ for some $i\in [k]$, and thus $d(v, t_i) \leq R + r$.

    Otherwise, there exists a center such that $d(c_i, S) > r$ and therefore $d(w, S) > r$. W.h.p, the set $S$ hits
    the closest $n^\delta$ nodes to every vertex in the graph,
    and so $S$ hits $W$. Therefore, there exists a node $u\in W \cap S$ such that $d(w,u) > r$. By the definition of $W$, the set must contain all the nodes closer to $w$ than $u$ and so it contains all nodes of distance $\leq r$ from $w$.

    For some $i\in [k]$, $d(w,c_i) \leq R$. Consider the shortest path between $w$ and $c_i$. There exists a vertex $s_1$ on the path that has $d(w, s_1) = r$ and $d(s_1, c_i) \leq R-r$. By the above observation, we have that $s_1 \in W$.
\end{proof}

Using this lemma we can now prove our $3/2$ approximation algorithm.


\begin{theorem}\label{thm:32approxk}
    Given an unweighted, undirected graph $G$, there is an algorithm that computes a $(3/2, 1/2)$-approximation to the $k$-center w.h.p.\ and runs in time \begin{itemize}
    \item $\tilde{O}(n^{\omega+1/(\omega+1)})$ time for $k=3$,
    \item $\tilde{O}(n^{k-(3-\omega) + 1 / (k+1)})$ for $k\geq 4$, and 
    \item $n^{k-1 + 1 / (k+1)+o(1)}$ time for $k\geq 10$.
    \end{itemize}
\end{theorem}

\begin{proof}
    For any integer $R$ we will construct an algorithm that finds a set of $k$ points that cover all vertices with radius $\leq R + \ceil{\frac{R}{2}}$, if $R_k(G)\leq R$, thus determining if $R_k(G)> R$ or $R_k(G) \leq R + \ceil{\frac{R}{2}}$. By performing a binary search over the possible values of $R$, this algorithm achieves a $3/2$ approximation to the $k$-center if $R_k(G)$ is even. If $R_k(G)$ is odd, we obtain a $(3/2, 1/2)$-approximation. 
    
    Assume $R_k(G)\leq R$ and begin by computing APSP in $O(n^\omega)$ time via Seidel's algorithm~\cite{seidel}.
    Sample a vertex set $S$ of size $|S| = O(n^{1-\delta}\log n)$. Let $w$ be the furthest point from $S$ and take $W$ to be the closest $n^\delta$ nodes to $w$. By \cref{lm:step1gen}, either there exist $k$ points in $S$ that cover the graph with radius $R + \ceil{\frac{R}{2}}$ or there exists a point $s_1\in S$ such that $d(s_1, c_1) \leq R - \ceil{\frac{R}{2}} \leq  \frac{R}{2}$.

    To check if we are in the first case, construct an $|S|^{\ceil{k/2}}\times n$ matrix $A$ with rows indexed by $\ceil{k/2}$-tuples of points in $S$ and columns indexed by $V$, such that $A[(t_1, \ldots, t_{\ceil{k/2}}), v] = 0$ if $\exists i~d(v, t_i) \leq \ceil{\frac{3R}{2}}$ and $1$ otherwise. Similarly, construct an $n\times |S|^{\floor{k/2}}$ matrix $B$ such that $B[v, (t_1, \ldots, t_{\floor{k/2}})] = 0 \iff \exists i~d(v, t_i) \leq \ceil{\frac{3R}{2}}$. Multiply $A$ by $B$. By construction, $A B[(t_1, \ldots, t_{\ceil{k/2}}), (t'_1, \ldots, t'_{\floor{k/2}})] = 0$ if and only if the points $(t_1, \ldots, t_{\ceil{k/2}}, t'_1, \ldots, t'_{\floor{k/2}})$ cover all vertices of $V$ with radius $\ceil{\frac{3R}{2}}$. Thus, if we are in the first case, where there exist $t_1,\ldots, t_k\in S$ that cover $V$ with radius $\ceil{\frac{3R}{2}}$, we have $AB[(t_1, \ldots, t_{\ceil{k/2}}), (t_{\ceil{k/2} +1}, \ldots, t_k)] = 0$ we are able to find this $k$-tuple at this point.

    Otherwise, we have a point $s_1\in W$ such that $d(s_1, c_1) \leq \frac{R}{2}$, and hence $s_1$ covers all nodes with radius $\le \frac{3R}{2}$. We now search for the remaining $k-1$ centers. To do so we construct an $n^{\delta + \ceil{k/2} - 1} \times n$ matrix $C$ with rows indexed by  $W\times V^{\ceil{k/2} - 1}$ and columns indexed by $V$ where $C[(s_1, v_2, \ldots, v_{\ceil{k/2}}), v] = 0$ if $d(v, \{s_1, v_2, \ldots, v_{\ceil{k/2}}\}) \leq \frac{3R}{2}$ and 1 otherwise. Similarly define an $n\times n^{\floor{k/2}}$ matrix $D$ where $D[v, (v_1, \ldots, v_{\floor{k/2}})] = 0 \iff \exists i ~ d(v, v_i) \leq \frac{3R}{2}$. Multiply $C$ by $D$ and note that, as before, $CD[(s_1, v_2, \ldots, v_{\ceil{k/2}}), (v_{\ceil{k/2}+1}, \ldots, v_k)]=0$ if and only if the points $s_1, v_2, \ldots, v_k$ cover all vertices with radius $\leq \frac{3R}{2}$. Therefore, we are able to find such a $k$-tuple.

    The runtime of this algorithm is dominated by computing APSP  and the two matrix products. Assuming $k\geq 3$, we can bound the final runtime within polylogs by
\[n^\omega + \textrm{MM}\left(n^{(1-\delta)\ceil{k/2}}, n, n^{(1-\delta)\floor{k/2}}\right) + \textrm{MM}\left(n^{\delta + \ceil{k/2} - 1}, n, n^{\floor{k/2}}\right).\]

Let's consider the running time for $k=3$. We bound all rectangular matrix multiplication bounds by appropriate square matrix multiplications; we also set $\delta=\frac{1}{\omega+1}$:

\begin{align*}
& \tilde{O}\left(n^\omega + \textrm{MM}\left(n^{2(1-\delta)}, n, n^{1-\delta}\right) + \textrm{MM}\left(n^{\delta + 1}, n, n\right)\right) \leq \\
& \tilde{O}\left(n\cdot \textrm{MM}\left(n^{1-\delta}, n^{1-\delta}, n^{1-\delta}\right) + n^{\delta}\cdot\textrm{MM}\left(n, n, n\right) \right)\leq \\
& \tilde{O}\left(n^{1+\omega(1-\delta)}+n^{\omega+\delta}\right) = \tilde{O}(n^{\omega+\frac{1}{\omega+1}}).
\end{align*}

We proceed similarly for $k\geq 4$ but setting $\delta = \frac{1}{1+k}$. We omit the $\tilde{O}$ below:

\begin{align*}
& n^\omega + \textrm{MM}\left(n^{(1-\delta)\ceil{k/2}}, n, n^{(1-\delta)\floor{k/2}}\right) + \textrm{MM}\left(n^{\delta + \ceil{k/2} - 1}, n, n^{\floor{k/2}}\right) \leq \\
  & n^\omega + n^{(1-\delta)\ceil{k/2} - 1}\cdot n^{(1-\delta)\floor{k/2} - 1}\cdot \textrm{MM}(n,n,n) + n^{\delta + \ceil{k/2} - 2}\cdot n^{\floor{k/2} - 1}\cdot \textrm{MM}(n,n,n)  = \\
 & n^\omega + n^{(1-\delta)k + \omega - 2} + n^{\delta + k +\omega -3}\leq  n^{k + \frac{1}{k+1} + \omega - 3}.
 \end{align*}


    If $\omega > 2$ we can improve on this runtime using the current fastest techniques for rectangular matrix multiplication. 
    We give the best improvement which holds for $k\geq 10$.
We again set $\delta = \frac{1}{1+k}$.
For $k\geq 10$ we can use the fact that for all $a<0.3213$, $\textrm{MM}(N,N^a,N)\leq N^{2+o(1)}$ (see \cite{VXXZ24}) to bound the running time. We omit $n^{o(1)}$ factors below:
\begin{align*}
   & n^\omega + \textrm{MM}\left(n^{(1-\delta)\ceil{k/2}}, n, n^{(1-\delta)\floor{k/2}}\right) + \textrm{MM}\left(n^{\delta + \ceil{k/2} - 1}, n, n^{\floor{k/2}}\right) \leq \\
   & n^\omega + n^{(1-\delta)(\ceil{k/2}-\floor{k/2})}\cdot
   \textrm{MM}\left(n^{(1-\delta)\floor{k/2}}, n, n^{(1-\delta)\floor{k/2}}\right) + n^{1-\delta+\ceil{k/2}-\floor{k/2}}\cdot \textrm{MM}\left(n^{\delta + \floor{k/2} - 1}, n, n^{\delta + \floor{k/2} - 1}\right) \leq \\
   & n^\omega + n^{(1-\delta)k}+n^{\delta-1+k} =  n^{k-1+\frac{1}{k+1}},
\end{align*}
where the last inequality holds because for $k\geq 10$, 
$0.3213(1-\delta)\floor{k/2}\geq 0.3213(1-1/(k+1))5> 1.5 (1-1/11)>1$ and 
$0.3213(\delta + \floor{k/2} - 1)\geq 0.3213\cdot 4>1.$
\end{proof}

\subsection{\texorpdfstring{$2-\frac{1}{2k-1}$}{2-1/(2k-1)} Approximation for \texorpdfstring{$k$}{k}-Center}
Next we show a faster algorithm for approximating $k$-center, with the approximation factor growing with $k$. Our algorithm will begin the same way as the $3/2$ approximation, by sampling a set $S$ that either contains $k$ points that can act as approximate centers, or finding a set $W$ which has a point $s_1\in W$ such that $d(s_1, c_1)$ is small. We then show that we can find $k-1$ centers in $S$ to cover all vertices along with $s_1$, or we find a point $s_2$ which is close to $c_2$. We repeat this process until we have $s_1, s_2, \ldots, s_k$ such that $d(s_i, c_i)$ is bounded for every $i$, and use these $k$ points as approximate centers. We first prove a combinatorial version of our algorithm, with a simple running time analysis. In the following section we show how to improve the runtime using fast matrix multiplication.


\begin{theorem} \label{thm:2kapproxkcenterwordy} 
    There exists an algorithm running in $\tilde{O}(mn + n^{k/2 + 1})$ time that, when given an integer $R$ and an unweighted, undirected graph $G$,  determines w.h.p one of the following 
    \begin{enumerate}
        \item $R_k(G) > R$,
        \item $R_k(G) \leq 2R - \floor{\frac{R}{2k-1}}$.
    \end{enumerate}
\end{theorem}

By binary searching over the possible values of $R$ we obtain the desired approximation. If $R_k(G)$ is divisible by $(2k-1)$, this algorithm gives a multiplicative $(2 - \frac{1}{2k-1})$ approximation. Otherwise, we incur an additive error of at most $+\left(1 - \frac{1}{2k-1}\right)$. This gives us our desired result:

\begin{theorem}\label{thm:2kapproxkcenter}
    There exists an algorithm running in $\tilde{O}(mn + n^{k/2 + 1})$ that computes a $\left(2 - \frac{1}{2k-1}, 1 - \frac{1}{2k-1}\right)$-approximation to $k$-center.
\end{theorem}

\begin{proofof}{Proof of \cref{thm:2kapproxkcenterwordy}}
    Denote by $\alpha = \floor{\frac{R}{2k-1}}$. Assume $R_k(G) \leq R$, we will show in this case that we are able to find a set of $k$ points that cover all vertices with radius $2R- \alpha$. If we are unable to do so we determine that $R_k(G) > R$. As with the $3/2$ approximation, we begin by computing APSP in time $O(mn)$ and then randomly sampling a set $S$ of size $|S| = O(\sqrt{n} \log n)$. Let $w_1$ be the furthest node from $S$ and let $W_1$ be the closest $\sqrt{n}$ vertices to $w_1$.

    By \cref{lm:step1gen}, either there exist $k$ points in $S$ that cover all vertices with radius $2R - \alpha$, or there exists a point $s_1\in W_1$ with $d(s_1, c_1) \leq \alpha$. To check if we are in the first case, check for every every $k$-tuple in $S^k$ whether it covers all vertices with radius $2R - \alpha$, doing so takes $n\cdot |S|^k = \tilde{O}(n^{k/2 + 1})$ time.

    If we have not finished, we can assume $\exists s_1\in W_1~d(s_1, c_1) \leq \alpha$. Fix a vertex $s_1\in W_1$ and define $U = B(s_1, R + \alpha)^c$, the points of distance $>R+\alpha$ from $s_1$, and $Y = B(s_1, 2R - \alpha)^c$. If we have the correct $s_1$ (i.e. $d(s_1, c_1) \le \alpha$), then none of the points in $U$ are covered by $c_1$. Let $w_2$ be the furthest point in $U$ from $S$.
    In particular, there exists a center $c_i$ for $i\neq 1$ such that $d(c_i, w_2) \leq R$, w.l.o.g $i=2$.
    Let $W_2$ be the $\sqrt{n}$ closest nodes to $w_2$. If $d(w_2, S) > R - 3\alpha$, then $W_2$ contains all the vertices of distance $\leq R - 3\alpha$ from $w_2$, and therefore, there exists a point $s_2\in W_2$ on the shortest path from $w_2$ to $c_2$ such that $d(w_2, s_2) = R - 3\alpha$ and $d(s_2, c_2) \leq 3\alpha$. Otherwise, we claim that $k-1$ points in $S$ can cover $Y$ with radius $2R - \alpha$. More generally, we show the following key lemma.

    \begin{lemma}\label{lm:keylemma}
        Let $1\leq i < k$, let $r, \alpha \in \Z$ be such that $r +\alpha \leq R$, and let $C_i = \{s_1, \ldots, s_i\}$ be a set of $i$ points where $d(s_j, c_j) \leq r$ for every $1\le j\le i$. Define $U = B(C_i, R + r)^c, Y = B(C_i, 2R -\alpha)^c$. Let $w_{i+1}$ be the furthest point in $U$ from $S$ and let $W_{i+1}$ be the closest $\sqrt{n}$ nodes to $w_{i+1}$. One of the following two holds.
        \begin{enumerate}
            \item $\exists s_{i+1}\in W_{i+1}$ such that $d(s_{i+1}, c_j) \le  r + 2\alpha$ for some $j > i$, w.l.o.g $j = i+1$.
            \item There exist $k-i$ points in $S$ that cover $Y$ with radius $2R - \alpha$.
        \end{enumerate}
    \end{lemma}

    \begin{proof}
        First consider the case when $d(w_{i+1}, S) > R - r - 2\alpha$. By the argument we have shown before, since $S$ hits $W_{i+1}$ and $w_{i+1}$ is covered by a center $c_j$ for some $j > i$, there exists a point $s_{i+1}\in W_{i+1}$ on the shortest path from $w_{i+1}$ to $c_j$ that has $d(w_{i+1}, s_{i+1}) = R - r - 2\alpha$ and $d(s_{i+1}, c_j) \leq r + 2\alpha$.

        Otherwise, any point $u\in U$ has $d(u, S)\leq R - r - 2\alpha$. We claim that in this case, any center $c_j$ that has $d(c_j, y) \leq R$ for some point  $y\in Y$, must have $d(c_j, S) \leq R - \alpha$.
        
        If $c_j \in U$, then $d(c_j, S) \leq R- r - 2\alpha \le R - \alpha$ and we are done. Otherwise, fix a shortest path from $c_j$ to $y$, let $x$ be the first vertex on the path that belongs to $U$ (it must exist since $y\in Y \subseteq U$). 
        Since $x\in U$ we know $d(x, S)\leq R - r - 2\alpha$. Thus, it suffices to show that $d(c_j, x) \leq \alpha + r$, in which case we would have
         \[
         d(c_j, S) \leq d(c_j, x) + d(x, S) \leq \alpha + r + R - r - 2\alpha = R - \alpha.
         \]

        Let $p$ be the predecessor of $x$ on the shortest path from $c_j$ to $y$, see \cref{fig:keylemma}. By our choice of $x$ we know $p\notin U$ and so $d(p, s_\ell) \leq R + r$ for some $1\leq \ell \leq i$, by the definition of $U$. Then, since $d(s_\ell, y) > 2R - \alpha$, we have by the triangle inequality,
        \[
        d(p, y) \geq d(s_\ell, y) - d(p, s_\ell) > 2R - \alpha - (R+r) = R - r - \alpha.
        \]

        \begin{figure}
            \centering
            \includegraphics[width=0.50\linewidth]{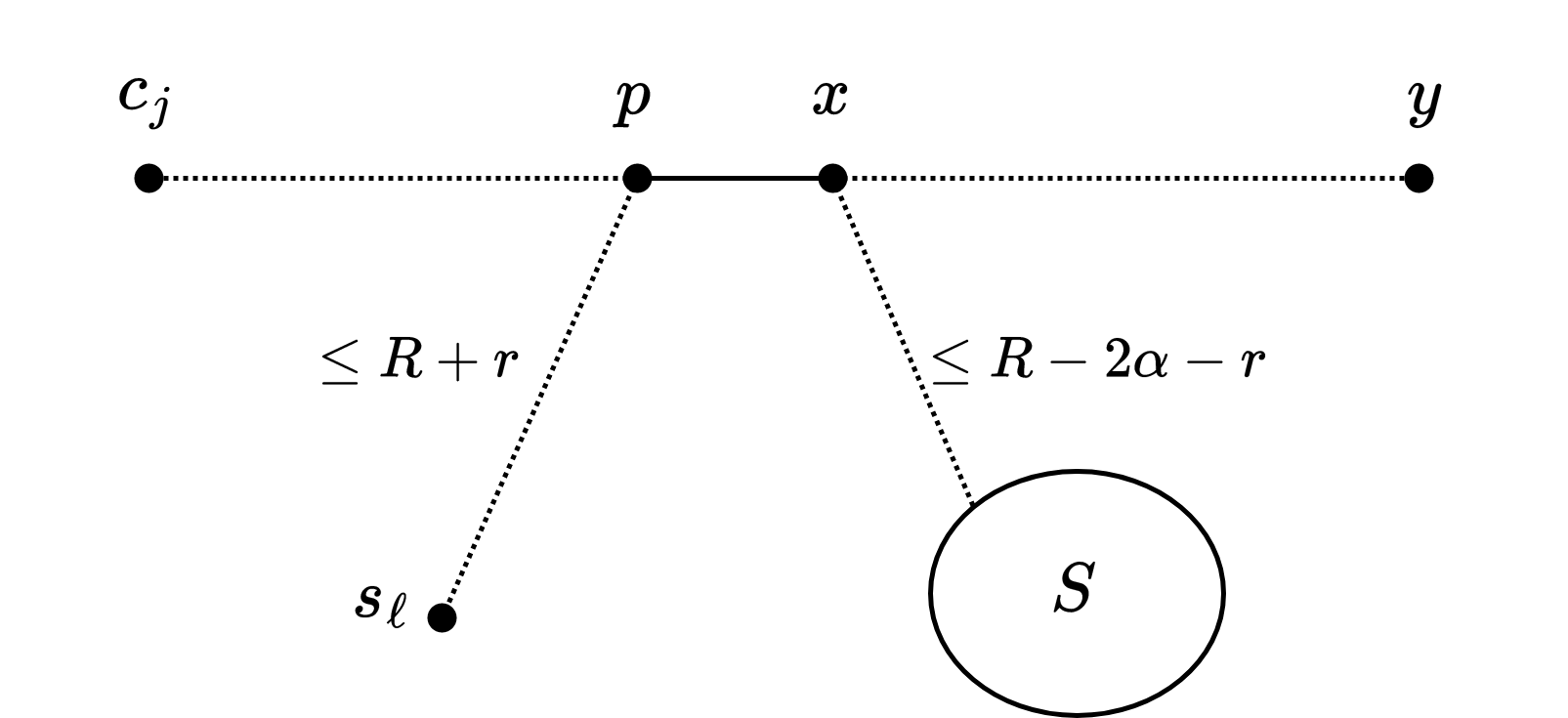}
            \caption{Shortest Path from $c_j$ to $y$}
            \label{fig:keylemma}
        \end{figure}
        
        Now, since $d(c_j, y)\leq R$ we conclude that $d(c_j, p) < r + \alpha$ and so $d(c_j, x) < r + \alpha + 1$. Since $d(c_j, x), r$ and $\alpha$ are all integers we can conclude that $d(c_j, x)\leq r + \alpha$. 
        
        Therefore, since the points in $Y$ are covered by $k-i$ centers $c_{i+1}, \ldots, c_k$, there exist $k-i$ points in $S$, each within distance $R-\alpha $ of the appropriate center, that cover $Y$ with radius $2R - \alpha$.
    \end{proof}

    Using this lemma, we perform $k-1$ iterations of the following algorithm. At step $i+1$ we have defined $W_1, \ldots, W_i$ and are guaranteed to have $s_1\in W_1, \ldots, s_i\in W_i$ such that $d(s_j, c_j) \le (2j - 1)\alpha$ for every $j \leq i$. We check if there exists a $k$-tuple in $W_1\times W_2 \times \ldots \times W_i \times S^{k-i}$ that covers all vertices in $V$ with radius $2R - \alpha $ and if so, return this $k$-tuple. This takes us $\tilde{O}(n^{k/2+1})$ time. Otherwise, for every $i$-tuple in $W_1\times \ldots \times W_i$ define $w_{i+1}$ and $W_{i+1}$ as in \cref{lm:keylemma}. By \cref{lm:keylemma} we know that there exists $s_{i+1}\in W_{i+1}$ such that $d(s_{i+1}, c_{i+1}) \le (2i + 1)\alpha$ and so we can begin step $i+1$.

    If we were to perform this step $k$ times we would have that $d(s_k, c_k) \le (2k - 1)\alpha $. In order for $s_1, \ldots, s_k$ to cover all vertices in $V$ with radius $2R - \alpha$, we would need $\alpha \leq \frac{R}{2k}$, which would give us a $2 - \frac{1}{2k}$ approximation. Instead, we do something slightly different at the last step that allows us to get a $2 - \frac{1}{2k-1}$ approximation, using the approach from our $5/3$ approximation for $2$-center.

    At step $k-1$, we have $W_1, \ldots, W_{k-1}$ such that $\exists s_i \in W_i ~d(s_i, c_i) \leq (2k-3)\alpha$ for any $i<k$. For every $(k-1)$-tuple $\tau \in W_1 \times \ldots \times W_{k-1}$ define $U^\tau, Y^\tau, w_k^\tau, W_k^\tau$ as in \cref{lm:keylemma}, using $r = (2k-3)\alpha$. We use the following lemma, generalizing \cref{lm:step2}.

    \begin{lemma}\label{lm:laststep}
        Given $\tau = (s_1, \ldots, s_{k-1})$ such that $d(s_i, c_i) \leq (2k-3)\alpha$ for every $i<k$, one of the following two holds.
        \begin{enumerate}
            \item $\exists s_k \in W_k^\tau$ such that $d(s_k, c_k) \le (2k-2)\alpha$.
            \item $Q\coloneqq \bigcap_{t\in S \cap U^\tau} B(t, R)\neq \emptyset$, and any point $s'_k \in Q$ covers all points in $Y^\tau$ with radius $2R - \alpha$.
        \end{enumerate}
    \end{lemma}

    We will prove a slightly more general statement, which we will use in other contexts later on in the paper. Setting $r = (2k - 3)\alpha, \delta = \frac{1}{2}$ in the following lemma gives us \cref{lm:laststep}.

    \begin{lemma}\label{lm:laststepgen}
        Let $r, \alpha\in \Z$ be such that $r+\alpha\leq R$, let $S$ be a random subset of vertices of size $O(n^{1-\delta}\log n)$ and $C =\{ s_1, \ldots, s_{k-1}\}$ be nodes such that $d(s_i, c_i) \leq r$ for every $i < k$. Define $U = B(C, R + r)^c$ and $Y = B(C, 2R - \alpha)^c$.
        Let $w$ be the furthest point in $U$ from $S$ and let $W$ be the closest $n^\delta$ vertices to $w$. One of the following two holds.
        \begin{enumerate}
            \item $\exists s_k \in W$ such that $d(s_k, c_k) \le r + \alpha$.
            \item $Q\coloneqq \bigcap_{t\in S \cap U} B(t, R)\neq \emptyset$, and any point $s'_k \in Q$ covers all points in $Y$ with radius $2R - \alpha$.
        \end{enumerate}
    \end{lemma}

    \begin{proof}
        If $d(w, S) > R - (r +\alpha)$, then by the argument we have seen in previous proofs we have that $\exists s_k \in W$ such that $d(s_k, c_k) \le r + \alpha$. Otherwise, any point $u\in U$ has $d(u, S) \leq R - (r+\alpha)$. 

        Any point $t\in S \cap U$ is not covered by $c_1, \ldots, c_{k-1}$ and must be covered by $c_k$. Thus, $c_k\in B(t, R)$ for any $t\in S \cap U$ and $c_k \in Q$, meaning $Q\neq \emptyset$. 
        
        Let $s'_k$ be an arbitrary vertex in $Q$ and let $y\in Y$. Let $t\in S$ be the closest vertex to $y$ in $S$. 
        Since $Y \subseteq U$, $d(t,y) \leq R - (r+\alpha) $. If $t\notin U$, then there exists some $s_i$ such that $d(s_i, t) \leq R + r$ and hence $d(s_i, y) \leq R + r + R - (r+\alpha) = 2R - \alpha$. However, this would imply $y\notin Y$ and so we conclude $t\in U$. Therefore, by definition of $Q$, $d(s'_k, t) \leq R$, and so $d(s'_k, y) \leq R + R - (r + \alpha) \leq 2R - \alpha$.
    \end{proof}

    We can now finish the algorithm. For every $\tau\in W_1\times \ldots \times W_{k-1}$, compute $W_k^\tau$ and check if any point in $s_k \in W_k^\tau$ covers all vertices in $Y^\tau$ within distance $R+(2k-2)\alpha \leq 2R - \alpha$. If so, return $(\tau, s_k)$. Otherwise, compute $Q$, and if it is not empty check if an arbitrary point in $s_k'\in Q$ covers all vertices in $Y^\tau$ within distance $2R- \alpha$. If so, return $(\tau, s'_k)$. For some $\tau\in W_1\times \ldots \times W_{k-1}$ we know that the conditions of \cref{lm:laststep} hold, and so we will be able to find the appropriate $s_k$ or $s_k'$ and complete the algorithm. 

    We describe the entire algorithm in Algorithm \ref{alg:2kapprox}.

\begin{algorithm}
         \caption{$\left(2-\frac{1}{2k-1}\right)$-Approximation Algorithm for $k$-Center}\label{alg:2kapprox}
         \begin{algorithmic}[1]
            \item \textbf{Input:} $G,R, k$ such that $R_k(G) \leq R$.
            \item \textbf{Output:} Vertices $s_1, s_2, \ldots, s_k$ that cover all $V$ with radius $\leq 2R -\alpha$, where $\alpha \coloneqq \floor{\frac{R}{2k-1}}$ 
            \State Compute APSP, randomly sample a subset $S\subset V$ of size $|S| = O(\sqrt{n}\log n)$.
            \State $\step(R, k, S, 0, ())$.\\
            \Procedure{Step}{$R, k, S, i, (s_1, \ldots, s_i)$}
                \State $w_{i+1} \leftarrow $ furthest vertex from $S$ in $B(\{s_1, s_2, \ldots, s_{i}\}, R + (2i - 1)\alpha)^c$.
                \State $W_{i+1}\leftarrow$ closest $\sqrt{n}$ vertices to $w_{i+1}$.
                \If{$i < k - 1$}
                    \If { $\exists(t_{i+1},\ldots,  t_k)\in S^{k-i}$ that covers $B(\{s_1, s_2, \ldots, s_{i}\}, 2R - \alpha)^c$ with radius $2R - \alpha$} \label{ln:algcheckcover}
                            \State \Return$s_1, \ldots, s_{i}, t_{i+1}, \ldots, t_k$.
                            \EndIf
                            \For{$s_{i+1} \in W_{i+1}$}
                                \State $\step(R, k, S, i+1, (s_1, \ldots, s_i, s_{i+1}))$.
                            \EndFor
                \Else \Comment{Last step, $i=k-1$}
                    \If { $\exists s_k\in W_k$ that covers $B(\{s_1, s_2, \ldots, s_{k-1}\}, 2R - \alpha)^c$ with radius $2R - \alpha$} \label{ln:laststep1}
                    \State \Return $s_1, \ldots, s_{k-1}, s_k$.
                    \EndIf
                    \State $Q \leftarrow \bigcap_{t\in S \cap B(\{s_1, \ldots, s_{k-1}\}, R + (2k-3)\alpha)^c} B(t, R)$,  $s'_k \leftarrow $ arbitrary point in $Q$. \label{ln:laststep2}
                    \If{$(s_1, \ldots, s_{k-1}, s'_k)$ cover all vertices with radius $\leq 2R - \alpha$} 
                        \State \Return $(s_1, \ldots, s_{k-1}, s'_k)$. 
                    \EndIf \label{ln:laststep2end}
                \EndIf
                
            \EndProcedure
            \end{algorithmic}
    \end{algorithm}   

    \paragraph{Runtime:} Denote by $T_i$ the runtime of the function $\step(R, k, S, i, (s_1, \ldots, s_i))$. The final runtime of Algorithm \ref{alg:2kapprox} is ${O}(mn) + T_0 = \tilde{O}(mn + n^{k/2+1})$, as we prove in the following claim.
    
    \begin{claim}
        $T_i = \tilde{O}(n^{(k-i)/2 + 1})$.
    \end{claim}
    \begin{proof}
        For $i< k-1$, the function $\step(R, k, S, i, (s_1, \ldots, s_i))$ checks for $|S|^{k-i}$ tuples whether they cover a set of points, taking $O(n^{(k-i)/2 + 1})$, and then runs $\step(R, k, S, i+1, (s_1, \ldots, s_{i+1}))$ on $|W_{i+1}| = \sqrt{n}$ different values of $s_{i+1}$. Therefore, 
        \[
        T_i = \tilde{O}(n^{(k-i)/2 + 1}) + \sqrt{n}\cdot T_{i+1}.
        \]
    
        We prove the claim by induction on $i$. For $i=k-1$, $\step(R, k, k-1, S, (s_1,\ldots, s_{k-1}))$ checks for $\sqrt{n}$ vertices if they cover a set of point, and then computes the intersection of $\leq |S|=\tilde{O}(\sqrt{n})$ balls and checks if a single $k$-tuple covers all points with a sufficiently small radius. Both these steps run in $\tilde{O}(n^{3/2})$. Now assume the claim is true for $i+1$, we have
        \begin{align*}
        T_i &= \tilde{O}(n^{(k-i)/2 + 1}) + \sqrt{n}\cdot T_{i+1} \\&= \tilde{O}(n^{(k-i)/2 + 1}) + \sqrt{n}\cdot \tilde{O}(n^{(k-i - 1)/2 + 1}) \\&= \tilde{O}(n^{(k-i)/2 + 1}).
        \end{align*}
    
    \end{proof}

\end{proofof}

\subsubsection{Improving the Runtime Using Fast Matrix Multiplication}

In this section we use fast matrix multiplication to speed up the running time of the algorithm presented in \cref{thm:2kapproxkcenter}. To introduce the techniques, we first assume in our analysis that $\omega = 2$. We then get rid of that assumption and prove \cref{thm:approxscheme}.

The approach will be analogous to the algorithm above, but we will produce a new random sample $S_i$ during each stage of the algorithm, instead of using the same set $S$ in all stages. The sample sizes $|S_i|$ and the sizes of the sets $W_i$ will vary (unlike before when they were $\sqrt n$). We then use fast matrix multiplication to check if there exists a set of points covering another set with sufficiently small radius.

We first prove the following result, which generalizes our earlier algorithm for $k=2$ in \cref{thm:53approx2center}. We still include the result of \cref{thm:53approx2center} as a special case of the following theorem.
\begin{theorem} \label{thm:2kapproxmatmult}
   If $\omega = 2$, there is an algorithm that computes w.h.p. a $\left(2 - \frac{1}{2k-1}, 1 - \frac{1}{2k-1}\right)$-approximation to the $k$-center of any unweighted, undirected graph, and runs in time 
    \begin{itemize}
        \item $\tilde{O}(mn^{2/3})$ for $k=2$, and
        \item $\tilde{O}(n^{k/2+22/9(k+1)})$ for $k\geq 3$. 
    \end{itemize}
\end{theorem}

\begin{proof}
    We begin by recursively defining a sequence of values $\delta_{k-1}, \delta_{k-2},\ldots, \delta_0$. For every $1\leq i\leq k$, let $t_{k-i} = 1 + \sum_{j=1}^i \delta_{k-j}$.
    \begin{itemize}
        \item $\delta_{k-1} = \delta_{k-2} = \frac{1}{3}$,
        \item $\delta_{k-3} = \frac{4}{9}$, and
        \item For $4\leq i \leq k$, $\delta_{k-i} = \frac{i - t_{k-i+1}}{i + 1}$.
    \end{itemize}

    We note several properties.
    \begin{obs}\label{obs:deltas}
        For all $i\geq 4$, $i\cdot (1-\delta_{k-i}) = \delta_{k-i} + t_{k-i+1} = t_{k-i}$.
    \end{obs}

    \begin{claim} \label{clm:2kruntimeomega2}
        For all $i\geq 3$, $t_{k-i} = \frac{i}{2} + \frac{22}{9(i+1)}$.
    \end{claim}
    We will prove a slightly more general claim.
    \begin{claim} \label{clm:generalrecursion}
        Let $\delta_0, \ldots, \delta_j, t_0, \ldots, t_{k-j}$ be such that $t_{k-i} = \delta_{k-i} + t_{k-i+1}$ and $\delta_{k-i} = \frac{i - t_{k-i+1} + c}{i+1}$ for some constant $c$ for all $i > j$. Then there exists a constant $a$ such that $t_{k-i} = \frac{i}{2} + \frac{a}{i+1} + c$ for all $i\geq j$.
    \end{claim}
    \begin{proof}
        By induction on $i$. For $i = j$ pick $a$ such that $t_{k-j} = \frac{j}{2} + \frac{a}{j+1} + c$. Now assume the claim holds for $i-1$.
        \begin{align*}
            t_{k-i} &= t_{k-i+1} + \delta_{k-i} = 
            t_{k-i+1} + \frac{i - t_{k-i+1} + c}{i + 1} \\&=
            \frac{i - 1}{2} + \frac{a}{i} + c + \frac{i - (\frac{i-1}{2} + \frac{a}{i} +c )- c}{i+1} = \\&=
            \frac{i-1}{2} + \frac{a}{i} + c + \frac{\frac{i+1}{2} - \frac{a}{i}}{i + 1}  \\&=
            \frac{i-1}{2} + \frac{1}{2} + \frac{a}{i} - \frac{a}{i(i+1)} + c = \frac{i}{2} + \frac{a}{i+1} + c.
        \end{align*}
    \end{proof}
    \begin{proofof}{Proof of \cref{clm:2kruntimeomega2}}
         For $i=3$, $t_{k-3} = 1 + \frac{1}{3} + \frac{1}{3} + \frac{4}{9} = \frac{19}{9} = \frac{3}{2} + \frac{22}{9\cdot 4}$. Therefore, by \cref{clm:generalrecursion}, for all $i\geq 3$, $t_{k-i} = \frac{i}{2} + \frac{22}{9(k+1)}$.
    \end{proofof}
    \begin{claim} \label{clm:deltalesshalf}
        Let $\delta_0, \ldots, \delta_j, t_0, \ldots, t_{k-j}$ be such that $t_{k-i} = \delta_{k-i} + t_{k-i+1}$ and $t_{k-i} = \frac{i}{2} + \frac{a}{i+1}$ for  $a \geq 0$ for all $i > j$. Then $\delta_{k-i}\leq \frac{1}{2}$ for all $i > j$.
    \end{claim}
    \begin{proof}
        $\delta_{k-i} = t_{k-i} - t_{k-i+1} = \frac{i}{2} + \frac{a}{i+1} - \frac{i-1}{2}-\frac{a}{i} = \frac{1}{2}  - \frac{a}{i(i+1)} \leq \frac{1}{2}$.
    \end{proof}
    We conclude that $\delta_{k-i}\leq \frac{1}{2}$ for all $1\leq i \leq k$.

    Given these values, we can now adjust Algorithm \ref{alg:2kapprox}. 
    For $k>2$, we begin by running APSP in $\tilde{O}(n^\omega) = \tilde{O}(n^2)$ time \cite{seidel}. For $k=2$ we avoid running APSP and instead run BFS from the necessarily points and sets throughout the algorithm.
    We adjust the subroutine $\step(R, k, S, i, (s_1, \ldots, s_i))$ and optimize the final step of the algorithm in a similar way to \cref{thm:53approx2center}. We show that the runtime of the resulting algorithm, when $\omega = 2$, is $\tilde{O}(n^{k/2+22/9(k+1)}).$

    First, instead of using the set $S$, the routine $\step(R, k, S, i, (s_1, \ldots, s_i))$ samples a set $S_i$ of size $|S_i| = O(n^{1-\delta_i}\log n)$ and uses this set in place of the set $S$. We will therefore omit the $S$ from calls to $\step()$ in the remainder of this section. We then define $W_{i+1}$ to be the $n^{\delta_i}$ closest vertices to $w_{i+1}$. W.h.p, the set $S_i$ hits $W_{i+1}$, so the resulting arguments still hold.

    Second, to check whether there exist $(t_{i+1}, \ldots, t_k)\in S_i^{k-i}$ that cover $Z\coloneqq B(\{s_1, \ldots, s_i\}, 2R-\alpha)^c$ (line \ref{ln:algcheckcover} in Algorithm \ref{alg:2kapprox}) we use fast matrix multiplication. Construct a $n^{(1-\delta_i)\cdot \floor{(k-i)/2}}\times |Z|$ boolean matrix $X$ whose rows are indexed by $\floor{(k-i)/2}$-tuples of vertices of $S_i$ and whose columns are indexed by $Z$. For a tuple $\tau$ and a vertex $u$, $X[\tau,u]=1$ if $d(\tau,u)>2R-\alpha$, and $X[\tau,u]=0$ otherwise. Similarly, we form a $|Z|$ by $\tilde{O}(n^{(1-\delta_i)\cdot \ceil{(k-i)/2}})$  matrix $Y$ whose rows are indexed by $Z$ and whose columns are indexed by $\ceil{(k-i)/2}$-tuples of nodes of $S_i$. For a tuple $\tau$ and a vertex $u$, $Y[u,\tau]=1$ if $d(\tau,u)>2R-\alpha$, and $Y[\tau,u]=0$ otherwise.
    
    After multiplying $X$ and $Y$ we can determine if a $(k-i)$-tuple of vertices of $S_i$ covers $Z$ within distance $\leq 2R - \alpha $ by finding a $0$ entry in $XY$.
    By Claim \ref{clm:deltalesshalf}, $1-\delta_i \geq \frac{1}{2}$. Therefore, the running time, when $\omega=2$ and $k-i\geq 4$, is up to polylogarithmic factors 
    
    \[
    \leq \textrm{MM}(n^{(1-\delta_i)\cdot \floor{(k-i)/2}}, n, n^{(1-\delta_i)\cdot \ceil{(k-i)/2}}) = n^{(1-\delta_i)(k-i)}.
    \]
    
    If we haven't found an approximate $k$-center at this point (for $i < k-1$), we run $\step(R, k, i+1, (s_1, \ldots, s_{i+1}))$ for $|W_{i+1}| = n^{\delta_i}$ different values of $s_{i+1}$. Thus, if we denote by $\hat{T}_i$ the running time of the subroutine $\step$ at stage $i$, we have for $i < k- 1$,
    \[
    \hat{T}_i = \tilde{O}(\textrm{MM}(n^{(1-\delta_i)\cdot \floor{(k-i)/2}}, n, n^{(1-\delta_i)\cdot \ceil{(k-i)/2}}) + n^{\delta_i}\cdot\hat{T}_{i+1}).
    \]
    
    If $i\leq k-4$ we have,
    \[
    \hat{T}_i = \tilde{O}(n^{(1-\delta_i)(k-i)} + n^{\delta_i}\cdot\hat{T}_{i+1}).
    \]

    Finally, we change the final steps of the algorithm for $i=k-1$ and $i=k-2$ using the technique from \cref{thm:53approx2center}. When calling $\step(R, k, k-1, (s_1, \ldots, s_{k-1}))$ we only check for a point $s_k\in W_k$ that completes the center and do nothing, i.e. we run line \ref{ln:laststep1} but not lines \ref{ln:laststep2}-\ref{ln:laststep2end}. Thus, $\hat{T}_{k-1} = \tilde{O}(n^{1+\delta_{k-1}}) = \tilde{O}(n^{4/3})$.

    When calling $\step(R, k, k-2, (s_1, \ldots, s_{k-2}))$ we add an additional step. Construct a $|W_{k-1}|\times |S_{k-1}|$ boolean matrix $A$ where $A[s_{k-1}, t] = 1 \iff d(s_{k-1}, t) > R + (2k-3)\alpha$ and a $|S_{k-1}|\times V$ boolean matrix $B$ where $B[t, v]=1 \iff d(t,v) > R$. Multiply them in time $n^{\omega(\delta_{k-2}, 1 - \delta_{k-1}, 1)}$. Now, for every $s_{k-1}\in W_{k-1}$, if there is a zero in the row corresponding to it $AB[s_{k-1}, u] = 0$, check if $u$ covers $B(\{s_1, \ldots, s_{k-1}\}, 2R -\alpha)^c$ with radius $2R - \alpha$. Otherwise, for every $s_{k-1} \in W_{k-1}$ run $\step(R, k, k-1, (s_1, \ldots, s_{k-1}))$. As we showed in the proof of \cref{thm:53approx2center}, this obtains the same result as our original version of Algorithm \cref{alg:2kapprox}.

    The new running time for $\step(R, k, k-2, (s_1, \ldots, s_{k-2}))$ now includes two matrix products and a recursive call to $\step(R, k, k-1, (s_1, \ldots, s_{k-1}))$. So we have, omitting polylogarithmic factors,
    \begin{align*}
    \hat{T}_{k-2} &= n^{\omega(1-\delta_{k-2}, 1, 1 - \delta_{k-2})} + n^{\omega(\delta_{k-2}, 1 - \delta_{k-1}, 1)} + n^{\delta_{k-2}}\hat{T}_{k-1} \\&= n^{2 - \delta_{k-2}} +n^{2 - \delta_{k-1}} + n^{1 + \delta_{k-1} + \delta_{k-2}} = n^{5/3}.
    \end{align*}

    Consider the running time for $k=2$. In this case, $\hat{T}_0 = \hat{T}_{k-2} = n^{5/3}$.
    In addition to running $\step(R, 2, 0, ())$ in $\hat{T}_0$ time, for $k=2$ we need to run BFS from all points in the sets $S_0$, $S_1$ in time $\tilde{O}(mn^{1 - \delta_0} + mn^{1-\delta_1}) \leq \tilde{O}(mn^{2/3})$. Additionally, we run BFS from all points in $W_1$ in time $O(mn^{\delta_0}) \leq \tilde{O}(mn^{2/3})$ during our call to $\step(R, 2, 0, ())$, and from all points in $W_2$ in every call to $\step(R, 2, 1, (s_1))$. This takes $O(mn^{\delta_1})$ and is called $n^{\delta_0}$ times, so in total runs in time ${O}(mn^{\delta_0 + \delta_1}) \leq \tilde{O}(mn^{2/3})$. Therefore, for $k=2$ our algorithm runs in time $\tilde{O}(mn^{2/3})$.

    We can now use the properties we proved about $\delta_0,\ldots, \delta_{k-1}$ to compute the runtime of our adjusted algorithm for $k > 2$.

    \begin{claim}
        For $1\leq i\leq k$, $\hat{T}_{k-i} = \tilde{O}(n^{t_{k-i}})$.
    \end{claim}
    \begin{proof}
        For $i = 1, 2$ we showed $\hat{T}_{k-1} = \tilde{O}(n^{4/3}) = \tilde{O}(n^{t_{k-1}})$ and $\hat{T}_{k-2} = \tilde{O}(n^{5/3}) = \tilde{O}(n^{t_{k-2}})$.

        For $i = 3$, since $1/2 > \delta_{k-3}>0$ we have $(1-\delta_{k-3})<1<2(1-\delta_{k-3})$ and so the matrix product step running time, 
        \[
            \textrm{MM}(n^{(1-\delta_{k-3})},n,n^{2(1-\delta_{k-3})}) = n^{3-2\delta_{k-3}}.
        \]
        
        Thus, within polylogarithmic factors
        \[
            \hat{T}_{k-3} = n^{3-2\delta_{k-3}} + n^{\delta_{k-3}}\cdot \hat{T}_{k-2} = n^{3-2\delta_{k-3}} + n^{\delta_{k-3} + t_{k-2}} = n^{19/9} = n^{t_{k-3}}.
        \]

        For $i\geq 4$, assume the claim holds for $i-1$, then by \cref{obs:deltas}, $$\hat{T}_{k-i} = \tilde{O}(n^{(1-\delta_{k-i})i} + n^{\delta_{k-i}}\cdot\hat{T}_{k-i+1}) = \tilde{O}(n^{(1-\delta_{k-i})i} + n^{\delta_{k-i} + t_{k-i+1}}) = \tilde{O}(n^{t_{k-i}}).$$
    \end{proof}
    
    The final runtime of the algorithm is dominated by running APSP in $\tilde{O}(n^\omega)$ time (and we assumed $\omega=2$) and calling $\step(R, k, 0, ())$. Therefore, when $k\geq 3$, our algorithm runs in time $\tilde{O}(n^2) + \hat{T}_0 = \tilde{O}(n^2) + \tilde{O}(n^{t_0}) =  \tilde{O}(n^{k/2 + 22/9(k+1)})$.
\end{proof}


We can now compute the runtime of our adjusted Algorithm \ref{alg:2kapprox} without the assumption that $\omega = 2$.

\begin{theorem} \label{thm:2kapproxmatmultgen}
   There is an algorithm that computes w.h.p. a $\left(2 - \frac{1}{2k-1}, 1 - \frac{1}{2k-1}\right)$-approximation to the $k$-center of any unweighted, undirected graph, and runs in time 
    \begin{itemize}
        \item $\tilde{O}(mn^{\omega/3})$ for $k=2$, 
        \item $\tilde{O}(n^{\frac{k}{2} + \frac{\beta_0}{k+1} + (\omega - 2)})$ for $3\leq k \leq 13$, where $\beta_0 \coloneqq \frac{-8\omega^2 + 18\omega + 18}{3\omega + 3}\approx 1.5516$, and
        \item $\tilde{O}(n^{\frac{k}{2} + \frac{\beta_1}{k+1}+o(1)})$ for $k\geq 13$, where $\beta_1 \coloneqq  \frac{34\omega^2 - 24 \omega - 66}{3\omega + 3}\approx 6.7533$.
    \end{itemize}
\end{theorem}

\begin{proof}
    We begin by defining slightly different values of $\delta_0, \ldots, \delta_{k-1}$. Let $t_{k-i} = 1 + \sum_{j=1}^i \delta_{k-j}$ as before.
    \begin{itemize}
        \item $\delta_{k-1} = \frac{2\omega - 3}{3\omega - 3}$,
        \item $\delta_{k-2} = \frac{\omega^2 - 3\omega +3}{3\omega - 3}$, and
        \item $\delta_{k-3} = \frac{2\omega}{3\omega + 3}$.
        \item For $4\leq i \leq 13$, let $\delta_{k-i} = \frac{i - t_{k-i+1} + (\omega - 2)}{i + 1}$.
        \item For $i\geq 14$, let $\delta_{k-i} = \frac{i - t_{k-i+1}}{i + 1}$.
    \end{itemize}

    We will again show that $\hat{T}_{i} = \tilde{O}(n^{t_{i}})$ and use the following two claims to compute our runtime. For the remainder of this proof we omit logarithmic factors.

    \begin{claim}\label{clm:fullruntime3}
        For $3\leq i \leq 13$, $t_{k-i} = \frac{i}{2} + \frac{-8\omega^2 + 18\omega + 18}{(3\omega + 3)(i+1)} + (\omega - 2).$
    \end{claim}
    Note that $\frac{-8\omega^2 + 18\omega + 18}{3\omega + 3} > 0$ (for any $2\leq \omega < 3$) and so \cref{clm:deltalesshalf} holds for $\delta_{k-4}, \ldots, \delta_{k-13}$.
    \begin{proof}
        For $i=3$, 
        \begin{align*}
          t_{k-3} = 1 + \frac{2\omega - 3}{3\omega - 3} + \frac{\omega^2 - 3\omega +3}{3\omega - 3} + \frac{2\omega}{3\omega + 3}  = \frac{\omega^2 + 6\omega + 3}{3\omega + 3} = \frac{3}{2} + \frac{-8\omega^2 + 18\omega + 18}{(3\omega + 3)\cdot 4} + (\omega - 2).
        \end{align*}
        Therefore, by \cref{clm:generalrecursion}, the claim holds for all $3\leq i \leq 13$.
    \end{proof}

    \begin{claim}\label{clm:fullruntime12}
        For $i \geq 13$, $t_{k-i} = \frac{i}{2} + \frac{34\omega^2 -24 \omega - 66}{(3\omega + 3)(i+1)}.$
    \end{claim}

    Again note that $\frac{34\omega^2 - 24\omega -66}{3\omega + 3} > 0$ (for any $2\leq \omega$) and so \cref{clm:deltalesshalf} holds for $\delta_{k-14}, \ldots, \delta_{0}$.
    
    \begin{proof}
        For $i = 13$, by \cref{clm:fullruntime3}, 
        \begin{align*}
          t_{k-13} = \frac{13}{2} + \frac{-8\omega^2 + 18\omega + 18}{(3\omega + 3)\cdot 14} + (\omega - 2) = \frac{13}{2} + \frac{34\omega^2 -24 \omega - 66}{(3\omega + 3 )\cdot 14}.
        \end{align*}

        Therefore, by \cref{clm:generalrecursion}, the claim holds for all $i \geq 13$.
    \end{proof}
    
    For $k=2$ these parameters are identical to the ones we chose in the proof of \cref{thm:53approx2center}, taking $\delta = \delta_{k-1}, \gamma = \delta_{k-2}$. Therefore, by the calculations shown there, 
    \begin{align*}
    \hat{T}_{k-2} &= n^{\omega(1-\delta_{k-2}, 1, 1 - \delta_{k-2})} + n^{\omega(\delta_{k-2}, 1 - \delta_{k-1}, 1)} + n^{1 + \delta_{k-1} + \delta_{k-2}} \leq n^{1 + \omega/3}.
    \end{align*}

    The running time for $k=2$ is $\tilde{O}(mn^{\omega/3})$, as shown in \cref{thm:53approx2center}. For $k>2$, as computed previously, 
    \[
    \hat{T}_{k-3} = n^{\omega(2 - 2\delta_{k-3}, 1, 1 - \delta_{k-3})} + n^{\delta_{k-3} + t_{k-2}}.
    \]
    Since $1-\delta_{k-3}\leq 1 \leq 2 - 2\delta_{k-3}$ we have \begin{align*}
        \omega(2 - 2\delta_{k-3}, 1, 1 - \delta_{k-3}) &\leq (1-\delta_{k-3}) + \delta_{k-3}  + (1-\delta_{k-3})\omega = 1 + \omega(1-\delta_{k-3}) \\&=
        1 + \omega\left(1 - \frac{2\omega}{3\omega + 3} \right) = \frac{\omega^2 + 6\omega + 3}{3\omega + 3} = t_{k-3}.
    \end{align*}
    Thus, 
    \[
    \hat{T}_{k-3} \leq n^{1 + \omega (1-\delta_{k-3})} + n^{\delta_{k-3} + t_{k-2}} = n^{t_{k-3}}.
    \]

    Now consider $k\geq 4$. By \cref{clm:deltalesshalf}, $1-\delta_{k-i}\geq \frac{1}{2}$ and so for any $4\leq i\leq 13$, $\floor{\frac{i}{2}}\cdot (1-\delta_{k-i})\geq 1$. Thus,
    \begin{align*}
        &\omega\left(\floor{\frac{i}{2}}\cdot (1-\delta_{k-i}), 1, \ceil{\frac{i}{2}}\cdot (1-\delta_{k-i})\right) \leq
        \floor{\frac{i}{2}}\cdot (1-\delta_{k-i})- 1 +  \ceil{\frac{i}{2}}\cdot (1-\delta_{k-i}) - 1 + \omega \\ =&~ i \cdot (1 - \delta_{k-i}) + (\omega - 2) =  t_{k-i+1} + \delta_{k-i} = t_{k-i}.
    \end{align*}
    Where the final equations hold by the definition of $\delta_{k-i}$ for $4\leq i \leq 13$. Therefore, for $4\leq i \leq 13$, assuming the claim for $i-1$, 
    \[
    \hat{T}_{k-i} = \textrm{MM}(n^{\floor{\frac{i}{2}}\cdot (1-\delta_{k-i})}, n, n^{\ceil{\frac{i}{2}}\cdot (1-\delta_{k-i})}) + n^{\delta_{k-i}}\cdot \hat{T}_{k-i+1} \leq n^{t_{k-i}} + n^{\delta_{k-i} + t_{k-i+1}} = n^{t_{k-i}}.
    \]

    Now consider $k\geq 14$, we again use the fact that for all $a < 0.3213$, $\textrm{MM}(N, N^a, N)\leq N^{2 + o(1)}$ \cite{VXXZ24}. For $i\geq 14$, $1-\delta_{k-i}\geq \frac{1}{2}$ and so $\floor{\frac{i}{2}}\cdot (1-\delta_{k-i})\geq \frac{1}{0.3213}$. Thus for any $i\geq 14$, omitting $n^{o(1)}$ factors,

    \begin{align*}
        &\omega\left(\floor{\frac{i}{2}}\cdot (1-\delta_{k-i}), 1, \ceil{\frac{i}{2}}\cdot (1-\delta_{k-i})\right) \leq
        \floor{\frac{i}{2}}\cdot (1-\delta_{k-i}) +  \ceil{\frac{i}{2}}\cdot (1-\delta_{k-i})  \\ =&~ i \cdot (1 - \delta_{k-i})  =  t_{k-i+1} + \delta_{k-i} = t_{k-i}.
    \end{align*}
    
    And so once again,
    \[
    \hat{T}_{k-i} = \textrm{MM}(n^{\floor{\frac{i}{2}}\cdot (1-\delta_{k-i})}, n, n^{\ceil{\frac{i}{2}}\cdot (1-\delta_{k-i})}) + n^{\delta_{k-i}}\cdot \hat{T}_{k-i+1} \leq n^{t_{k-i}} + n^{\delta_{k-i} + t_{k-i+1}} = n^{t_{k-i}}.
    \]

    We conclude that for all $0\leq i\leq k$, $\hat{T}_{i} = \tilde{O}(n^{t_i})$ and so $\hat{T}_0 = \tilde{O}(n^{t_0}) = \tilde{O}(n^{t_{k-k}})$. The final runtime of the algorithm for $k\geq 3$ is $\tilde{O}(n^\omega) + \hat{T}_0$. Therefore, from \cref{clm:fullruntime3} and \ref{clm:fullruntime12} we get our desired runtimes.
\end{proof}

\subsubsection{Obtaining a Runtime / Approximation Tradeoff}
In our proof of \cref{thm:2kapproxkcenter}, we introduced an algorithm running in $k$ steps that achieves a $2-\frac{1}{2k-1}$ approximation to the $k$-radius. In each step the algorithm finds an `approximate center', a point $s_i$ that is close to one of the real centers $c_i$. However, this distance grows with each approximate center. In this section we optimize the algorithm to stop after $\ell$ steps, obtaining a tradeoff between approximation and runtime. 

In the following theorem we demonstrate how to achieve this runtime/approximation tradeoff with a {combinatorial} algorithm. Similarly to \cref{thm:2kapproxmatmult}, we can speed up the algorithm using fast matrix multiplication. If $\omega = 2$ we get a runtime of $\tilde{O}(n^{k-\ell +  \frac{\ell(\ell + 1)}{2(k+1)}})$. Using the current bounds on $\omega$, if $\ell \leq k - 13$ we get $\tilde{O}(n^{k-\ell +  \frac{\ell(\ell + 1)}{2(k+1)}+o(1)})$. We omit the details for the algebraic optimization as they are similar to previous calculations.



\begin{theorem}\label{thm:2lapproxkcenter}
    For any integer $1\leq \ell \leq k$, there exists a combinatorial algorithm running in $\tilde{O}(mn + n^{k-\ell +  \frac{\ell(\ell + 1)}{2(k+1)}+1})$ time that computes a $\left( 2- \frac{1}{2\ell}, 1 - \frac{1}{2\ell}\right)$-approximation to the $k$-center of $G$ with high probability.
\end{theorem}


\begin{proof}
    Let $\alpha = \floor{\frac{R}{2\ell}}$, and $\gamma$ be a parameter to be set later. Fix $c_1,...,c_k$ to be a $k$-center of $G$. We assume that $R_k(G) \leq R$ and show how to find a set of $k$ points that cover all vertices with radius $\leq 2R - \alpha$. As before, this allows us to distinguish between $R_k(G)> R$ and $R_k(G) \leq 2R - \alpha$, and by binary searching over possible values of $R$ we obtain the desired approximation. We begin by computing APSP and running Algorithm \ref{alg:2kapprox} with different sized hitting sets, as described in \cref{thm:2kapproxmatmult}. However, after $\ell$ steps, when we have $s_1, \ldots, s_\ell$, we stop and try every $k-\ell$ tuple in $V^{k-\ell}$ to complete the approximate $k$-center. See Algorithm \ref{alg:2lapprox} for details.

    \begin{algorithm}
         \caption{$\left(2-\frac{1}{2\ell}\right)$-Approximation Algorithm for $k$-Center}\label{alg:2lapprox}
         \begin{algorithmic}[1]
            \item \textbf{Input:} $G,R, k$ such that $R_k(G) \leq R$, $\ell\leq k$, parameters $\bar{\delta} = (\delta_0, \ldots, \delta_{\ell -1})$.
            \item \textbf{Output:} Vertices $s_1, s_2, \ldots, s_k$ that cover all $V$ with radius $\leq 2R -\alpha$, where $\alpha \coloneqq \floor{\frac{R}{2\ell}}$.
            \State Compute APSP.
            \State \textsc{Step}$(R, k, \ell, \bar{\delta}, 0, ())$.\\
            \Procedure{Step}{$R, k, \ell, \bar{\delta}, i, (s_1, \ldots, s_i)$}
                \If{$i < \ell$}
                    \State Sample a set $S_i$ of size $n^{1-\delta_i}\log n$.
                    \State $w_{i+1} \leftarrow $ furthest vertex from $S_i$ in $B(\{s_1, s_2, \ldots, s_{i}\}, R + (2i - 1)\alpha)^c$.
                    \State $W_{i+1}\leftarrow$ closest $n^{\delta_i}$ vertices to $w_{i+1}$.
                    \If { $\exists(t_{i+1},\ldots,  t_k)\in S^{k-i}$ that covers $B(\{s_1, s_2, \ldots, s_{i}\}, 2R - \alpha)^c$ with radius $2R - \alpha$} 
                            \State \Return$s_1, \ldots, s_{i}, t_{i+1}, \ldots, t_k$.
                            \EndIf
                            \For{$s_{i+1} \in W_{i+1}$}
                                \State \textsc{Step}$(R, k, \ell, \bar{\delta}, i+1, (s_1, \ldots, s_i, s_{i+1}))$.
                            \EndFor
                \Else \Comment{Last step, $i=\ell$}
                    \If{$\exists(t_{\ell+1}, \ldots , t_k)\in V^{k - \ell }$ that covers $B(\{s_1, s_2, \ldots, s_{\ell}\}, 2R - \alpha)^c$ with radius $2R - \alpha$} \label{ln:ltuple}
                        \State \Return$s_1, \ldots, s_{\ell}, t_{\ell+1}, \ldots, t_k$. \label{ln:lapproxreturn}
                    \EndIf
                \EndIf               
            \EndProcedure
            \end{algorithmic}
    \end{algorithm}

    Note that \cref{lm:keylemma} holds for any $|S| = O(n^{1-\gamma}\log n)$ and $|W| = n^\gamma$. Therefore, by \cref{lm:step1gen} and \cref{lm:keylemma}, either we find a $k$-tuple that covers all $V$ with radius $2R - \alpha$, or we will have at least one call to line \ref{ln:ltuple} in which the points $s_1, \ldots, s_\ell$ satisfy $d(s_i, c_i) \leq (2\ell - 1)\alpha\leq 2R - \ell$ for every $1\leq i \leq \ell$. Therefore, there exist $k-\ell$ points that complete this $k$-tuple to an approximate $k$-center, in particular $c_{\ell+1}, \ldots, c_k$, and so line \ref{ln:lapproxreturn} will return a $k$-tuple that covers the entire graph with radius $\leq 2R - \alpha$.

    \paragraph{Runtime:} We define $\bar{\delta} = (\delta_0, \ldots, \delta_{\ell - 1})$ as follows. Define $t_{\ell} = k-\ell + 1$ and $t_{i} = \delta_i + t_{i+1}$ for any $i<\ell$.
    \begin{itemize}
        \item $\delta_{\ell - 1} = \frac{1}{k-\ell + 2}$.
        \item For $k-\ell + 1\leq i\leq k$ define $\delta_{k-i} = \frac{i - t_{k-i+1} + 1}{i+1}$.
    \end{itemize}

    \begin{claim}
        For all $k-\ell + 1\leq i\leq k$, $t_{k-i} = \frac{i}{2} + \frac{(k-\ell)(k-\ell + 1)}{2(i+1)} + 1$.
    \end{claim}
    \begin{proof}
        For $i = k-\ell$,
        \[
        t_{\ell} = k - \ell + 1 = \frac{k-\ell}{2} + \frac{(k-\ell)(k-\ell + 1)}{2(k -  \ell +1)} + 1.
        \]
        Therefore, by \cref{clm:generalrecursion} the claim holds for any $i\geq k-\ell$.
    \end{proof}

    \begin{cor}
        $t_0 = k - \ell + \frac{\ell(\ell+1)}{2(k+1)} + 1$.
    \end{cor}
    \begin{proof}
        \begin{align*}
            t_0 &= t_{k-k} = \frac{k}{2} + \frac{(k-\ell)(k-\ell + 1)}{2(k+1)} + 1 = \frac{k}{2} + \frac{(k-\ell)(k+1)}{2(k+1)} - \frac{(k-\ell)\ell}{2(k+1)} + 1 \\&=
            \frac{k}{2} + \frac{k-\ell}{2} -\frac{\ell(k+1)}{2(k+1)}+ \frac{\ell(\ell+1)}{2(k+1)} + 1 = k - \ell + \frac{\ell(\ell+1)}{2(k+1)} + 1.
        \end{align*}
    \end{proof}

    Denote by $\bar{T}_j$ the runtime of $\step(R, k, \ell, \bar{\delta}, j, (s_1, \ldots, s_j))$. We claim that $\bar{T}_j = \tilde{O}(n^{t_j})$. 

    For $j = \ell$, $\step(R, k, \ell, \bar{\delta}, (s_1, \ldots, s_\ell))$ checks for $\tilde{O}(n^{k-\ell})$ tuples if they cover a set of points. This takes $\tilde{O}(n^{k-\ell + 1}) = O(n^{t_\ell})$ time.

    For $j < \ell$, let $i = k-j > k-\ell$. The subroutine $\step(R, k, \bar{\delta}, j, (s_1, \ldots, s_j))$ checks for $\tilde{O}(n^{(1-\delta_{j})\cdot (k-j)})$ tuples if they cover a set of points. This takes time, $$\tilde{O}(n^{(1-\delta_{j})\cdot (k-j)+1}) = \tilde{O}(n^{t_{k-i+1} + i\delta_{k-i} + 1}) = \tilde{O}(n^{t_{k-i}}).$$ 

    Further, $\step(R, k, \ell, \bar{\delta}, k-i, (s_1, \ldots, s_{k-i}))$ performs $n^{\delta_{k-i}}$ calls to $\step(R, k, \bar{\delta}, k-i+1, (s_1, \ldots, s_{k-i+1}))$ and so

    \[
    \bar{T}_{k-i} = \tilde{O}(n^{t_{k-i}}) + \tilde{O}(n^{\delta_{k-i} + t_{k-i + 1}}) = \tilde{O}(n^{t_{k-i}}).
    \]

    The final runtime is comprised of running APSP in $O(mn)$ time and then running $\step(R, k, \ell,\bar{\delta}, 0, ())$ and so Algorithm \ref{alg:2lapprox} takes $O(mn) + \tilde{O}(n^{t_0}) = \tilde{O}(mn + n^{k-\ell +  \frac{\ell(\ell + 1)}{2(k+1)}+1})$ time.
\end{proof}


We note that the $mn$ component of the runtime comes from computing exact APSP. To avoid this runtime, we could instead compute an additive approximation of APSP (e.g.\ \cite{DorHZ00} or \cite{saha}) and obtain a faster runtime at the expense of a greater additive error.

\subsection{Approximating 3-Center}\label{sub:3center}

In this section we refine the techniques introduced in the approximation algorithms for $2$-center and general $k$-center for the case of 3 centers. By generalizing the ideas of \cref{thm:53approx2center} we obtain a nearly-$7/4$-approximation algorithm that works for graphs with integer edge weights bounded by a polynomial in $n$ and whose running time is always less than $mn$ for some values of $m$. 

We first adapt \cref{lm:step1gen} to weighted graphs with a slight change in the proof:

\begin{lemma}\label{lm:step1genw}
Given a graph with positive integer weights bounded by $M$, let $S$ be a random sample of size $O(n^{1-\delta}\log n)$. Let $w$ be the furthest node from $S$ and let $W$ be the closest $n^\delta$ nodes to $w$.

    For any real $r>0$, either there exist $k$ nodes $t_1, t_2, \ldots, t_k\in S$ that cover the graph with radius $R + r$, or w.h.p there exists some $s_1\in W$ such that $d(s_1, c_i) < R - r +M$ for some $i\in [k]$, w.l.o.g $i=1$.
\end{lemma}
\begin{proof}
    If for every center $d(c_i, S) \leq r$, denote the closest point to $c_i$ in $S$ by $t_i$. Now these $k$ points $t_1, t_2, \ldots, t_k$ cover the entire graph with radius $R + r$, as any point $v\in V$ has $d(v, c_i) \leq R$ for some $i\in [k]$, and thus $d(v, t_i) \leq R + r$.

    Otherwise, there exists a center such that $d(c_i, S) > r$ and therefore $d(w, S) > r$. W.h.p, the set $S$ hits every set of size $n^\delta$ and so $S$ hits $W$. Therefore, there exists a point $u\in W \cap S$ such that $d(w,u) > r$. By the definition of $W$, the set must contain all the points closer to $w$ than $u$ and so it contains all nodes of distance $\leq r$ from $w$.

    For some $i\in [k]$, $d(w,c_i) \leq R$. Consider the shortest path between $w$ and $c_i$. There are two consecutive nodes $s_1$ and $s'_1$ on the path so that $d(w, s_1) \leq r$ and $d(w, s'_1) >r$. Because the largest edge weight in $G$ is at most $M$, $d(w,s_1)> r-M$.
   Hence, $d(s_1, c_i) \leq R-d(w,s_1)< R-r+M$. By the above observation, we have that $s_1 \in W$.
\end{proof}

Now we slightly generalize \cref{lm:laststepgen} for weighted graphs.

    \begin{lemma}\label{lm:laststepgenw}
    Let $G$ be a given $n$-node graph with positive integer edge weights bounded by $M$.
        Let $r, \alpha$ be positive real numbers with $r+\alpha\leq R$, let $S$ be a random subset of vertices of size $O(n^{1-\delta}\log n)$ and $C =\{ s_1, \ldots, s_{k-1}\}$ be nodes such that $d(s_i, c_i) \leq r$ for every $i < k$. Define $U = B(C, R + r)^c$ and $Y = B(C, 2R - \alpha)^c$.
        Let $w$ be the furthest point in $U$ from $S$ and let $W$ be the closest $n^\delta$ vertices to $w$. One of the following two holds.
        \begin{enumerate}
            \item $\exists s_k \in W$ such that $d(s_k, c_k) < r + \alpha+M$.
            \item $Q\coloneqq \bigcap_{t\in S \cap U} B(t, R)\neq \emptyset$, and any point $s'_k \in Q$ covers all points in $Y$ with radius $2R - \alpha$.
        \end{enumerate}
    \end{lemma}

    \begin{proof}
        If $d(w, S) > R - (r +\alpha)$, then
        by our argument from the previous lemma,
        we have that $\exists s_k \in W$ such that $d(s_k, c_k) < r + \alpha+M$. 
        Otherwise, any point $u\in U$ has $d(u, S) \leq R - (r+\alpha)$. 

        Assume that the latter happens.
        Any point $t\in S \cap U$ is not covered by $c_1, \ldots, c_{k-1}$ and must be covered by $c_k$. Thus, $c_k\in B(t, R)$ for any $t\in S \cap U$ and $c_k \in Q$, meaning $Q\neq \emptyset$. 
        
        Let $s'_k$ be an arbitrary vertex in $Q$ and let $y\in Y$. 
        Let $t\in S$ be the closest vertex to $y$ in $S$, by the above observation, since 
        $Y \subseteq U$,
        $d(t,y) \leq R - (r+\alpha) $. 
        
        If $t\notin U$, then there exists some $s_i$ such that $d(s_i, t) \leq R + r$ and hence $d(s_i, y) \leq R + r + R - (r+\alpha) = 2R - \alpha$. However, this would imply $y\notin Y$ and so we conclude $t\in U$. 
        
        Therefore, $d(s'_k, t) \leq R$ and so $d(s'_k, y) \leq R + R - (r + \alpha) \leq 2R - \alpha$.
    \end{proof}

Let $\mu_0=\omega-2+\frac{1}{\omega+1}$ and $\mu_1=\frac{3\omega^2 - 3\omega - 1}{4\omega + 1}$.
For the current bound on $\omega$, $\mu_0$ is about $0.6682$ and $\mu_1$ is about $0.835$.
For any $\omega\in [2,3]$ we have $\mu_0<\mu_1<2(\omega-1)/3$.



\begin{theorem}\label{thm:74approx3center}
    Given an integer $R$ and an undirected graph $G$ with $n$ vertices and $m=n^{1+\mu}$ edges. Suppose $G$ has positive integer edge weights bounded by $M=\textrm{poly}(n)$.
    For $\mu \in [\mu_0,\mu_1]$,
    there exists an algorithm running in time
    \[\tilde{O}\left(n^{\frac{(5\omega-\omega^2+1)-\mu\omega(\omega-2)}{(3\omega-\omega^2+1)}}\right)\]
    time that w.h.p.\ computes a 
    $(7/4,M)$-approximation to $R_3(G)$.
    For all $\mu\in [\mu_0,\mu_1]$, the running time of the algorithm is polynomially faster than both $\tilde{O}(mn)$ and $\tilde{O}(n^{\omega+1/(\omega+1)})$.
\end{theorem}

For the current bound on $\omega$ we get that the running time is $O(n^{2.90466 - 0.35388\mu})$ for $\mu \in [0.6682,0.835]$ and in this interval it is always better than $\tilde{O}(mn)$, the runtime of exact APSP.


\begin{proof}
We will give an algorithm that w.h.p. returns an estimate of at most $7R/4+M$ whenever $R_3(G) \leq R$. The algorithm can detect when $R_3(G) >7R/4+M$.
The statement of the theorem follows from binary searching on $R$ (there is a polynomial range).

    Assume $R_3(G) \leq R$ and let $c_1, c_2, c_3$ be an optimal $3$-center with radius $\leq R$, in this case our algorithm will find three points which cover all vertices with radius $\leq 7R/4+M$. 
    Let $S_1, S_2, S_3$ be random samples of vertices in $V$ of sizes $O(n^{1-\delta}\log n), O(n^{1-\gamma}\log n), O(n^{1-\beta}\log n)$ respectively, for parameters $\beta, \gamma, \delta$ we will set later. In 
    

     \[ {\Tilde{O}(m(n^{1-\delta} + n^{1-\gamma} + n^{1-\beta}))} \]
      
    time we compute all the distances between $s\in S_1 \cup S_2 \cup S_3$ and $v\in V$ using Dijkstra's algorithm. 

    Let $w_1$ be the furthest node from the set $S_1$, let $W_1$ be the closest $n^\delta$ nodes to $w_1$. In time $\Tilde{O}(mn^{\delta})$ run Dijkstra's from every node in $W_1$. By \cref{lm:step1genw}, either there exist 3 points $t_1, t_2, t_3 \in S_1$ that cover the graph with radius $R+x$
    or there exists some $s_1\in W_1$ such that $d(s_1, c_1)\leq R-x+M$, for a choice of $x$ (we will later see that $x=3R/4$ is optimal), where $s_1$ covers all nodes that $c_1$ covers, within $2R-x+M$.


    To handle the first case, construct a $|S_1|\times |V|$ binary matrix $X$ with rows indexed by $S_1$ and columns indexed by $V$ such that $X[s,v] = 0$ iff $d(s,v) \leq R+x$
    and a $|V|\times |S_1|^2$ binary matrix $Y$ such that $Y[v, (s_1, s_2)] = 1$ iff both $d(v, s_1) > R+x$
    and $d(v, s_2)> R+x$.
    Multiply $X$ by $Y$ in time asymptotically
    
    $$n^{\omega(1 - \delta, 1, 2 - 2\delta)}\leq n^{\delta+(1-\delta)+\omega(1 - \delta)}= n^{1+\omega(1-\delta)}.$$

    The nodes $t_1, t_2, t_3$ will give $XY[t_1, (t_2, t_3)] = 0$ and any nodes that satisfy $XY[a, (b,c)] = 0$ cover all vertices with radius $R+x$.

    We are now left to handle the second case, where $\exists s_1\in W_1$ with $d(s_1, c_1)\leq R-x+M$.
    We'll have an $n^\delta$ overhead to the runtime and assume we have $s_1\in W_1$ such that $d(s_1, c_1)\leq R-x+M$.
    Define 
    $U_1(s_1) = \{u\in V : d(u, s_1) > R+x+M\}$. 
    Note that by our choice of $x=3R/4$, all nodes in $U_1(s_1)$ have more than $(R+x+M) - d(s_1,c_1) \ge 2x > R$ distance from $c_1$, and hence must be covered by either $c_2$ or $c_3$.
    
    Let $w_2(s_1)$ be the furthest node in $U_1(s_1)$ from $S_2$ and let $W_2(s_1)$ be the closest $n^\gamma$ nodes to $w_2(s_1)$. By the same argument as seen in the proof of \cref{lm:laststepgenw}, either every point in $U_1(s_1)$ is within distance $R-(2R-2x)=2x-R$
    of $S_2$ or there exists a point $s_2\in W_2(s_1)$ such that $d(s_2, c_2) \leq (2R-2x)+M$.

    \paragraph{Case I} $\forall u\in U_1(s_1)~ \exists t\in S_2$ such that 
    $d(u, t) \leq 2x-R$. 

    Construct a $|W_1\times V|\times |S_2|$ binary matrix $\hat{X}$ where $\hat{X}[(s, u_1),t] = 1$ iff both
    $d(s,t) > 2R-x+M$ and $d(u_1,t)> R$,
    and a $|S_2|\times |V|$ binary matrix $\hat{Y}$ where $\hat{Y}[t,  v] = 1\iff d(t,v)> R$. 
    
    In time asymptotically
    
    $$n^{\omega(1 + \delta, 1-\gamma, 1)}\leq n^{\delta+2\gamma+\omega(1-\gamma)}=n^{\delta+\omega-\gamma(\omega-2)}$$ 
    compute the product $\hat{X}\hat{Y}$. 

    We have that $d(s_1,c_1)\leq R-x+M$ so for all $t$ such that $d(c_1,t)\leq R$, $d(s_1,t)\leq 2R-x+M$.
    Hence any point $t\in S_2$ with
    $d(s_1,t)>2R-x+M$
    must be covered by $c_2$ or $c_3$.
    We thus have $\hat{X}\hat{Y}[(s_1,c_2), c_3] = 0$. 
    
    Furthermore, we claim that any $s_2, s_3$ with $\hat{X}\hat{Y}[(s_1,s_2), s_3] = 0$ will cover all vertices with radius 
$\leq R+x+M$.
    
    Any point $u\in V$ not covered by $s_1$ within distance
    $R+x+M$
    is in $U_1(s_1)$ 
    and so by our assumption for case I there exists a point $t\in S_2$ with 
    $d(u, t) \leq (2x-R)$.
    Therefore, 
    $d(s_1, t) > (R+x+M)-(2x-R)=2R-x+M$
    and since $\hat{X}\hat{Y}[(s_1,s_2), s_3] = 0$ we must have either $d(t, s_2) \le  R$ or $d(t, s_3) \le R$. We conclude 
    $d(u,s_i)\leq R+x+M$
    for $i = 2$ or $i = 3$.
    
    \paragraph{Case II} $\exists s_2 \in W_2(s_1)$ such that $d(s_2, c_2) \leq (2R-2x)+M $.
    This node $s_2$ covers all nodes that $c_2$ covers within radius $3R-2x+M$.

    Define $U_2(s_1, s_2) = \{u\in V: d(u,s_1) > 3R-2x \land d(u, s_2)> 3R-2x\}$.


       
    Let $w_3(s_1, s_2)$ be the furthest in $U_2(s_1, s_2)$ from $S_3$ and define $W_3(s_1, s_2)$ to be the closest $n^{\beta}$ points to $w_3(s_1, s_2)$. 
    
    We use \cref{lm:laststepgenw} with parameters 
    $\alpha=R-x$, $r=2R-2x$, $S = S_3, C = \{s_1, s_2\}$, where we will pick $x=3R/4\geq 2R/3$ so that the conditions of the lemma apply.
   Either $\exists s_3\in W_3(s_1, s_2)$ such that 
   $d(s_3, c_3) \leq 3R-3x+M$
   or for any point $s'_3\in Q \coloneqq \bigcap _{t\in S_3 \cap U_2(s_1, s_2)} (t, R)$ we have that $s_1, s_2, s'_3$ cover all vertices with radius $R+x$.

    To handle the first case, in total time 
    $$O(mn^{\beta + \gamma + \delta})$$ run Dijkstra's from every point $s_3\in W_3(s_1, s_2)$ and check if $s_1, s_2, s_3$ cover all points in $V$ with radius
    $R+(3R-3x+M)= 4R-3x+M$.

    To handle the second case, construct a $|W_2(s_1)|\times |S_3|$ binary matrix $\tilde{X}$ where
    $\tilde{X}[s, t] = 1 \iff d(s,t) > 3R-2x$ 
    and a $|S_3|\times |V|$ binary matrix $\tilde{Y}$ where $\tilde{Y}[t,u] = 1 \iff d(t,u)> R$. 
    
    In time asymptotically 
    $$n^{\delta + \omega(\gamma, 1 - \beta, 1)}\leq n^{\delta +(1-\beta-\gamma)+(1-\gamma)+\omega \gamma}=n^{\delta +2-\beta+(\omega-2)\gamma}$$
    compute the product $\tilde{X}\tilde{Y}$ for every $s_1\in W_1$. The inequality holds because we will ensure that $\gamma+\beta<1$ and hence $\gamma<1$ and $1-\beta<1$. By the same argument we have seen, for the correct $s_1$ we have $\tilde{X}\tilde{Y}[s_2, c_3] = 0$ and for any $s_3$ such that $\tilde{X}\tilde{Y}[s_2, s_3] = 0$ the triple $s_1, s_2, s_3$ will cover the graph with radius $R+x$.
    
    Thus, after computing the matrix product we need to check for every $s_2\in W_2(s_1)$ at most one potential center. This takes time ${O}(n^{\delta + \gamma + 1})$. 

    Let's see what $x$ should be set to. In the worst case the approximate radius values that we get are $4R-3x+M$, $2R-x+M$, $3R-2x+M$ and $R+x+M$.
    
If we set $x=3R/4$, we get that $4R-3x=R+x=7R/4$ and
since
$R-x\geq 0$ we have that 
%
$2R-x\leq 3R-2x\leq 4R-3x$. Thus the final approximate radius is at most $7R/4+M$.
    


    \paragraph{Final runtime:} Let $m = n^{1 + \mu}$, the runtime exponents from above are:
    \begin{align*}
        & 2+\mu-\delta\\
         & 2+\mu-\beta\\
          & 2+\mu-\gamma\\
          & 1+\omega(1-\delta)\\
          & \delta+\omega-\gamma(\omega-2)\\
          & 1+\mu +\beta+\gamma+\delta\\
          &\delta+2-\beta+\gamma(\omega-2).
    \end{align*}

We will assume that 
\[\gamma\leq\beta, \gamma\leq\delta, \gamma+\delta+\beta\leq 1,\]
which allows us to also conclude that
$\gamma\leq \frac{1}{3}<\frac{1}{2}+\frac{\beta}{2(\omega-2)}$ and so $2-\beta+\gamma(\omega-2)<\omega-\gamma(\omega-2)$.
This allows us to assume that several exponents above are dominated and what remains is:
    \begin{align*}
          & 2+\mu-\gamma\\
          & 1+\omega(1-\delta)\\
          & \delta+\omega-\gamma(\omega-2)\\
          & 1+\mu +\beta+\gamma+\delta .
    \end{align*}

Let's set
    \begin{align*}
          & \beta = \frac{2(\omega-1)-3\mu}{3-\omega+1/\omega}\\
          & \gamma = \frac{1}{3}-\frac{\beta(\omega+1)}{3\omega}\\
          & \delta = \frac{1}{3}-\frac{\beta(\omega-2)}{3\omega}.
    \end{align*}

We want to show that $\beta,\gamma,\delta\geq 0,$ and $\gamma\leq\beta, \gamma\leq\delta, \gamma+\delta+\beta\leq 1$.

As we assumed that $\mu<2(\omega-1)/3$ we gave that $\beta\geq 0$.

As we assumed that $\mu\geq \mu_0=\omega-2+\frac{1}{\omega+1}$, we get that 
\[\gamma = \frac{1}{3}-\frac{(\omega+1)}{3\omega} \cdot \frac{2(\omega-1)-3\mu}{3-\omega+1/\omega} =\]
\[\frac{1}{3}-\frac{2(\omega^2-1)-3\mu(\omega+1)}{3(3\omega-\omega^2+1)} =\]
\[\frac{(3\omega-\omega^2+1)-2(\omega^2-1)+3\mu(\omega+1)}{3(3\omega-\omega^2+1)} =\]
\[\frac{(\omega-\omega^2+1)+\mu(\omega+1)}{(3\omega-\omega^2+1)} \geq 0\]

By construction, since $\omega+1>\omega-2$, $\delta>\gamma$.

Since we assumed that $\mu\leq \mu_1=\frac{3\omega^2-3\omega-1}{4\omega+1}$, we also have that

\[\beta = \frac{2(\omega-1)-3\mu}{3-\omega+1/\omega}\geq\]
\[\frac{2(\omega-1)(4\omega+1)-3(3\omega^2-3\omega-1)}{(3-\omega+1/\omega)(4\omega+1)}=\]
\[\frac{8\omega^2-6\omega-2-9\omega^2+9\omega+3)}{(3-\omega+1/\omega)(4\omega+1)}=\]
\[\frac{\omega(1-\omega^2+3\omega)}{(3\omega-\omega^2+1)(4\omega+1)}=\frac{\omega}{4\omega+1}.\]

Hence:
\[\beta-\gamma=\beta-\frac{1}{3}+\frac{\beta(\omega+1)}{3\omega}=\]
\[\frac{\beta(4\omega+1)-\omega}{3\omega}\geq 0,\]
and thus $\gamma\leq \beta$, as promised.

Let's consider
\[\gamma+\delta+\beta = \frac{2}{3}+\beta\cdot \frac{3\omega-(\omega+1)-(\omega-2)}{3\omega}=\frac{2}{3}+\beta\cdot \frac{\omega+1}{3\omega}.\]

The above is $\leq 1$ iff $\beta\leq\frac{\omega}{\omega+1}$.
Since we have that $\mu\geq\mu_0=\frac{\omega^2-\omega-1}{\omega+1}$, we get that
\[\beta \leq \frac{\omega(2(\omega^2-1)-3(\omega^2-\omega-1))}{(3\omega-\omega^2+1)(\omega+1)}=\frac{\omega}{\omega+1},\]
and thus $\gamma+\delta+\beta<1$, as promised.

Now let's look at our running time exponents.


First we see that
\[1-\beta-\delta = \frac{2}{3}+\beta{\omega-2-3\omega}{3\omega}=2(\frac{1}{3}-\frac{\beta(\omega+1)}{3\omega}=2\gamma,\]
and so we have that two of the exponents are equal 
$$2+\mu-\gamma=1+\mu+\beta+\gamma+\delta.$$

Next we see that
\[\delta (\omega+1)-\gamma(\omega-2) = \frac{\omega+1}{3}-\frac{\omega-2}{3}-\frac{\beta(\omega-2)(\omega+1)}{3\omega}+\frac{\beta(\omega-2)(\omega+1)}{3\omega}=1,\]

And so another pair of exponents is equal:
$$1+\omega-\omega\delta = \omega+\delta-\gamma(\omega-2).$$

Now consider the implications:
$$\frac{\omega(2(\omega-1))}{3\omega}= \frac{2(\omega-1)}{3}$$
is true and implies
$$\frac{\omega(2(\omega-1)-3\mu)+3\omega\mu}{3\omega}= \frac{2(\omega-1)}{3}$$
which by our setting of $\beta$ implies
$$\mu+\frac{\beta(1-\omega^2+3\omega)}{3\omega}= \frac{2(\omega-1)}{3}$$
which implies 
$$\mu+\frac{\beta(\omega+1)}{3\omega}= \frac{2\omega-2}{3}+\frac{\beta(\omega-2)}{3}$$
which implies 
$$1+\mu-\frac{1}{3}+\frac{\beta(\omega+1)}{3\omega}= \omega-\frac{\omega}{3}+\frac{\beta(\omega-2)}{3}$$
and thus 
$$2+\mu-\gamma= 1+\omega(1-\delta).$$

Thus all four exponents are the same and they equal

\[2+\mu-\gamma=2+\mu-\frac{\omega-\omega^2+1+\mu(\omega+1)}{(3\omega-\omega^2+1)}=\]
\[\frac{(6\omega-2\omega^2+2)+(3\omega-\omega^2+1)\mu-\omega+\omega^2-1-\mu(\omega+1)}{(3\omega-\omega^2+1)}=\]
\[\frac{(5\omega-\omega^2+1)-\mu\omega(\omega-2)}{(3\omega-\omega^2+1)}.\]

\end{proof}

%% file: main.bbl
\newcommand{\etalchar}[1]{$^{#1}$}
\begin{thebibliography}{ACLM23}

\bibitem[ACIM99]{AingworthCIM99}
Donald Aingworth, Chandra Chekuri, Piotr Indyk, and Rajeev Motwani.
\newblock Fast estimation of diameter and shortest paths (without matrix multiplication).
\newblock {\em {SIAM} J. Comput.}, 28(4):1167--1181, 1999.

\bibitem[ACLM23]{AbboudCLM23}
Amir Abboud, Vincent Cohen{-}Addad, Euiwoong Lee, and Pasin Manurangsi.
\newblock On the fine-grained complexity of approximating \emph{k}-center in sparse graphs.
\newblock In {\em 2023 Symposium on Simplicity in Algorithms, {SOSA} 2023, Florence, Italy, January 23-25, 2023}, pages 145--155. {SIAM}, 2023.

\bibitem[AGV23]{AbboudGW23}
Amir Abboud, Fabrizio Grandoni, and Virginia {Vassilevska Williams}.
\newblock Subcubic equivalences between graph centrality problems, apsp, and diameter.
\newblock {\em {ACM} Trans. Algorithms}, 19(1):3:1--3:30, 2023.

\bibitem[AVW16]{AbboudWW16}
Amir Abboud, Virginia {Vassilevska Williams}, and Joshua~R. Wang.
\newblock Approximation and fixed parameter subquadratic algorithms for radius and diameter in sparse graphs.
\newblock In {\em Proceedings of the Twenty-Seventh Annual {ACM-SIAM} Symposium on Discrete Algorithms, {SODA} 2016, Arlington, VA, USA, January 10-12, 2016}, pages 377--391. {SIAM}, 2016.

\bibitem[BHW20]{BalcanHW20}
Maria{-}Florina Balcan, Nika Haghtalab, and Colin White.
\newblock \emph{k}-center clustering under perturbation resilience.
\newblock {\em {ACM} Trans. Algorithms}, 16(2):22:1--22:39, 2020.
\newblock \href {https://doi.org/10.1145/3381424} {\path{doi:10.1145/3381424}}.

\bibitem[BRS{\etalchar{+}}21]{BackursRSWW21}
Arturs Backurs, Liam Roditty, Gilad Segal, Virginia {Vassilevska Williams}, and Nicole Wein.
\newblock Toward tight approximation bounds for graph diameter and eccentricities.
\newblock {\em {SIAM} J. Comput.}, 50(4):1155--1199, 2021.

\bibitem[CFG{\etalchar{+}}24]{cruciani2024dynamic}
Emilio Cruciani, Sebastian Forster, Gramoz Goranci, Yasamin Nazari, and Antonis Skarlatos.
\newblock Dynamic algorithms for k-center on graphs.
\newblock In {\em Proceedings of the 2024 Annual ACM-SIAM Symposium on Discrete Algorithms (SODA)}, pages 3441--3462, 2024.
\newblock \href {https://doi.org/10.1137/1.9781611977912.123} {\path{doi:10.1137/1.9781611977912.123}}.

\bibitem[CGR16]{CairoGR16}
Massimo Cairo, Roberto Grossi, and Romeo Rizzi.
\newblock New bounds for approximating extremal distances in undirected graphs.
\newblock In {\em Proceedings of the Twenty-Seventh Annual {ACM-SIAM} Symposium on Discrete Algorithms, {SODA} 2016, Arlington, VA, USA, January 10-12, 2016}, pages 363--376. {SIAM}, 2016.

\bibitem[CIP10]{cip10}
Chris Calabro, Russell Impagliazzo, and Ramamohan Paturi.
\newblock On the exact complexity of evaluating quantified \emph{k}-cnf.
\newblock In {\em Parameterized and Exact Computation - 5th International Symposium, {IPEC} 2010, Chennai, India, December 13-15, 2010. Proceedings}, volume 6478 of {\em Lecture Notes in Computer Science}, pages 50--59. Springer, 2010.

\bibitem[CLR{\etalchar{+}}14]{ChechikLRSTW14}
Shiri Chechik, Daniel~H. Larkin, Liam Roditty, Grant Schoenebeck, Robert~Endre Tarjan, and Virginia {Vassilevska Williams}.
\newblock Better approximation algorithms for the graph diameter.
\newblock In {\em Proceedings of the Twenty-Fifth Annual {ACM-SIAM} Symposium on Discrete Algorithms, {SODA} 2014, Portland, Oregon, USA, January 5-7, 2014}, pages 1041--1052. {SIAM}, 2014.

\bibitem[DF85]{DyerF85}
Martin~E. Dyer and Alan~M. Frieze.
\newblock A simple heuristic for the p-centre problem.
\newblock {\em Op. Res. Lett.}, 3(6):285–--288, 1985.

\bibitem[DFHT03]{planar2}
Erik~D Demaine, Fedor~V Fomin, Mohammad~Taghi Hajiaghayi, and Dimitrios~M Thilikos.
\newblock Fixed-parameter algorithms for the (k, r)-center in planar graphs and map graphs.
\newblock In {\em International Colloquium on Automata, Languages, and Programming}, pages 829--844. Springer, 2003.

\bibitem[DHZ00]{DorHZ00}
Dorit Dor, Shay Halperin, and Uri Zwick.
\newblock All-pairs almost shortest paths.
\newblock {\em SIAM Journal on Computing}, 29(5):1740--1759, 2000.
\newblock \href {https://doi.org/10.1137/S0097539797327908} {\path{doi:10.1137/S0097539797327908}}.

\bibitem[EKM14]{planar1}
David Eisenstat, Philip~N Klein, and Claire Mathieu.
\newblock Approximating k-center in planar graphs.
\newblock In {\em Proceedings of the Twenty-Fifth Annual ACM-SIAM Symposium on Discrete Algorithms}, pages 617--627. SIAM, 2014.

\bibitem[Elk05]{elkin}
Michael Elkin.
\newblock Computing almost shortest paths.
\newblock {\em ACM Trans. Algorithms}, 1(2):283–323, oct 2005.

\bibitem[Fel19a]{Feldmann19}
Andreas~Emil Feldmann.
\newblock Fixed-parameter approximations for k-center problems in low highway dimension graphs.
\newblock {\em Algorithmica}, 81(3):1031--1052, 2019.
\newblock URL: \url{https://doi.org/10.1007/s00453-018-0455-0}, \href {https://doi.org/10.1007/S00453-018-0455-0} {\path{doi:10.1007/S00453-018-0455-0}}.

\bibitem[Fel19b]{feldmann2019fixed}
Andreas~Emil Feldmann.
\newblock Fixed-parameter approximations for k-center problems in low highway dimension graphs.
\newblock {\em Algorithmica}, 81:1031--1052, 2019.

\bibitem[FG88]{FederG88}
Tom{\'{a}}s Feder and Daniel~H. Greene.
\newblock Optimal algorithms for approximate clustering.
\newblock In {\em Proceedings of the 20th Annual {ACM} Symposium on Theory of Computing, May 2-4, 1988, Chicago, Illinois, {USA}}, pages 434--444. {ACM}, 1988.
\newblock \href {https://doi.org/10.1145/62212.62255} {\path{doi:10.1145/62212.62255}}.

\bibitem[FM20]{feldmann2020parameterized}
Andreas~Emil Feldmann and D{\'a}niel Marx.
\newblock The parameterized hardness of the k-center problem in transportation networks.
\newblock {\em Algorithmica}, 82:1989--2005, 2020.

\bibitem[GG24]{gan2024fully}
Jinxiang Gan and Mordecai~J Golin.
\newblock Fully dynamic k-center in low dimensions via approximate furthest neighbors.
\newblock In {\em 2024 Symposium on Simplicity in Algorithms (SOSA)}, pages 269--278. SIAM, 2024.

\bibitem[GHL{\etalchar{+}}21]{goranci2021fully}
Gramoz Goranci, Monika Henzinger, Dariusz Leniowski, Christian Schulz, and Alexander Svozil.
\newblock Fully dynamic k-center clustering in low dimensional metrics.
\newblock In {\em 2021 Proceedings of the Workshop on Algorithm Engineering and Experiments (ALENEX)}, pages 143--153. SIAM, 2021.

\bibitem[Gon85]{Gonzalez85}
Teofilo~F. Gonzalez.
\newblock Clustering to minimize the maximum intercluster distance.
\newblock {\em Theor. Comput. Sci.}, 38:293--306, 1985.
\newblock \href {https://doi.org/10.1016/0304-3975(85)90224-5} {\path{doi:10.1016/0304-3975(85)90224-5}}.

\bibitem[HN79]{HsuN79}
Wen{-}Lian Hsu and George~L. Nemhauser.
\newblock Easy and hard bottleneck location problems.
\newblock {\em Discret. Appl. Math.}, 1(3):209--215, 1979.

\bibitem[HS86]{HochbaumS86}
Dorit~S. Hochbaum and David~B. Shmoys.
\newblock A unified approach to approximation algorithms for bottleneck problems.
\newblock {\em J. {ACM}}, 33(3):533--550, 1986.
\newblock \href {https://doi.org/10.1145/5925.5933} {\path{doi:10.1145/5925.5933}}.

\bibitem[IP01]{ip2}
Russell Impagliazzo and Ramamohan Paturi.
\newblock On the complexity of k-sat.
\newblock {\em Journal of Computer and System Sciences}, 62(2):367--375, 2001.
\newblock URL: \url{https://www.sciencedirect.com/science/article/pii/S0022000000917276}, \href {https://doi.org/10.1006/jcss.2000.1727} {\path{doi:10.1006/jcss.2000.1727}}.

\bibitem[KLM19]{SLM19}
{Karthik {C. S.}}, Bundit Laekhanukit, and Pasin Manurangsi.
\newblock On the parameterized complexity of approximating dominating set.
\newblock {\em J. {ACM}}, 66(5):33:1--33:38, 2019.
\newblock \href {https://doi.org/10.1145/3325116} {\path{doi:10.1145/3325116}}.

\bibitem[KLP19]{KatsikarelisLP19}
Ioannis Katsikarelis, Michael Lampis, and Vangelis~Th. Paschos.
\newblock Structural parameters, tight bounds, and approximation for (k, r)-center.
\newblock {\em Discret. Appl. Math.}, 264:90--117, 2019.
\newblock URL: \url{https://doi.org/10.1016/j.dam.2018.11.002}, \href {https://doi.org/10.1016/J.DAM.2018.11.002} {\path{doi:10.1016/J.DAM.2018.11.002}}.

\bibitem[Lin19]{Lin19}
Bingkai Lin.
\newblock A simple gap-producing reduction for the parameterized set cover problem.
\newblock In {\em 46th International Colloquium on Automata, Languages, and Programming, {ICALP} 2019, July 9-12, 2019, Patras, Greece}, volume 132 of {\em LIPIcs}, pages 81:1--81:15. Schloss Dagstuhl - Leibniz-Zentrum f{\"{u}}r Informatik, 2019.
\newblock URL: \url{https://doi.org/10.4230/LIPIcs.ICALP.2019.81}, \href {https://doi.org/10.4230/LIPICS.ICALP.2019.81} {\path{doi:10.4230/LIPICS.ICALP.2019.81}}.

\bibitem[LVW18]{LincolnWW18}
Andrea Lincoln, Virginia {Vassilevska Williams}, and R.~Ryan Williams.
\newblock Tight hardness for shortest cycles and paths in sparse graphs.
\newblock In {\em Proceedings of the Twenty-Ninth Annual {ACM-SIAM} Symposium on Discrete Algorithms, {SODA} 2018, New Orleans, LA, USA, January 7-10, 2018}, pages 1236--1252. {SIAM}, 2018.

\bibitem[PW10]{PatrascuW10}
Mihai P{\u{a}}tra{\c{s}}cu and Ryan Williams.
\newblock On the possibility of faster {SAT} algorithms.
\newblock In {\em Proceedings of the Twenty-First Annual {ACM-SIAM} Symposium on Discrete Algorithms, {SODA} 2010, Austin, Texas, USA, January 17-19, 2010}, pages 1065--1075. {SIAM}, 2010.
\newblock \href {https://doi.org/10.1137/1.9781611973075.86} {\path{doi:10.1137/1.9781611973075.86}}.

\bibitem[RV13]{RodittyW13}
Liam Roditty and Virginia {Vassilevska Williams}.
\newblock Fast approximation algorithms for the diameter and radius of sparse graphs.
\newblock In {\em Symposium on Theory of Computing Conference, STOC'13, Palo Alto, CA, USA, June 1-4, 2013}, pages 515--524. {ACM}, 2013.

\bibitem[Sei95]{seidel}
R.~Seidel.
\newblock On the all-pairs-shortest-path problem in unweighted undirected graphs.
\newblock {\em Journal of Computer and System Sciences}, 51(3):400--403, 1995.
\newblock \href {https://doi.org/10.1006/jcss.1995.1078} {\path{doi:10.1006/jcss.1995.1078}}.

\bibitem[SY24]{saha}
Barna Saha and Christopher Ye.
\newblock Faster approximate all pairs shortest paths.
\newblock In {\em Proceedings of the 2024 Annual ACM-SIAM Symposium on Discrete Algorithms (SODA)}, pages 4758--4827, 2024.
\newblock \href {https://doi.org/10.1137/1.9781611977912.170} {\path{doi:10.1137/1.9781611977912.170}}.

\bibitem[TFL83]{kcentersurvey}
Barbaros~C. Tansel, Richard~L. Francis, and Timothy~J. Lowe.
\newblock Location on networks: A survey. part i: The p-center and p-median problems.
\newblock {\em Management Science}, 29(4):482–--497, 1983.

\bibitem[Tho04]{Thorupcenter04}
Mikkel Thorup.
\newblock Quick k-median, k-center, and facility location for sparse graphs.
\newblock {\em {SIAM} J. Comput.}, 34(2):405--432, 2004.

\bibitem[{Vas}18]{vsurvey}
Virginia {Vassilevska Williams}.
\newblock On some fine-grained questions in algorithms and complexity.
\newblock In {\em Proceedings of the ICM}, volume~3, pages 3431--3472. World Scientific, 2018.

\bibitem[VXXZ24]{VXXZ24}
Virginia {Vassilevska Williams}, Yinzhan Xu, Zixuan Xu, and Renfei Zhou.
\newblock New bounds for matrix multiplication: from alpha to omega.
\newblock In {\em Proceedings of the 2024 Annual ACM-SIAM Symposium on Discrete Algorithms (SODA)}, pages 3792--3835, 2024.
\newblock \href {https://doi.org/10.1137/1.9781611977912.134} {\path{doi:10.1137/1.9781611977912.134}}.

\bibitem[Wil05]{Williams05}
Ryan Williams.
\newblock A new algorithm for optimal 2-constraint satisfaction and its implications.
\newblock {\em Theor. Comput. Sci.}, 348(2-3):357--365, 2005.
\newblock URL: \url{https://doi.org/10.1016/j.tcs.2005.09.023}, \href {https://doi.org/10.1016/J.TCS.2005.09.023} {\path{doi:10.1016/J.TCS.2005.09.023}}.

\end{thebibliography}
